\documentclass[11pt]{article}
\usepackage{fullpage}
\usepackage{url}
\usepackage[colorlinks=true,linkcolor=Emerald,citecolor=RoyalBlue]{hyperref}
\usepackage{amsmath,nccmath, amsfonts,amsthm,amssymb,multirow,mdframed}

\usepackage{times}
\usepackage{paralist}
\usepackage{graphicx}
\usepackage{floatpag}
\usepackage{bbold}
\usepackage{enumitem}
\usepackage{subcaption} 
\usepackage{soul}
\usepackage{comment}
\usepackage{calc}
\usepackage{upgreek}
\usepackage{mathtools}
\usepackage{color}
\usepackage[dvipsnames]{xcolor}
\usepackage{xspace}
\usepackage{tcolorbox}
\allowdisplaybreaks

\DeclareMathOperator*{\E}{\mathbb{E}}

\DeclareMathOperator*{\Ind}{\mathbb{1}}

\DeclarePairedDelimiter\norm{\lVert}{\rVert}
\DeclarePairedDelimiter\ip{\langle}{\rangle}
\newcommand\skipi{{\vskip 10pt}}
\def \F {{\mathbb F}}

\def \N {\mathbb{N}}
\def \S {S}

\def \cB {\mathcal{B}}

\def \cD {\mathcal{D}}

\def \cS {\mathcal{S}}
\def \cT {\mathcal{T}}

\def \Gr {\text{Gr}}
\def \GL {\text{GL}}
\def \incons {\text{incons}}

\def \poly {\mathrm{poly}}

\def \eps {{\varepsilon}}

\def \val {\mathrm{val}}
\def \viol {\mathrm{viol}}

\def\ggg{\gtrsim}
\def\lll{\lesssim}

\def \supp {{\sf supp}}
\def \viol {\text{viol}}
\def \bad {\text{Bad}}
\def \good {\text{Good}}
\def \sp {\text{SP}}
\def \spn {\text{span}}
\def \symp {\text{Symp}}
\def \diam {\text{diam}}

\renewcommand{\leq}{\leqslant}
\renewcommand{\le}{\leqslant}
\renewcommand{\geq}{\geqslant}

\newcommand{\wt}[1]{\widetilde{#1}}

\newcommand{\card}[1]{\left|#1\right|}

\newcommand{\dor}[1]{#1}
\newcommand{\mitali}[1]{#1}

\newtheorem{theorem}{Theorem}[section]

\newtheorem{fact}[theorem]{Fact}

\newcounter{Assumption}
\newtheorem{assumption}[Assumption]{Assumption}

\newtheorem{lemma}[theorem]{Lemma}
\newtheorem*{lemma*}{Lemma}
\newtheorem*{theorem*}{Theorem}
\newtheorem{claim}[theorem]{Claim}
\newtheorem*{claim*}{Claim}

\newtheorem{remark}[theorem]{Remark}
\newtheorem{definition}[theorem]{Definition}
\theoremstyle{definition}

\newtheorem{agr-test}{Agreement-Test}
\newtheorem{algo}{Algorithm}
\newtheorem{list-agr-test}{List-Agreement-Test}

\makeatletter
\let\c@fconjecture\c@conjecture
\makeatother

\makeatletter
\let\c@fconj\c@conj
\makeatother

\title{Constant Degree Direct Product Testers with Small Soundness}
\author{Mitali Bafna
\thanks{Department of Mathematics, Massachusetts Institute of Technology. \texttt{mitali.bafna@gmail.com}}
\and Noam Lifshitz
\thanks{Einstein institute of Mathematics, Hebrew University. \texttt{noamlifshitz@gmail.com}. Supported by the Israel Science Foundation (grant no.~1980/22).}
\and Dor Minzer
\thanks{Department of Mathematics, Massachusetts Institute of Technology. \texttt{dminzer@mit.edu}. Supported by NSF CCF award 2227876 and NSF CAREER award 2239160.}
}
\date{\vspace{-5ex}}
\begin{document}
\maketitle
\begin{abstract}
Let $X$ be a $d$-dimensional simplicial complex. A function $F\colon X(k)\to \{0,1\}^k$ is said to be a direct product function if there exists a function $f\colon X(1)\to \{0,1\}$ such that $F(\sigma) = (f(\sigma_1), \ldots, f(\sigma_k))$ for each $k$-face $\sigma$. In an effort to simplify components of the PCP theorem, Goldreich and Safra~\cite{GoldreichSafra}
introduced the problem of direct product testing, which asks whether one can test if $F\colon  X(k)\to \{0,1\}^k$ is correlated with a direct product function by querying $F$ on only $2$ inputs. Dinur and Kaufman~\cite{DinurK17} conjectured that there exist bounded degree complexes with a direct product test in the small soundness regime.
We resolve their conjecture by showing that for all $\delta>0$, there exists a family of high-dimensional expanders with degree 
$O_{\delta}(1)$ and a 
$2$-query direct product tester
with soundness $\delta$.

We use the characterization given by~\cite{BafnaMinzer} and independently by~\cite{DiksteinD-agreement}, who showed that some form of non-Abelian coboundary expansion (which they called ``Unique-Games coboundary expansion'') is a necessary and sufficient condition for a complex to admit such direct product testers. 
Our main technical contribution is a general technique for showing coboundary expansion of complexes with coefficients in a non-Abelian group. This allows us to prove that the high dimensional expanders constructed by~\cite{ChapmanL} satisfy the conditions of~\cite{BafnaMinzer}, thus admitting a 2-query direct product tester with small soundness. 
\end{abstract}

\section{Introduction}
The primary goal of this paper is 
to construct direct product testers 
with constant degree. Earlier works by a subset of the authors~\cite{BafnaMinzer} and by Dinur and Dikstein~\cite{DiksteinD-agreement} have shown 
that there are coboundary-type properties that are sufficient for a complex to have in order to 
admit such a direct product tester. 
The main contribution of the
current paper is to establish that
the simplicial complex 
of~\cite{ChapmanL} has the sufficient condition presented in~\cite{BafnaMinzer}, implying that it admits direct product testers
with small soundness.

\subsection{Direct Product Testers}
The goal in direct product testing 
is to encode the values of a function $f\colon [n]\to\{0,1\}$
via a table $F$ which is testable. By that, we mean that there is an 
    efficient, randomized tester 
    that makes queries to $F$, 
    performs a test on the received 
    values, and decides to accept or reject. The tester should always 
    accept if $F$ is indeed a valid
    encoding of a function $f$; this property is often referred to
    as the completeness of the test. 
    In the converse direction, if 
    the tester accepts $F$ with 
    probability at least $s>0$, then
    $F$ must be close (in the case $s$ is close to $1$) or correlated (in the case $s$ is close to $0$) to a valid encoding of some function $f$. This property is often referred to as the soundness of the tester. The smallest $s$ for 
    which the last property holds 
    is called the soundness parameter
    of the encoding, and the number
    of queries the algorithm makes 
    is called the query complexity 
    of the algorithm. In this paper,
    we will restrict our attention to
    direct product testers with $2$ queries, which is the smallest 
    query complexity one can hope for. 

    The study of direct product testing and 
    the importance of its parameters such as completeness, soundness and 
    query complexity, originally
    comes from the work of 
    Goldreich and Safra~\cite{GoldreichSafra}.
    Their motivation was to 
    simplify components in the 
    proof of the PCP theorem~\cite{FGLSS,AS,ALMSS}, 
    and the parameters of the 
    direct product tester correspond exactly
    to the parameters of the
    PCP verifier. Another 
    motivation for the study of 
    direct product testing comes 
    from hardness amplification, 
    particularly in the 
    style of the parallel repetition theorem~\cite{Raz,Holenstein,Rao,DinurSteurerAnalytical,BravermanGarg}. Indeed, direct product testers with small soundness 
    can be viewed as a combinatorial
    analog of parallel repetition
    theorems, and these can sometimes
    be turned into proper parallel
    repetition theorems~\cite{ImpagliazzoKW09,DinurMeir}. We
    refer the reader to~\cite{DinurK17,BafnaMinzer} for further discussion on the 
    role of direct product testing
    in theoretical computer science.

    The most natural direct product 
    encoding of a function $f\colon [n]\to\{0,1\}$ is given by 
    the Johnson scheme. Here, for 
    $k\in\mathbb{N}$ thought of as
    a constant, one may encode $f$ 
    via the table $F\colon \binom{[n]}{k}\to\{0,1\}^k$ 
    as $F[A] = f|_{A}$.\footnote{We think of $[n]$ as being order in a canonical way, thus for a set $A$ of size $k$ containing $i_1<\ldots <i_k$ we define 
    $F[A] = (f(i_1),\ldots,f(i_k))$.}
    The Johnson encoding scheme also
    admits a corresponding natural
    direct product tester:
    \begin{enumerate}
        \item Sample $I\subseteq [n]$ of size $t$.
        \item Independently sample 
        $A,A'\subseteq [n]$ of size
        $k$ containing $I$.
        \item Query $F[A]$ and $F[A']$ and check that they
        agree on the coordinates of
        $I$.
    \end{enumerate}
    The above direct product tester
    has received significant attention,  as well as its variations. By now, all of these are well understood in the entire range of parameters $t$, see
    for example~\cite{DinurG08,ImpagliazzoKW09,DinurSteurer,BKM3}. In particular, it is well known that the tester has vanishing soundness in the setting that $t\approx \sqrt{k}$ and for sufficiently large $k$.\footnote{In fact, in that 
    case it is known that the soundness is exponentially small in $t$.} 
    
    The primary disadvantage with the Johnson
    encoding scheme and the corresponding direct product
    tester is its size. Indeed, 
    the size of the encoding $F$ of
    a function $f$ is $n^{k}$, which
    is polynomially large in the size of $f$. Additionally, if
    one wishes the soundness of the
    Johnson direct product tester to 
    be small, say $\delta$, one must
    take $k$ to be sufficiently large, hence getting a large blow-up in the encoding size. 
    For some applications, this 
    size blow-up is too costly. 
    For instance, in applications in
    PCPs, the blow-up introduced by
    the use of parallel repetition
    theorem is often what dominates
    the blow-up in the instance size
    reductions produce.

    \subsection{Direct Product Testers via High Dimensional 
    Expanders}
    In the quest for more efficient ways to amplify hardness (often called derandomized hardness amplification), Dinur and Kaufman~\cite{DinurK17}
    suggested high dimensional 
    expanders as a sparse object
    that may facilitate direct
    product testers. To describe 
    their result, we first take a quick detour to present several basic notions from the field of high dimensional expansion, abbreviated HDX henceforth.
    
    A $d$-dimensional 
    simplicial complex 
    $X = (X(0),\ldots,X(d))$ 
    with vertex set $X(1) = [n]$
    is a downwards closed 
    collection of subsets of $[n]$.
    We follow the convention that 
    $X(0) = \{\emptyset\}$, 
    and for each $i>1$ the set of 
    $i$-faces $X(i)$ is a collection
    of subsets of $X(1)$ of size $i$. The size of $X$ is the total
    number of faces in $X$. The 
    degree of a vertex $v\in X(1)$
    is the number of faces containing
    it, and the degree of $X$ is 
    the maximum degree over all $v\in X(1)$.
    \begin{definition}
For a $d$-dimensional simplicial complex $X = (X(0),X(1),\ldots,X(d))$, $0\leq i\leq d-2$
and $I\in X(i)$, the link of $I$ is the $(d-i)$-dimensional complex $X_I$ whose faces are given as
\[
X_I(j-i) = \{J\setminus I~|~J\in X(j), J\supseteq I\}.
\]
\end{definition}
For a $d$-dimensional complex $X=(X(0),X(1),\ldots,X(d))$ 
and $I\in X$ of size at most $d-2$, the 
graph underlying the link of $I$ is 
the graph whose vertices are $X_I(1)$ 
and whose edges are $X_I(2)$. 
We associate with $X$ a collection 
of distributions over faces. The distribution $\mu_d$ is the uniform 
distribution over $X(d)$, and 
for each $i<d$ the distribution 
$\mu_i$ is a distribution over 
$X(i)$ which results by picking 
$D\sim \mu_d$, and then taking 
$I\subseteq D$ of size $i$ 
uniformly. For convenience, we 
encourage the reader to think of $\mu_i$ as the uniform distribution 
over $X(i)$ (though often times this
is not the case, the fact that $\mu_i$ is not actually the uniform
distribution is rarely an issue).

\begin{definition}
    We say a $d$-dimensional
    simplicial complex $X$ is 
    a $\gamma$ one-sided local
    spectral expander if for 
    every $I\in X$ of size at 
    most $d-2$, the second 
    eigenvalue of the 
    normalized adjacency matrix
    of the graph $(X_I(1), X_I(2))$ 
    is at most $\gamma$.
\end{definition}
With this definition, the result of 
Dinur and Kaufman~\cite{DinurK17} 
asserts that if $X$ is a $\gamma$ 
one-sided local spectral expander
for a sufficiently small $\gamma$, 
then $X$ admits a $2$-query direct
product tester with soundness $s=1-\eps$. Their direct product tester 
is a direct analog of the Johnson
direct product tester. It is parameterized by $k\in\mathbb{N}$
which is much smaller than $d$, 
and proceeds as follows. 
Given an assignment $F\colon X(k)\to\{0,1\}^k$, the tester proceeds as follows:
\begin{enumerate}
    \item Sample $D\sim \mu_d$.
    \item Sample $I\subseteq D$
    of size $\sqrt{k}$ uniformly.
    \item Sample $I\subseteq A,A'\subseteq D$ independently.
    \item Query $F[A]$ and $F[A']$
    and check that they agree on 
    the coordinates of $I$.
\end{enumerate}
The result of Dinur and Kaufman~\cite{DinurK17} asserts 
that provided that $\gamma$ is small enough, if $F$ passes the above tester with probability $1-\eps$,
then there exists $f\colon X(1)\to\{0,1\}$ such that with probability 
$1-O(\eps)$ over the choice of $A\sim \mu_k$ we have that $F[A] = f|_{A}$. We refer to the above tester
as the canonical direct product 
tester of $X$.

The main open problem left from the work~\cite{DinurK17} is whether 
there are high-dimensional expanders
that facilitate direct product testers in the low soundness regime.
With regards to this, the works~\cite{BafnaMinzer,DiksteinD-agreement} 
both show that spectral expansion is
insufficient. Namely, there are 
high-dimensional expanders with 
arbitrarily good local spectral
expansion for which the above natural
direct product tester fails in the 
low-soundness regime. These two 
works then went on to seek additional properties
of high-dimensional expanders 
that will imply that the canonical
direct product tester has 
small soundness.
    
\subsection{Main Result}
The main result of this paper is that there exist families of high-dimensional expanders for which the canonical direct product tester has small soundness. In fact, the complexes we prove this for are  variants of the complexes
constructed by Chapman and Lubotzky~\cite[Section 5]{ChapmanL} for sufficiently large parameters $p,n\in\mathbb{N}$, which
are explicit and efficiently computable. The relevance of variants of the complexes of~\cite{ChapmanL} was communicated to us in~\cite{DDLpersonal}.

\begin{theorem}\label{thm:main}
For all $\delta,\gamma>0$, for 
sufficiently large $d$, there 
    is $\Delta\in\mathbb{N}$ such 
    that the following holds. There are infinitely many $n$ for which there is a $d$-dimensional simplicial complex $X$ such that:
    \begin{enumerate}
        \item The complex $X$ has $n$ vertices and is a $\gamma$ one-sided spectral expander.
        \item For $k \geq k_0(\delta)$, the canonical direct product tester for $X$ 
        has soundness $\delta$.
        \item The complex $X$ has
        degree at most $\Delta$.
    \end{enumerate}
    \end{theorem}
    It has been brought to our attention that Dikstein, Dinur, and Lubotzky have concurrently and independently established Theorem~\ref{thm:main}.

    
    \subsection{UG Coboundary Expansion}
    The starting point of the proof
    of Theorem~\ref{thm:main} is the
    characterization of~\cite{BafnaMinzer} of high-dimensional expanders 
    for which the canonical tester
    has small soundness.
    Towards this end, we begin by defining non-Abelian affine unique games over graphs.
\begin{definition}\label{def:unique-games}
An instance of Affine Unique Games $\Psi = (G, \Pi)$ over the symmetric group $\S_m$ consists of a graph $G = (V,E)$ and a collection of permutations, one for each ordered edge, $\Pi = \{\pi_{u,v}\}_{(u,v)\in E}$, where $\pi_{u,v} \in \S_m$ and $\pi_{u,v} = \pi_{v,u}^{-1}$. An assignment to $\Psi$ is a function $A: V \to \S_m$, and we denote by $\val(A)$ the fraction of constraints satisfied by $A$, that is, 
\[
\val(A) = \Pr_{(u,v)\sim E}[A(u) = \pi_{u,v}A(v)].
\]
We denote by $\viol(A)$ the fraction of constraints violated by $A$, that is, $\viol(A) = 1-\val(A)$. The value of the instance $\Psi$ is defined as 
$\val(\Psi) = \max_{A: V \to \S_m}\val(A)$.
\end{definition}
We refer to the above instances as affine Unique Games because they can be thought of as systems of equations of the
form $A(u)A(v)^{-1} = \pi_{u,v}$ 
over $S_m$, wherein the goal is to find an $\S_m$-labeling of the vertices that satisfies as many of the equations as possible. 

\begin{definition}
Let $G = (V,E)$ be a graph, equipped with a distribution $\cD$ over triangles in $G$. We say that a UG instance $\Phi = (G,\{\pi_{u,v}\}_{(u,v) \in E(G)})$ is $(1-\delta)$-triangle consistent if:
\[\Pr_{(u,v,w) \sim \cD}[\pi_{u,v}\pi_{v,w}\pi_{w,u} = \text{id}] \geq 1-\delta.\]
We let $\incons(\Phi)$ denote the fraction of triangles that are inconsistent.
\end{definition}

Note that if we have a graph $G=(V,E)$ and an assignment 
$A\colon V\to \S_m$, then defining 
$\pi_{u,v} = A(u)A(v)^{-1}$ gives
a fully triangle consistent instance
of affine Unique-Games. The property
of coboundary expansion asserts that every highly triangle consistent instance on $G$ 
arises in this way. To define UG coboundary expansion for a complex we first look at a certain family of graphs arising from a complex.

\begin{definition}
Let $X$ be a $d$-dimensional simplicial complex, and let $r\leq d/3$. We define the graph 
$G_r[X]$ whose vertex set is 
$X(r)$, and whose set of edges 
$E_r[X]$ consists of pairs of 
vertices $(u,v)$ such that $u\cup v\in X(2r)$. The distribution over
triangles associated with this graph
is the distribution where we pick 
a face from $X(3r)$ according to 
the measure $\mu_{3r}$, and then 
split it randomly as $u\cup v\cup w$
where $u,v,w\in X(r)$.
\end{definition}

\dor{With the definition of the graph 
$G_r[X]$, we may introduce the notion of UG coboundary expansion from~\cite{BafnaMinzer}. A similar definition also appeared in~\cite{DiksteinD-agreement} and~\cite{dikstein2023swap}, but it will be more convenient for us to use the terminology of~\cite{BafnaMinzer}.}
\begin{definition}\label{def:ug_coboundary_expander}
We say that a $d$-dimensional simplicial 
complex $X$ is an $(m,r,\xi,c)$ UG coboundary expander if for all $t\leq r$
and for all functions $f\colon E_t[X]\to S_m$ that are $(1-\xi)$-consistent on triangles, there is $g\colon X(t)\to S_m$
such that
\[
\Pr_{u\cup v\sim \mu_{2t}}
\big[\pi(u,v) = g(u)g(v)^{-1}\big]\geq 1-c.
\]
\end{definition}
We remark that the notion of UG coboundary expansion can be seen as a non-Abelian version
of the usual notion of coboundary 
expansion. For $r=1$ this definition
appeared in earlier works \dor{of~\cite{DinurMeshulam,GotlibK23,DiksteinD23}}. 
In the case of Abelian groups, there is a natural extension of coboundary
expansion to higher dimensions. That definition however does not make much sense over non-Abelian groups, and therefore one has to come up with a 
different definition for non-Abelian groups and $r>1$. Indeed, this is the main point of Definition~\ref{def:ug_coboundary_expander}, and this will be the only notion of coboundary expansion discussed herein.

The main result of~\cite{BafnaMinzer}
asserts that a high-dimensional 
expander $X$ with sufficiently strong
UG coboundary parameters has a canonical direct product tester with 
small soundness; see~\cite[Theorem 1.9, Theorem B.1]{BafnaMinzer}. 
For future reference, we give here 
the version we use in our application:
\begin{theorem}\label{thm:BM}
    There exists $c>0$ 
    such that for all $\eps,\delta>0$ there is 
    $\eta$ and sufficiently large $m,r\in\mathbb{N}$, 
    such that for sufficiently large
    $k$, sufficiently large $d$ 
    and sufficiently small $\gamma>0$
    the following holds. Suppose that $X$ is a $d$-dimensional simplicial complex such that:
    \begin{enumerate}
        \item A $\gamma$ one-sided local spectral expander.
        \item An $(m,r,2^{-r/\log\log\log r},c)$ UG
        coboundary expander.
    \end{enumerate}
    If $F\colon X(k)\to\{0,1\}^k$ passes the canonical direct product tester on $X$ with
    probability at least $\delta$, 
    then there exists $f\colon X(1)\to\{0,1\}$ such that
    \[
        \Pr_{A\sim \mu_k}\left[\Delta(F[A],f|_{A})\leq \eps k\right]\geq \eta. 
    \]
\end{theorem}
Using Theorem~\ref{thm:BM}, it
suffices to construct a complex  
which is simultaneously a strong
spectral expander, as well as a UG
coboundary expander for sufficiently
good parameters. With regard to 
this, we mention the work of~\cite{DiksteinD23} that presents a technique that allows one to 
prove coboundary-type properties with bounds that are independent of the dimension of the complex. We also mention the recent work
~\cite{dikstein2023swap} which 
shows coboundary-type properties
for general buildings with for up to $\delta = 2^{-O(r)}$ inconsistent
triangles, and an improved bound 
of $\delta = 2^{-O(\sqrt{r})}$ 
for the spherical building of type A.

Studying spherical buildings and
proving coboundary expansion 
results for them is a central part
of our proof as well, however the
result of~\cite{dikstein2023swap} is
insufficient as far as we know. 
The spherical buildings we have
to study are spherical buildings of
type $C_{n}$, and for them we have to 
be able to handle Unique-Games instances with $\delta = 2^{-o(r)}$
fraction of inconsistent triangles.
Indeed, the main technical contribution is a fairly general technique for proving that a simplicial complex is a UG coboundary expander.

\subsection{Our Techniques}
The rest of this introductory section
is devoted to an overview of the proof of Theorem~\ref{thm:main}, 
and we start with the following
definition.
\begin{definition}\label{def:coboundary-constant}
We say that a graph is a $(C(G), \beta(G))$-coboundary expander over $\S_m$ if for all $\delta \in [0,1]$ and all UG instances $\Psi$ over $S_m$, with $\incons(\Psi) \leq \delta$, there exists an assignment $A: V \to S_m$ with $\viol_\Psi(A) \leq C(G)\delta+\beta(G)$.
When $\beta(G)$ is $0$, we say that a graph is a $C(G)$-coboundary expander over $\S_m$.
\end{definition}
The additive error $\beta(G)$ is necessary in our proofs, but 
we encourage the reader to ignore
it and think of it as $0$ for 
the purposes of this section.
Fix a simplicial complex $X$; we
will eventually take $X$ to be Chapman Lubotzky complex with an appropriate choice of parameteres~\cite{ChapmanL}. 
The main components of our proof 
proceed as follows:
\begin{enumerate}
    \item {\bf Local to global 
    for cosystolic expansion:} 
    we use an idea of~\cite{EvraKaufman, DiksteinD23} 
    who show that to prove
    that a simplicial complex 
    is a cosystolic expander, 
    it suffices to show that its
    links are coboundary expanders.
    Here and throughout, we 
    say that a graph $G=(V,E)$ is 
    a $C$-cosystolic expander if
    for every Unique-Games instances
    over $G$ with at most $\delta$ 
    fraction of inconsistent triangles, we may
    modify the constraints of $G$ 
    in at most $C\delta$ fraction 
    of the edges and get that 
    all triangles are fully consistent. We use this idea
    to show that if the links
    of $X$ are coboundary 
    expanders, then $X$ itself 
    is a cosystolic expander 
    with similar parameters.
    \item {\bf Vanishing Cohomology of $X$:} inspired by the first 
    item, one is motivated to ask
    the question of what is the 
    constraint structure of fully 
    triangle consistent Unique Games instances on the graphs
    associated with $X$. Here, we use the fact, 
    communicated to us by Dikstein, Dinur and 
    Lubotzky~\cite{DDLpersonal}, that the parameters of the Chapman--Lubotzky~\cite{ChapmanL} can be chosen 
    appropriately so that the complex has vanishing cohomology over $\S_m$.
    
    \item {\bf Coboundary expansion of the links of $X$:} Together, 
    the first two items imply 
    that to prove that $X$ is a 
    coboundary expander, it suffices
    to prove that the links of $X$
    are coboundary expanders. Indeed, in that case, if $\Psi$ 
    is a $(1-\delta)$-consistent 
    Unique-Games instance over 
    $\S_m$, by the first 
    item we can modify the constraints
    of $\Psi$ on at most $C\delta$
    fraction of edges for $C = 2^{o(r)}$ to get an instance $\Psi'$
    which is fully triangle consistent. Invoking the 
    second item, we conclude a structural result about the constraints of $\Psi'$, which 
    automatically implies a 
    similar result for $\Psi$.

    Proving that the links of $X$ have sufficiently good coboundary expansion consists of the bulk of our effort in this paper. We do so
    by an inductive argument on
    the parameter $r$. 
\end{enumerate}
In the rest of this section, we
elaborate on the third item above
as we consider it to be the main
contribution of this paper. 
We begin by noting that for our complex
$X$, the links are either isomorphic to
the spherical buildings 
of type $C_n$ or to product of two such spherical buildings. Given a 
prime $p$ and a dimension $d$, the spherical building of type $A$ refers to the complex whose vertex
set is the set of all nontrivial subspaces of 
$\mathbb{F}_p^d$, and whose $k$-faces
are $k$-flags of subspaces, namely 
$\{A_1,\ldots,A_k\}$ where $A_1\subseteq A_2\subseteq\ldots\subseteq A_k$. 

The spherical buildings with type 
$C$ are similar, except that 
the vertices only consist of isotropic subspaces with respect to a symplectic form. 
Here, the dimension $d$ is an even
number $2n$, and one defines a 
symplectic form
$\omega\colon \mathbb{F}_p^{2n}\times\mathbb{F}_p^{2n}\to\mathbb{F}_p$ 
by
\[
\omega(x,y) = \sum\limits_{i=1}^nx_i y_{i+n} -x_{i+n}y_i.
\]
A subspace $V\subseteq \mathbb{F}_p^{2n}$ is then called 
\emph{isotropic} if $\omega(x,y) = 0$
for all $x,y\in V$. With this in 
mind, the spherical building of 
type $C$ certainly has a nice 
structure. However, as we are 
aiming for an inductive approach
which will necessitate for us to 
look at various ``restrictions''
of this complex, it turns out 
to be convenient to think about
this complex more abstractly 
as a measure $\mu$ over $\prod\limits_{i=1}^{n}X_i$, 
where $X_i$ is the collection 
of all isotropic subspaces of $\mathbb{F}_p^d$ of dimension $i$.
Namely, $\mu$ is just the uniform distribution over top dimensional 
faces of $X$.

The most crucial feature of $\mu$ 
that makes it easier to work with
is the fact that it is an $\eps$-product
measure in the sense of~\cite{GurLL22}, where $\eps = O(1/\sqrt{p})$. Informally, an $\eps$-product measure is one in which one can perform the usual type of discrete Fourier
analysis as in product
domains such as the Boolean 
hypercube $\{0,1\}^n$, or product-like domains as the Johnson scheme
$\binom{[n]}{k}$ with constant $k$. 
Additionally, $\eps$-productness 
is a property that is preserved under
conditioning, making it friendly for inductive processes. We defer a 
precise definition of $\eps$-product 
measures to Section~\ref{sec:prelim}. Below,
we will focus on the measure $\mu$ itself, but we remark 
that our arguments have to also
work with restrictions of $\mu$ 
(which they do; the notations become somewhat more complicated, and hence we omit this discussion).

\subsubsection{Moving to Tripartite
Instances and $\eps$-product Distributions}
\paragraph{Set-up: moving to the tripartite problem:} 
consider the $n$-dimensional 
spherical building
of type $C$ complex $X$ as above,
and let $\Psi$ be a Unique-Games
instance over its $r$ faces which
is $(1-\delta)$-triangle consistent. 
Sample $R_1,R_2,R_3\subseteq [n]$
independently of size $r$, and consider the 
Unique-Games instance $\Psi'$ 
induced by $\Psi$ on the tripartite
graph $T(R_1,R_2,R_3)$. Here and 
throughout, the tripartite weighted graph $T(R_1,R_2,R_3)$ is the graph
whose faces are $r$-faces that are 
subsets of $R_1, R_2$ and $R_3$, 
and whose edges and constraints are 
induced by $\Psi$. It can be shown
via standard arguments that the 
fraction of inconsistent triangles in
$\Psi'$ is $(1+o(1))\delta$. Furthermore, using an idea from~\cite{DiksteinD23}, we show that it suffices
to prove that the instance $\Psi'$ has a good cosystolic coefficient 
for a good fraction of choices of $R_1,R_2,R_3$. Roughly speaking, if
we show that with probability at least $p$, the induced instance $\Psi'$ has cosystolic coefficient 
$C$, then the instance $\Psi$ will
have a cosystolic coefficient 
at most $O(C/p)$. 
\dor{We remark that the idea to switch to a restricted instance appears in~\cite{dikstein2023swap} too, and they switch to a $\sqrt{d}$-partite version of the problem (see Section 8.2 therein).}

Our argument will not 
show a strong enough cosystolic coefficient for every choice of 
$R_1,R_2,R_3$, and it depends 
on two features of them:
\begin{enumerate}
    \item Separatedness: we would like the elements in $R_1\cup R_2\cup R_3$ to be as far apart as possible form each other. Note
that the expected magnitude of 
an element in the union is $\Theta(n)$, and as we are choosing 
$3r$ of them uniformly, we expect
them to be roughly $\Theta(n/r)$ 
apart. Ideally, we would have liked any two distinct elements from $R_1\cup R_2\cup R_3$ to be $\Theta(n/r)$ apart, however the probability for that (over the choice of $R_1,R_2,R_3$)  is $2^{-\Theta(r)}$,
which is too small for our purposes as we are shooting for a 
$2^{o(r)}$ coboundary constant. 
We thus settle for weaker separatedness, and for our argument it suffices to have $R_1,R_2,R_3$ to be $\Theta(n/r^2)$-separated (which happens with probability $\Theta(1)$).
\item Well spread: consider 
an interval $I\subseteq [d]$, 
say $I = \{s,s+1,\ldots,s+L-1\}$, and think of $L$ as being of the order $n/r^{0.99}$. Note that in expectation, each one of $R_1,R_2,R_3$ contains 
$L\frac{r}{n}$ element from 
the interval. We say that $R_1$
is well spread if for every 
interval $I$ of length $L$ we 
have that $R_1\cap I$ has 
size roughly $L\frac{r}{n}$, and 
say that $R_1,R_2,R_3$ is well 
spread if each one of $R_1$, $R_2$ and $R_3$ is well spread. 
We would like $R_1,R_2,R_3$ to 
be well spread, and a standard 
application of concentration bounds show that the probability 
a randomly chosen $R_1,R_2,R_3$ 
is well spread with probability $1-o(1)$.
\end{enumerate}
\dor{We remark that notions that are close in spirit to the notion of well-spread above have also been used in~\cite{dikstein2023swap}, see Definition 8.3 therein}.
For the rest of this section 
we fix $R_1,R_2,R_3$ that are well spread and $\Theta(n/r^2)$ 
separated.

\paragraph{Moving to the language of $\eps$-product distributions:} 
now that we are considering tripartite
instances $\Psi'$ over $R_1,R_2,R_3$,
it will be useful to think of the 
measure $\mu$ over $\prod\limits_{i\in R_1\cup R_2\cup R_3} X_i$ that underlies this complex; namely, this is the distribution $\mu_{3r}$ restricted
to triangles in $\Psi'$. This measure
can be proved to be an $\eps$-product
measure for $\eps = \eps(p)$ which
is a vanishing function of the field 
size $p$; for our intents this means
that $\eps$ can be guaranteed to be
as small as we wish compared to the
dimension of the complex and $r$. 
We denote by $C_{r,r,r}(\mu)$ the
coboundary constant of the tripartite
graph underlying $\Psi'$. With 
these notations, our primary objective now is to prove that $C_{r,r,r}(\mu) = 2^{o(r)}$. More 
generally, for $1\leq t\leq r$, 
we define
\[
C_{t,t,t}(\mu)
=\max_{
\substack{
R_1'\subseteq R_1, \card{R_1'}=t\\ 
R_2'\subseteq R_2, \card{R_2'}=t\\ 
R_3'\subseteq R_3, \card{R_3'} = t}\\
}\max_a \text{ Coboundary constant of }T(R_1',R_2',R'_3~|~x_{(R_1\cup R_2\cup R_3)\setminus (R_1'\cup R_2'\cup R_3')} = a), 
\]
where the graph 
$T(R_1',R_2',R'_3~|~x_{(R_1\cup R_2\cup R_3)\setminus (R_1'\cup R_2'\cup R_3')} = a)$ is
the induced subgraph of $T$ 
on the vertices that agree 
with the restriction 
$x_{(R_1\cup R_2\cup R_3)\setminus (R_1'\cup R_2'\cup R_3')} = a$. The main benefit 
of $\eps$-product measures 
is that it provides a clean 
abstraction of he spectral 
properties of our tripartite 
graphs which is suitable for induction.

\subsubsection{An Inductive 
Approach to Coboundary Expansion: the Base Case}
To get some intuition, we first 
consider the case that $r=1$. 
In that case, the underlying graph
of $\Psi'$ is a tripartite inclusion
graph induced by $3$ distinct dimensions, say $i<j<\ell$, which we know to be $\Theta(n/r^2)$
separated. Inspecting this graph, 
one can show that the diameter of
this graph is at most 
\[
O\left(\max\left(\frac{j}{\card{i-j}},\frac{\ell}{\card{\ell-j}}\right)\right)
=O(r^2),
\]
which suggests that the graph is 
very well connected. 
In this case, 
we use the cones method \dor{as in~\cite{gromov2010singularities,lubotzky2016expansion,KaufmanMass2,kozlov2019quantitative,kaufman2019coboundary} in the setting of Abelian groups, and~\cite{DiksteinD23,ChapmanL,dikstein2023swap} in the setting of non-Abelian groups}. The cones method is a standard combinatorial technique allowing one to deduce
global assignments from good local consistency. The method requires us to construct a set of canonical paths between vertices which will be used for ``propagation'', as well as triangulations of cycles that are formed by an edge $(V,W)$ in the graph and the two canonical constructed paths from some vertex $U$, using only a small number of triangles.\footnote{The cones method also requires sufficient good transitive symmetry from the complex, which holds for free in the case of the spherical buildings we study.} To make 
this possible, we have to construct the set of
canonical paths in a rather careful manner. 
With these paths, we are able to break each
formed cycle between $U,V,W$ into a small 
collection of $8$-cycles (as well as some triangles), which we then triangulate via 
a careful case analysis.

In our context, using the cones method 
we can show that the coboundary
constant of the graph underlying
$G$ is polynomial in its diameter. 
In particular, this implies 
that $C_{1,1,1}(\mu)\leq r^{O(1)}$. We remark that working with well-separated indices 
is a useful idea if one
wishes to prove results that are 
independent of the dimension of the
complex, first introduced by~\cite{DiksteinD23}.
For us, it will be important that 
the above applies not
only to the measure $\mu$ itself, 
but rather to any restriction of $\mu$ leaving at least $3$ coordinates
alive. 

\mitali{We remark that a related statement\footnote{See Section~\ref{sec:grass-base-case} for a comparison between the statements.} was established in~\cite[Lemma 8.10]{dikstein2023swap} for spherical buildings of type A.}

\subsubsection{An Inductive 
Approach to Coboundary Expansion: the Extended Base Case}
The inductive base case above
for $C_{1,1,1}(\mu)$ is 
insufficient for our purposes;
indeed, using it and our inductive step yields a bound
of $2^{\tilde{O}(r)}$ on the 
coboundary expansion of $T(R_1,R_2,R_3)$.\footnote{We remark that with additional work, this bound may be improved to $2^{O(r)}$, but we do not
know how to break this barrier using only the base case for $r=1$.} To go beyond this bound 
we identify another base case, 
that we explain next, for which 
we can also directly bound the
coboundary constant. 

Let $k$ be a parameter to be thought of as a small power of 
$r$, and consider subsets 
$S_1, S_2, S_3\subseteq R_1\cup R_2\cup R_3$ with the following
properties:
\begin{enumerate}
    \item Equal sizes: we have that $|S_1| = |S_2| = |S_3| = 3k$ and 
    $|S_i\cap R_j| = k$ for any $i,j\in \{1,2,3\}$.
    \item Well ordered: we have that $\max_{s_1\in S_1} s_1 < \min_{s_2\in S_2}s_2\leq \max_{s_2' \in S_2} s_2'<\min_{s_3\in S_3} s_3$. 
    In words, all elements of $S_1$ are smaller than 
    all elements of $S_2$, 
    and all elements of $S_2$ are smaller than all elements of $S_3$.
\end{enumerate}
Our extended base case gives
good bounds for the coboundary 
constant of the graph $T(S_1,S_2,S_3;\mu)$. To get 
some intuition, note that 
when we attempt to construct paths and 
triangulations in the graph 
$T(S_1,S_2,S_3)$, incidences between vertices depend only on partial information on them. For example, if $U\in \supp(\mu^{S_1})$ and $V\in \supp(\mu^{S_2})$ are vertices, 
then the fact of whether $(U,V)$ is an edge or not depends only 
on the subspace in $U$ corresponding to dimension $\max_{s_1\in S_1} s_1$ and 
the subspace in $V$ corresponding to dimension $\min_{s_2\in S_2} s_2$. Thus, 
incidences in this graph depend
only on the $4$ parameters in the second condition above, and 
it thus stands to reason one 
can extend the triangulation from the case $r=1$ to the current case. 

Luckily, it turns out there is a 
surprisingly clean way to go 
about proving a statement along
these lines that only uses the 
base case $r=1$ in a black box manner. Indeed, letting 
$i = \max_{s_1\in S_1} s_1$, 
$j = \min_{s_2\in S_2} s_2$, 
$j' = \max_{s_2\in S_2} s_2$ 
and $\ell = \min_{s_3\in S_3} s_3$, we show, using an argument similar to the inductive step explained below, that
\begin{equation}\label{eq:ind_intro_break}
C(T(S_1,S_2,S_3;\mu))
\leq 
C(T(\{i\},\{j\},\{\ell\};\mu)
\cdot
\max(C_1,C_2,C_3),
\end{equation}
where 
\begin{align*}
&C_1 = \max_{a}C(T(S_1\setminus\{i\},S_2,S_3; \mu~|~X_{i} = a)),
\qquad
C_2 = \max_{b}C(T(S_1,S_2\setminus\{j\},S_3; \mu~|~X_{j} = b)),\\
&\qquad\qquad\qquad\qquad\qquad\qquad C_3 = \max_{c}C(T(S_1,S_2,S_3\setminus\{\ell\}; \mu~|~X_{\ell} = c)).
\end{align*}
The first term on the right hand
side of~\eqref{eq:ind_intro_break} 
can be bounded by the base 
case $r=1$ above. The second
term, namely 
$\max(C_1,C_2,C_3)$, is more interesting, and we focus on $C_1$ for concreteness. 
Thinking of $\mu$ as the uniform
distribution over top faces of
the type C spherical building, 
one observes that once we 
condition on the vertex of dimension $i$ in a face, the 
vertices of dimension strictly 
less than $i$ and the vertices of dimension more than $i$ form a product structure. In particular, since all elements of $S_1\setminus \{i\}$ are smaller than $i$ and 
all elements of $S_2,S_3$ are bigger than it, we have that 
for $\nu = \mu~|~X_i = a$ it holds that
\[
\nu^{(S_1\setminus\{i\})\cup S_2}
=\nu^{S_1\setminus\{i\}}\times \nu^{S_2},
\qquad
\nu^{(S_1\setminus\{i\})\cup S_3}
=\nu^{S_1\setminus\{i\}}\times \nu^{S_3}.
\]
Thus, the tripartite graph 
$T(S_1\setminus\{i\}, S_2, S_3; \nu)$ is composed of the bipartite graph between $S_2$ and $S_3$, and the graph between $S_1\setminus\{i\}$ 
and them is complete. This 
trivializes the task of constructing triangulations, 
and indeed we show that the coboundary constant $C_1$ is
dominated by the diameter 
of the bipartite graph between $S_2$ and $S_3$, 
which is easily seen to be $O(r^2)$ 
using the separatedness. \dor{The ``productization'' 
phenomenon, namely that the fact that a complex/distribution 
breaks into a tensor/product of several complexes/ distributions after conditioning, has also appeared in~\cite[Section 8.6.1]{dikstein2023swap}.}

Overall, the extended base 
case allows us to argue
that $C(S_1,S_2,S_3; \mu)\leq r^{O(1)}$ even when the sets
$S_1$, $S_2$ and $S_3$ are large, so long as they are 
``well ordered''. The ability to find well ordered subsets is precisely the reason we require $R_1,R_2,R_3$ to be well spread.

\begin{remark}
We remark that the additional 
feature of $\mu$ we use here 
is that restricting a coordinate 
of it breaks $\mu$ into a 
product of $2$ 
distributions. While this property is easy to see directly  for type A and C spherical buildings discussed herein, it
is in fact more general for 
any spherical building. 
This is best seen by looking 
at the Coxeter diagrams that
correspond to a given spherical building, and inspecting that
removing a vertex from them (which corresponds to the operation of restriction/ taking a link) either turns a type into an easier-to-handle type, or else it disconnects the diagram, 
in which case the product structure as discussed above emerges.
\end{remark}

\subsubsection{The Inductive Approach: the Inductive Step}
Armed with the base cases, we are now
ready to describe an inductive approach to bound $C_t(\mu) = C_{t,t,t}(\mu)$.

\paragraph{Restrictions and 
finding good local solutions:}
Suppose that we 
have $R_1'\subseteq R_1$, 
$R_2'\subseteq R_2$ 
and $R_3'\subseteq R_3$ 
of size $t$ during our induction,
and choose disjoint well ordered subsets 
$S_1, S_2, S_3$ of $R_1'\cup R_2'\cup R_3'$ that each contains $k$ elements from each one of $R_1,R_2$ and $R_3$. 
An easy argument shows that
for appropriate choice of parameters, so long as $t\geq r^{0.99}$ we are able to find 
such $S_1,S_2,S_3$, and we 
fix such sets henceforth. 
We later explain how to handle
the case $t<r^{0.99}$.

Choose restrictions 
$(a_1,a_2,a_3)\sim \mu^{S_1\cup S_2\cup S_3}$. We may consider $3$ 
distinct induced graphs on $T = T(R_1',R_2',R_3')$ corresponding to
these restrictions, which are 
$T_1 = T(R_1'\setminus S_1,R_2'\setminus S_1,R_3'\setminus S_1~|~\mu~|~x_{S_1} = a_1)$, $T_2 = T(R_1'\setminus S_2,R_2'\setminus S_2,R_3'\setminus S_2~|~\mu~|~x_{S_2} = a_2)$ and 
$T_3 = T(R_1'\setminus S_3,R_2'\setminus S_3,R_3'\setminus S_3~|~\mu~|~x_{S_3} = a_3)$. Each one of these graphs 
refers to the induced graph on vertices that agree with the restriction $a_1$, $a_2$ and $a_3$
respectively, and we may consider the induced affine Unique-Games instance $\Psi_{a_1}'$ $\Psi_{a_2}'$ and $\Psi_{a_3}'$ on them. 
Let $\eps_{a_1}$, $\eps_{a_2}$ and $\eps_{a_3}$ denote the
fraction of inconsistent triangles in $\Psi_{a_1}'$, $\Psi_{a_2}'$ and $\Psi_{a_3}'$ respectively, and choose
$X_{a_1}$ $X_{a_2}$ and $X_{a_3}$
to assignments that satisfy the maximum fraction of constraints. Thus, by definition
we get that
\[
\viol(X_{a_1}; T_1)
\leq C_{t-k} \eps_{a_1},
\qquad
\viol(X_{a_2}; T_2)
\leq C_{t-k} \eps_{a_2},
\qquad 
\viol(X_{a_3}; T_3)
\leq C_{t-k} \eps_{a_3}.
\]

Since we are working with affine 
instances of Unique-Games, once we
have one good solution we may 
apply affine shifts to it to get
a collection of good solutions; 
in our case, for each $\pi\in \S_m$
we may define $L_{a_1}[\pi] = X_{a_1}\pi$, and have that $L_{a_1}$
consists of solutions to $\Psi_{a_1}$
each satisfying all but $C_{t-k}\eps_{a_1}$ of the 
constraints. Similarly, we may define $L_{a_2}$ and $L_{a_3}$.

\paragraph{Relating good local solutions:} consider now the instance induced by $\Psi$ 
on the graph 
\[
G_{a_1,a_2} = T(R_1',R_2',R_3'~|~x_{S_1\cup S_2} = (a_1,a_2)).
\]
Denote by $\viol(X_{a_1}; G_{a_1,a_2}), \viol(X_{a_2}; G_{a_1,a_2})$ the fraction of
constraints violated by $X_{a_1}, X_{a_2}$ in 
$G_{a_1,a_2}$ respectively.
From the point 
of view of the restriction $S_1, a_1$, this is a randomly chosen
induced sub-instance of $T_1$, 
and hence we expect that 
$\viol(X_{a_1}; G_{a_1,a_2})\lll \viol(X_{a_1})$, and analogously $\viol(X_{a_2}; G_{a_1,a_2})\lll \viol(X_{a_2})$. Using the $\eps$-productness of $\mu$, the graph $G_{a_1,a_2}$ has second
singular value bounded away from 
$1$, and therefore any two good
solutions to an Affine Unique-Games
over it must be the same up to an 
affine shift (we remark that 
this is an idea whose origin is 
algorithms for solving 
affine Unique-Games over some special classes of graphs~\cite{BBKSS,BafnaMinzerSolve}). Thus, we may find a shift $\pi_{a_1,a_2}\in \S_m$ that nearly
forms a matching between the lists 
$L_{a_1}, L_{a_2}$. More precisely, 
we are able to show that
\[
\E_{a_1, a_2}
\left[\Pr_{u\in G_{a_1,a_2}}\left[X_{a_1}(u)\neq \pi_{a_1,a_2}X_{a_2}(u)\right]\right]
\lll
\E_{a_1, a_2}\left[
\viol(X_{a_1}; G_{a_1,a_2}) + 
\viol(X_{a_2}; G_{a_1,a_2})\right]
\lll C_{t-k}\delta.
\]

Using our observations so far, we 
may consider the tripartite graph
over restrictions, whose triangles
correspond to a triplet of restrictions $a_1,a_2,a_3$ that 
are valid under $\mu$. We consider an
Affine Unique-Games instance $\Psi_{{\sf restrict}}$ on it,
that has the constraint $\pi_{a_1,a_2}^{-1}$ on the edge $(a_1,a_2)$, and the goal is to 
choose labels from $\S_m$ to each
vertex so as to satisfy as many of
the constraints as possible.  Using our
arguments so far, we can 
argue that the fraction of 
inconsistent triangles is 
at most $\eta = O\left(C_{t-k} \delta\right)$, 
and as the underlying graph 
of $\Psi_{{\sf restrict}}$ is 
$T(S_1,S_2,S_3)$ we are able
to conclude that we may find a
solution to $\Psi_{{\sf restrict}}$ that satisfies 
all but $O(C_kC_{t-k} \delta)$
of constraints. Going in this route, one can conclude the recursion $C_t\lll C_k C_{t-k}$, 
which ultimately gives $C_t\leq 2^{O(t\log r)}$. This is, of 
course, insufficient, and the reason for why we require the extended base case. Using 
it, we have that $C(T(S_1,S_2,S_3))\leq r^{O(1)}$ 
giving us the recursion 
$C_t\leq r^{O(1)} C_{t-k}$ 
so long as $t\geq r^{0.99}$.
Iterating, this gives
\[
C_r\leq r^{r/k} C_{r^{0.99}}
\leq r^{r/k} 2^{O(r^{0.99}\log r)}\leq 2^{O(r^{1-c})}
\]
for some absolute constant $c>0$.

To complete the overall picture, we now explain how we lift good solutions to $\Psi_{{\sf restrict}}$ to a good solution 
of $\Psi'$. Suppose that $A$ is 
an assignment to $\Psi_{{\sf restrict}}$ 
satisfying $1-\eta$ fraction 
of constraints, and let $U$ be 
a vertex in $\Psi'$, say $U\in \supp(\mu^{R_1'})$ without loss of generality. The idea is to 
ask the opinion of a random vertex in $S_1$ that is consistent with $U$, and assign $U$ accordingly. More precisely, 
we sample $a_1\sim \mu~|~x_{R_1'} = U$, and 
assign $B[U] = A[a_1](U)$. 
Indeed, using standard spectral 
arguments, we show that for typical $U$, the value of 
$A[a_1](U)$ for $a_1$ chosen in
this way is almost fixed, and 
furthermore that the assignment
$B$ satisfies all but $\eta$ 
fraction of constraints in $\Psi'$.
\section{Preliminaries}\label{sec:prelim}
\paragraph{Notations:} 
We use standard big-$O$ notations: we denote $A = O(B)$ 
or $A\lll B$ if $A\leq c\cdot B$ for some absolute constant
$c>0$. Similarly, we denote $A = \Omega(B)$ or $A\ggg B$ 
if $A\geq c B$ for some absolute constant $c>0$. We also 
denote $k\ll d$ to denote the fact that $d$ is taken to 
be sufficiently large compared to any function of $k$. 
For a distribution $\mu$ over $X_1 \times \ldots \times X_r$ and 
a subset $S\subseteq [r]$, we denote
by $\mu^{S}$ the marginal distribution of $\mu$ on the coordinates of $S$. We denote 
by ${\sf supp}(\mu)$ the support
of $\mu$, and for a subset 
$S\subseteq [r]$ and an assignment 
$V\in \prod\limits_{i\in S}X_i$ 
we denote by $\mu~|~X_S=V$ the 
distribution of $X\sim \mu$ conditioned on $X_S = V$. 
If $A$ is a finite set and 
$i\leq \card{A}$, the notation
$B\subseteq_{i} A$ means that 
we sample a subset of size $i$
of $A$ uniformly.

\subsection{Graphs Associated with Distributions and $\eps$-product Distributions}
Let $\mu$ over $X_1 \times \ldots \times X_r$. The following 
definition describes bipartite 
graphs that can be associated with $\mu$:
\begin{definition}\label{def:bip-mu}
Let $\mu$ be a distribution on $X_1 \times \ldots \times X_r$. Let $L,R \subseteq [r]$ be two disjoint non-empty sets. Let $A(L,R;\mu)$ be the bipartite graph produced as follows: the vertices are $\supp(\mu^{L})$ and $\supp(\mu^{R})$ and to sample an edge we choose a sample $X$ from $\mu$, and output $(X_{L},X_{R})$. We will let $A_{L,R}$ denote the corresponding operator $A_{L,R}: L^2(X_L,\mu^L) \rightarrow L^2(X_R,\mu^R)$, with $A_{L,R}f(v) = \E_{w \sim \mu^L|X_R = v}[f(w)]$.
\end{definition}
\dor{We remark that the bipartite graphs from Definition~\ref{def:bip-mu} 
have also been considered in~\cite{DiksteinD19}, wherein they go under the name the ``colored walk'' (see Definition 4.10 therein).}

Similarly, one can create tripartite graphs from $\mu$:
\begin{definition}
Let $\mu$ be a distribution on $X_1 \times \ldots \times X_r$. Let $S_1, S_2, S_3 \subseteq [r]$ be three disjoint non-empty sets. Let $T(S_1,S_2,S_3;\mu)$ be the tripartite graph produced as follows: the vertices are $\supp(\mu^{S_1}), \supp(\mu^{S_2})$ and $\supp(\mu^{S_3})$, and to sample an edge we choose a sample $X$ from $\mu$, and with probability $1/3$ each output $(X_{S_i},X_{S_j})$ for $i \neq j, i,j \in [3]$. 

Note that when one of the $S_i$'s is $\emptyset$, say $S_1$, the graph $T(\emptyset,S_2,S_3;\mu)$ denotes the tripartite graph with one vertex $\emptyset$ in its first part, $\supp(\mu^{S_2}|a),\supp(\mu^{S_3}|a)$ in its second and third parts, and we sample an edge by sampling $X \sim \mu$ and and outputting the pairs $(\emptyset,X_{S_2})$, $(\emptyset,X_{S_3})$ or $(X_{S_2},X_{S_3})$ with equal probability. When more of the $S_i$'s are $\emptyset$ we can similarly create these tripartite graphs.
\end{definition}

We now discuss the definition of $\eps$-product distributions 
from~\cite{GurLL22}. We begin by defining $\eps$-pseudorandom distributions.
\begin{definition}
We say that a distribution $\cD$ over $X_1 \times X_2$ is $\eps$-pseudorandom if the second largest singular value of $A_{\{1\},\{2\}}$  is at most $\eps$.
\end{definition}

Next, we define the notion of having 
$\eps$-pseudorandom skeletons.
\begin{definition}
We say that a distribution $\cD$ over $Y_1 \times \ldots \times Y_t$ has $\eps$-pseudorandom skeletons if for all $i \neq j \in [t]$, the marginal distribution $\cD^{\{i,j\}}$ is $\eps$-pseudorandom.
\end{definition}
We are now ready to define the notion
of $\eps$-product distributions.
\begin{definition}
    We say that $\mu$ is an $\eps$-product distribution over $X_1 \times \ldots \times X_r$ if for all $S \subseteq [r]$ of size at most $r-2$ and all $V \in \supp(\mu^S)$, the conditional distribution $\mu|X_S = V$ has $\eps$-pseudorandom skeletons.
\end{definition}

\mitali{In \cite{DiksteinD19} it is shown that a $d$-partite high-dimensional expander is equivalent to the notion of $\eps$-product distributions (see Lemma 7.5 therein).} Many properties of $\eps$-product distributions were established in~\cite{GurLL22} and~\cite{DiksteinD19}, and we will require a few of them. \mitali{In particular, we will need~\cite[Lemma 3.3]{GurLL22} or~\cite[Claim 4.11]{DiksteinD19}, which asserts that if $L,R\subseteq [r]$ are disjoint and $\mu$ is an $\eps$-product distribution, then the second singular values of the bipartite graphs $A(L,R;\mu)$ are small.} 

\begin{lemma}\label{claim:gll-sing-val}
Let $\mu$ be an $\eps$-product distribution over $X_1 \times \ldots \times X_r$ and let $L,R \subseteq [r]$ be two disjoint sets. Then the second largest singular value of $A_{L,R}$ and $A_{R,L}$ is at most $\poly(r)\eps$.
\end{lemma}

\subsection{Properties of Expanders}
We need the following well known version of the expander
mixing lemma for bipartite graphs.
\begin{lemma}\label{lem:bip-eml}
Let $G = (U,V,E)$ be a bipartite graph in 
which the second singular value of the normalized adjacency matrix is at most $\lambda$. Then for all $A \subset U$ and $B \subset V$ we have that
\[
\left|\Pr_{(u, v) \in E}[u \in A, v \in B] -\mu(A)\mu(B)\right| \leq \lambda\sqrt{\mu(A)(1-\mu(A))\mu(B)(1-\mu(B))}.\]    
\end{lemma}

We also use the following standard sampling 
property of bipartite expanders.
\begin{lemma}\label{lem:sampling}
Let $G = (U,V,E)$ be a weighted bipartite graph with second singular value at most $\lambda$. Let $B \subset U$ and set 
$T = \left\{v \in V \mid \Pr_{u\sim v}[u \in B] > \eps+\Pr[B] \right\}$. Then $\Pr[T]\leq \lambda^2\delta/\eps^2$.
\end{lemma}

\subsection{Properties of Local Spectral Expanders}
Recall that we associated with each $d$-dimensional simplicial 
complex $X$ a sequence of measures $\{\mu_k\}_{1\leq k\leq d}$, where 
$\mu_k$ is a probability measure over $X(k)$. Note that for all $0 \leq t \leq r \leq d$, a sample according to $\mu_t$ can be drawn by first sampling $R \sim \mu_r$, and then sampling $T\subseteq_{t} R$ uniformly. The converse is also true: a sample from $\mu_r$ can be drawn by first sampling $T \sim \mu_t$, and then sampling $R$ from $\mu_r$ conditioned on containing $T$. These observations 
give rise to the standard ``up'' and ``down'' operators, which we present next. We only mention a few of their 
properties that are necessary for our arguments, and refer the reader to~\cite{dikstein2018boolean} for a more
comprehensive exposition.

\begin{definition}
The operator $U_i^{i+1}$ is a map from $L_2(X(i); \mu_i)$ to $L_2(X(i+1); \mu_{i+1})$ defined as
\[
U_i^{i+1}f(u) 
= 
\E_{v \subset_i u}\big[f(v)\big]
\]
for all $u\in X(i+1)$. For $j\geq k+1$, we define $U_k^j$ via composition of up operators: $U_k^j = U_{j-1}^j \circ \ldots \circ U_k^{k+1}$.
\end{definition}

\begin{definition}
The operator $D_i^{i+1}$ is a map from $L_2(X(i+1); \mu_{i+1})$ to $L_2(X(i); \mu_i)$ defined as
\[
D_i^{i+1}f(u) = \E_{v \supseteq_{i+1} u}\big[f(v)\big]
\]
for all $u \in X(i)$.
For $j\geq k+1$, we define $D_k^j$ via composition of down operators: $D_k^j = D_{k}^{k+1} \circ \ldots \circ D^j_{j-1}$.
\end{definition}

Abusing notations, we use the notations $U^j_k, D^j_k$ 
to denote the operators, as well as the real valued 
matrices associated with them. A key property of
the down and up operators is that they are adjoint:

\begin{claim}
For all $k \leq j \leq d$, $U_k^{j}$ and $D^{j}_k$ are adjoint operators: for all functions $f\colon X(k)\to\mathbb{R}$ and $g\colon X(j)\to\mathbb{R}$ it holds that $\ip{U_k^{j}f,g} = \ip{f,D^{j}_kg}$. \end{claim}

We need the following result due to~\cite{Oppenheim18} known as the trickling-down theorem that uses the eigenvalues of links at $X(d-2)$ to show that $X$ is a one-sided local spectral expander.
\begin{theorem}\label{thm:trickling-down}
Let $X$ be a $d$-dimensional simplicial complex such that the 1-skeleton of every link (including the empty one) is connected and for all $I \in X(d-2)$, the 1-skeleton of $I$ has second eigenvalue at most $\lambda$. Then $X$ is a $\frac{\lambda}{1-(d-1)\lambda}$-one-sided local spectral expander.  
\end{theorem}

We need the following lemma regarding the second eigenvalue of the down-up walks $U^j_{k}D^j_{k}$ on $X(j)$ ($j \geq k$), that can be found in~\cite{AlevL20}.

\begin{lemma}\label{lem:spectral_gap_of_graphs_from_HDX}
Let $(X, \mu)$ be a $d$-dimensional $\gamma$ one-sided local spectral expander. For all $i \leq d$ and $\alpha \in (1/i, 1)$, the largest singular value of $U^i_{\alpha i}$ and $D^i_{\alpha i}$ is at most $\sqrt{\alpha}+\poly(i)\gamma$. Thus the down-up random walk $U^i_{\alpha i}D^i_{\alpha i}$ on $X(i)$ has second largest singular value at most $\alpha + \poly(i)\gamma$. 
\end{lemma}

\subsection{Spherical Buildings of Type A}
Analyzing the Chapman-Lubotzky
complex will require us to study 
its links. In this section and 
in the next one, we present the 
spherical buildings of type A and 
of type C, which are morally the
graphs we will end up needing to study.
\begin{definition}
The spherical building of type A over $\F_q^{d+1}$ is a $d$-dimensional complex denoted by $SB^A_d(\F_q)$ with the set of maximal faces: 
\[\{(V_1,\ldots,V_d): V_1 \subset \ldots \subset V_{d}, V_i \subset_{i} \F_q^{d+1}\}.\] 
The $d$-faces are equipped with the uniform distribution which we also denote by $SB^A_d(\F_q)$. 
\end{definition}

The following is a well-known fact, but we give the proof here for completeness.
\begin{lemma}\label{lem:eps-product-A}
The distribution $SB^A_d(\F_q)$ is a $O(1/\sqrt{q})$-product distribution.   
\end{lemma}
\begin{proof}
Let $\mu=SB^A_d(\F_q)$. Take any set $S \subseteq [d], |S| \leq r-2$ and a valid restriction of it, say $V$. We need to show that $\mu | X_S = V$ has $1/q$-pseudorandom skeletons. Consider two coordinates $i \neq j \in [d]\setminus S$ and let $A_{i,j}$ be the bipartite graph/normalized adjacency operator corresponding to $\mu^{\{i,j\}}|X_S = V$. If there exists $k \in S$ such that $i < k < j$, then $A_{i,j}$ has second largest singular value $0$. So let us consider the case where there is no such coordinate $k$ between $i$ and $j$. Let $i'$ be the largest coordinate in $S$ that is less than $i$ and $j'$ be the smallest coordinate in $S$ that is greater than $j$ ($i' < i < j < j'$). Then the bipartite graph $A_{i,j}$ is isomorphic to the weighted inclusion graph between $i-i'$ and $j-i'$-dimensional subspaces of $\F_q^{j'-i'}$. It is well-known (see for example~\cite{brouwer2012distance,godsil2016erdos}) that this graph has second largest singular value $\lll 1/\sqrt{q^{j-i}} \leq 1/\sqrt{q}$. This shows that $\mu | X_S = V$ has $1/\sqrt{q}$-pseudorandom skeletons for all $S$ and $V$, thus showing that $\mu$ is a $1/\sqrt{q}$-product distribution.
\end{proof}

\subsection{Spherical Buildings of Type C}
Next, we present the spherical buildings of type $C$. Like the spherical buildings of type $A$, 
type $C$ spherical buildings 
are too defined using subspaces. 
However, we only consider subspaces
that are isotropic with respect
to a \emph{symplectic form}.
\begin{definition}\label{def:symp-form}
A symplectic bilinear form is a mapping $\omega: \F^{2n} \times \F^{2n}\rightarrow \F$ is a map which bi-linear, anti-symmetric -- $\omega(v,w) = -\omega(w,v), \forall v,w \in \F^{2n}$ and non-degenerate -- $\omega(u,v) = 0$ for all $v$ implies $u = 0$. 
\end{definition}
In this paper we fix the symplectic form:
\[\omega = \begin{pmatrix} 0 & I_n \\ -I_n & 0 \end{pmatrix},\]
which gives the bi-linear form 
$\omega(v,w) = \sum_{i = 1}^n v_i w_{n+i} - w_i v_{n+i}$.
One can check that $\omega(\cdot,\cdot)$ is a valid symplectic bilinear form.

\begin{definition}
A subspace $V \subseteq \F_q^{2d}$ is called isotropic if $\omega(v,w) = 0$ for all $v,w \in V$.   
\end{definition}

\begin{definition}
The symplectic group \(\text{Sp}(2d, F)\) is defined as the set of \(2d \times 2d\) matrices \(M\) over a field \(F\) that preserve the symplectic form \(\omega \). This group is characterized by matrices in $\GL_{2d}(\F)$ that satisfy \( M^T \omega M = \omega \), where \( \omega \) is the matrix representation of the symplectic form.
\end{definition}

\paragraph{Properties of $\sp_{2d}(\F)$:} For an invertible matrix $M$ and subspace $V \subseteq \F^{2d}$, let $M \circ V$ denote the subspace $\spn(Mv)_{v \in V}$. Then the group $\sp_{2d}(\F)$ has the following properties:
\begin{enumerate}
    \item If $V$ is a $t$-dimensional isotropic subspace of $\F_q^{2d}$, then $M\circ V$ is also a $t$-dimensional isotropic subspace of $\F_q^{2d}$.
    \item Furthermore, $\sp_{2d}(\F)$ acts transitively on the set of $t$-dimensional isotropic subspaces for all $t \leq d$.
\end{enumerate}

\begin{definition}
The spherical buildings of type C over $\F_q^{2d}$ is a $d$-dimensional complex denoted by $SB^C_d(\F_q)$ with the set of maximal faces:
\[
\{(V_1,\ldots,V_d): V_1 \subset \ldots \subset V_d, V_i \text{ is an isotropic subspace of dimension } i\}.
\]
The $d$-faces are equipped with the uniform distribution which we also denote by $SB^C_d(\F_q)$.
\end{definition}

The following lemma asserts that the
natural distribution associated with the spherical building of type $C$ is $\eps$-product for a small $\eps$.
\begin{lemma}\label{lem:eps-product-C}
Let $d,q \in \N$ with $q \geq \poly(d)$. Then $\mu=SB^C_d(\F_q)$ is a $O(1/\sqrt{q})$-product distribution.   
\end{lemma}

\begin{proof}
Fix $\eps = O(1/\sqrt{q})$. One can check that proving $\mu$ is a $\lambda$-product distribution is equivalent to showing that $SB^C_d(\F_q)$ denoted by $X$ is a $\lambda$-one-sided local spectral expander, and we focus on the latter task. To do so, it suffices to show that the 1-skeleton of all links $I \in X(d-2)$ have second eigenvalue at most $\eps$. Once we show that, the trickling-down theorem, Theorem~\ref{thm:trickling-down}, implies that $X$ is a $\frac{\eps}{1-d\eps} \lll \eps$ one-sided local spectral expander.

Fix a $(d-2)$-sized link $I$, and say that it fixes the set of coordinates $S \subset [d]$ of size $d-2$ to the isotropic subspaces $V_a, a\in S$, that form a valid inclusion chain. Let $i<j \in [d]$ be the two unfixed coordinates and let $\cD$ be the resulting conditional distribution on $X_{\{i,j\}}$. The second largest eigenvalue of the 1-skeleton of $I$ is equal to the second largest singular value of $A(\{i\},\{j\};\cD)$, since the 1-skeleton is a bipartite graph. We bound the latter using case analysis. 

\paragraph{Coordinates $i,j$ are not consecutive:}
In this case there exists $k \in S$ such that $i < k < j$. So we know that $\cD$ is a product distribution and $A(\{i\},\{j\};\cD)$ is the complete bipartite graph, hence the corresponding second largest singular value is $0$. 

\paragraph{Coordinates $i,j$ are consecutive but not equal to $(d-1,d)$:} In this case, $j=i+1$ and the coordinates $i+2$ and $i-1$ belong to $S$. Then the resulting bipartite graph is over the set of isotropic subspaces contained within $V_{i+1}$ and containing $V_{i-1}$. Since $V_{i+1}$ is an isotropic subspace
and every subspace within an isotropic subspace is isotropic, we get that $A(\{i\},\{j\};\cD)$ is isomorphic to the bipartite weighted inclusion graph between $1$ and $2$-dimensional subspaces of $\F_q^3$. This has largest singular value $\leq \eps$ as we saw in Lemma~\ref{lem:eps-product-A}.

\paragraph{Coordinates $(i,j)= (d-1,d)$:} In this case the coordinate $(d-2)$ belongs to $S$. The set of 
$d-1$ and $d$-dimensional isotropic subspaces containing $V_{d-2}$ is in one-to-one correspondence with the set of $1$ and $2$-dimensional subspaces respectively, that are symplectically orthogonal to $V_{d-2}$ and are not contained in $V_{d-2}$. Taking a quotient by $V_{d-2}$ we get that this is the set of $1$ and $2$-dimensional isotropic subspaces of $\F_q^{4}$. Therefore $A(\{i\},\{j\};\cD)$ is isomorphic to the bipartite inclusion graph between $1$ and $2$-dimensional isotropic subspaces of $\F_q^{4}$. Using~\cite[Theorem 9.4.3]{brouwer2012distance}, we get that this graph has second largest singular value $\lll \eps$ as required. 
\end{proof}

\section{A Local to Global Theorem for Coboundary Expansion}
The primary goal of this section
is to present an inductive approach
to prove upper bounds on the 
coboundary constants corresponding 
to the level $r$ faces of a simplicial complex $X$. Roughly 
speaking, starting with an 
initial assumption regarding 
the coboundary constant of $X$
on constant levels, we show how 
to lift it to a reasonable bound
on higher levels. 
\subsection{Tools}
\subsubsection{Basic Notions and Properties of the Tripartite Graph 
$T(R_1,R_2,R_3)$}
In this section, we present
some of the tools and notions 
that are necessary for our proof.
Throughout this section, we fix a set of indices $I = \{i_1,\ldots,i_{3r}\}$ and an $\eps$-product distribution $\mu$ over $\prod_{i \in I} X_{i}$. 
Here and throughout, $\eps$ should
be thought of as very small compared to all other parameters, 
and we encourage the reader to 
think of $\eps = 0$ at first reading.

We begin by formally defining the 
notion of coboundary expansion
(with additive error) for measures.
\begin{definition}\label{def:cobundary_expand_measure}
Let $\mu$ be a measure over $\prod\limits_{i\in I} X_i$ and 
let $0 \leq r_1,r_2,r_3 \leq r\in\mathbb{N}$ be integers. 
We say that $\mu$ is a $(C_{r_1,r_2,r_3},\beta_{r_1,r_2,r_3})$-coboundary expander over $\S_m$ if for all sets $S$ of size at most $3r - (r_1+r_2+r_3)$, for all restrictions $a \in \supp(\mu^S)$, and for all disjoint $R_1,R_2,R_3 \subseteq I \setminus S$ of sizes $r_1,r_2,r_3$ respectively, the tripartite graph $T(R_1,R_2,R_3;\mu|X_S = a)$ with respect to the distribution $\mu^{\cup R_i}|(X_S = a)$ over triangles, is a $(C_{r_1,r_2,r_3},\beta_{r_1,r_2,r_3})$-coboundary expander over $\S_m$  as per Definition~\ref{def:coboundary-constant}.  

When $\beta_{r_1,r_2,r_3} = 0$, we simply say that $\mu$ is a $C_{r_1,r_2,r_3}$-coboundary expander over $\S_m$.
\end{definition}

Definition~\ref{def:cobundary_expand_measure} will be of central interest to us, and to use it 
we must develop some tools to investigate coboundary 
expansion in the tripartite graphs $T(R_1,R_2,R_3;\mu)$. 
We begin with the following claim, 
asserting that this graph always
 has second singular value bounded away from $1$.
\begin{claim}\label{lem:eps-product-eigval}
Let $\mu$ be an $\eps$-product distribution over $X_1 \times \ldots \times X_r$, and let $S_1,S_2,S_3 \subseteq [r]$ be three disjoint non-empty sets. Let $T$ be the normalized adjacency operator of the graph $T(S_1,S_2,S_3;\mu)$. Then the second largest singular value of $T$is at most $1/2+\poly(r)\eps$.
\end{claim}

\begin{proof}
Let $T$ be the normalized adjacency operator of $T(S_1,S_2,S_3;\mu)$ and let $f$ be the second eigenvector of $T$ with $\E[f] = 0$ and $\|f\|_2 = 1$. We will now bound $\|Tf\|_2$. 

For any $i=1,2,3$, let $f_i$ be the function $f$ restricted to $X_{S_i}$ thought of as an element in $L_2(X_{S_i};\mu^{S_i})$; 
we also denote $a_i = \E[f_i] = |\E_{x \sim \mu^{S_i}}[f(x)]|$. We have that for all $x \in S_1$, 
\[Tf(x) = \frac{1}{2}A_{S_2,S_1}f(x) + \frac{1}{2}A_{S_3,S_1}f(x),\]
and similarly for $x \in S_2$ or $S_3$. Similarly, we denote by $(T f)_i$ the restriction of $Tf$ to $X_{S_i}$. With these notations, we have that
\begin{align*}
\|(Tf)_1\|_2 \leq \frac{1}{2}\|A_{S_2,S_1}f_2\|_2 + \frac{1}{2}\|A_{S_3,S_1}f_3\|_2 \leq \frac{1}{2}(a_2+\poly(r)\eps) + \frac{1}{2}(a_3+\poly(r)\eps),
\end{align*}
where in the last transition we used Claim~\ref{claim:gll-sing-val} to bound the second singular value of $A_{S_2,S_1}$ and $A_{S_3,S_1}$ to bound 
$\|A_{S_2,S_1}f_2 - a_2\|_2 
= \|A_{S_2,S_1}(f_2 - a_2)\|_2\leq {\sf poly}(r)\eps$. Squaring and simplifying gives us that
\begin{align*}
\|(Tf)_1\|^2_2 &\leq \frac{1}{4}(a_2+a_3)^2+\poly(r)\eps.
\end{align*}
As $a_1+a_2+a_3 = 3\E[f] = 0$, 
we conclude that $\norm{(Tf)_1}_2^2\leq \frac{1}{4}a_1^2 + {\sf poly}(r)\eps$, 
and similarly $\norm{(Tf)_2}_2^2
\leq \frac{1}{4}a_2^2 + {\sf poly}(r)\eps$ and 
$\norm{(Tf)_3}_2^2
\leq \frac{1}{4}a_3^2 + {\sf poly}(r)\eps$. Multiplying by $1/3$ and summing up we get
\begin{align*}
\|Tf\|^2_2 = \E_{i\in [3]}[\|(Tf)_i\|_2^2] 
&\leq \frac{1}{12}(a_1^2 + a_2^2 + a_3^2)+\poly(r)\eps\\
&= \frac{1}{12}(\E[f_1]^2 + \E[f_2]^2 + \E[f_3]^2)+\poly(r)\eps\\
&\leq \frac{1}{12}(\E[f_1^2] + \E[f_2^2]^2 + \E[f_3^2])+\poly(r)\eps\\
&= \frac{1}{4}\|f\|_2^2+\poly(r)\eps,
\end{align*}
which is at most $\frac{1}{4} + {\sf poly}(r)\eps$ as $\norm{f}_2=1$. 
Taking a square root gives that $\|Tf\|_2 \leq 1/2+\poly(r)\eps$.
\end{proof}

\subsubsection{Almost Uniqueness of Good Solutions to Affine UG Instances}
We begin by defining shifts of assignments to Affine Unique-Games instances $G$.
\begin{definition}
Given an affine Unique-Games instance $\Phi$ with alphabet $S_m$, an assignment $A$ to it and $\pi\in S_m$, we denote by $A \circ \pi$ denote the assignment that $(A\circ \pi)(v) = X(v)\pi$. 
\end{definition}
Suppose $\Phi$ is an instance of affine Unique-Games as in Definition~\ref{def:unique-games}, 
and suppose that $A$ is an assignment to it. It can easily 
be seen that $\val(A) = \val(A\circ \pi)$ for every $\pi\in \S_m$. Thus, if $A$ satisfies many of the constraints 
of $\Psi$, then so does $A\circ \pi$. The following claim says that if the underlying constraint 
graph is an expander, then all good assumptions are essentially shifts of one good assignment $X$.
This idea appeared in~\cite{BBKSS} 
in the context of the Abelian 
version of affine Unique-Games, 
but the proof in the non-Abelian case is essentially the same and
we give it for completeness.
\begin{claim}\label{claim:shift-partition}
Let $\Phi = (G,\Pi,\S_m)$ be a UG instance on $G$ which is an expander graph with second eigenvalue $\lambda$, and let $X,X' \in \S_m^{V(G)}$ be two solutions to $\Phi$. Then there exists a permutation $\pi \in \S_m$ such that
\[\Pr_{v \in G}[X(v) \neq X'(v)\pi] \leq \frac{\viol(X)+\viol(X')}{1-\lambda}.\]
\end{claim}

\begin{proof}
Fix $X$ and $X'$ as in the statement of the claim, and partition $V(G) = \cup_{\pi\in \S_m} C_\pi$ where for each $\pi \in \S_m$ we define
$C_\pi = \{u \in G: X(v) = X'(v)\pi\}$. Note that if an edge $(u,v) \in G$ is satisfied by both $X$ and $X'$, then its endpoints lie in the same part $C_{\pi}$. 
Indeed, suppose that $X(v) = X'(v)\pi$, then
\[
X(u) 
= \pi_{u,v} X(v)
= \pi_{u,v} X'(v) \pi
= X'(u) \pi.
\]
Thus, if $(u,v)$ goes across distinct parts of the partition $\{C_{\pi}\}_{\pi\in S_m}$, then it
must be violated either by $X$ 
or by $X'$.

For a set of vertices $S$ let $\mu(S)$ denote the measure of $S$ in $V(G)$ and for a set of edges $T$ let $\mu(T)$ denote the measure of $T$ in $E(G)$. Let $E(S,\overline{S})$ denote the set of edges that cross between $S$ and its complement. By (the easy direction of) Cheeger's inequality, for each $\pi\in \S_m$ we have that
\[\frac{1}{2}\mu(E(C_\pi,\overline{C_\pi})) \geq (1-\lambda) \mu(C_\pi)(1 - \mu(C_\pi)).\]
Summing this over $\pi \in \S_m$ we get,
\[
\frac{1}{2}\sum_{\pi\in \S_m}\mu(E(C_\pi,\overline{C_\pi})) \geq (1-\lambda) (1 - \sum_\pi \mu(C_\pi)^2).
\]
Note that each edge that crosses 
between distinct parts of $\{C_{\pi}\}_{\pi\in\S_m}$ is 
counted twice on the left side, 
and by our earlier observation 
it must be violated either by $X$ 
or by $X'$, and so 
$\frac{1}{2}\mu(E(C_\pi,\overline{C_\pi}))\leq \viol(X) + \viol(X')$. Plugging this and simplifying gives that
\[\sum_\pi \mu(C_\pi)^2 \geq 1 - \frac{\viol(X)+\viol(X')}{1-\lambda},\]
and as $\sum\limits_{\pi\in \S_m} \mu(C_{\pi})) = 1$, it follows that there exists $\pi\in \S_m$ 
such that $\mu(C_{\pi})\geq 1-\frac{\viol(X)+\viol(X')}{1-\lambda}$.
\end{proof}

\subsection{Exponential Bound on the Coboundary Constant via Lopsided Induction}
Our bounds on the coboundary constants $C_{r_1,r_2,r_3}$ will always
follow an inductive strategy. 
All of our inductive arguments 
will be of similar spirit, 
though they differ in some 
technical aspects that tailor
them for different uses.
This section is devoted for
the most simplistic of these 
inductive approaches, for which we have two utilities. 
First, it will be useful for us to handle  $r_1,r_2,r_3$ that are relatively small (e.g. $r^{0.99}$). Secondly, it will 
be useful for us when we extend
the base case into the extended base case as described the introduction.
In this case, we have the following result:
\begin{lemma}\label{lem:easy-recursion} 
There exists an absolute constant $K>0$ such that the following holds. Suppose that $\mu$ is an $\eps$-product measure with $\eps\leq r^{-K}$ which is a $(C_{1,1,1}(\mu),\beta_{1,1,1}(\mu))$-coboundary expander over $\S_m$. Then for all $k_1,k_2,k_3 \in \N$ such that $k_1+k_2+k_3 \leq r$, $\mu$ is a $(C_{k_1,k_2,k_3}(\mu),\beta_{k_1,k_2,k_3}(\mu))$-coboundary expander over $\S_m$ with 
\[
C_{k_1,k_2,k_3}(\mu) \leq O(C_{1,1,1}(\mu))^{k_1+k_2+k_3},~~~ \beta_{k_1,k_2,k_3}(\mu) \leq O(C_{1,1,1}(\mu))^{k_1+k_2+k_3}\beta_{1,1,1}(\mu).
\]
\end{lemma}
\dor{We remark that the statement of Lemma~\ref{lem:easy-recursion} is a variant of~\cite[Theorem 1.3]{dikstein2023swap}. 
To the best of our knowledge the proof we present here is 
different, and the ideas below will be important in our 
subsequent arguments. We now present the proof of Lemma~\ref{lem:easy-recursion}.} 
\begin{proof}
Fix any three pairwise disjoint sets $R_1,R_2,R_3$ with sizes $0 < r_1,r_2,r_3 \in \N$ such that $\sum r_i \leq r$, indices $j_1 \in R_1,j_2 \in R_2$ and $j_3 \in R_3$, a set $S \subset [d] \setminus \bigcup R_i$, and a restriction $A_0 \in \supp(\mu^S)$. Let $\cD$ denote $\mu~|~X_S = A_0$. For any restriction $a_1 \in {\sf supp}(\cD^{j_1})$, let $G_{a_1}$ denote the induced subgraph $T(R_1\setminus j_1, R_2,R_3;\cD|X_{j_1} = a_1)$ and similarly define the graphs $G_{a_2}, G_{a_3}$ for every $a_2 \in \supp(\cD^{j_2}), a_3 \in \supp(\cD^{j_3})$. Then we will show that,
\begin{align}
&C(T(R_1,R_2,R_3;\cD)) \leq C(T(\{j_1\},\{j_2\},\{j_3\};\cD))\cdot \max_{\substack{i \in [3] \\a_i \in \supp(\cD^{j_i})}}(C(G_{a_i})). \label{eq:easy-recursion-main-C}\\
&\beta(T(R_1,R_2,R_3;\cD)) \leq C(T(\{j_1\},\{j_2\},\{j_3\};\cD))\cdot\max_{\substack{i \in [3] \\a_i \in \supp(\cD^{j_i})}}(\beta(G_{a_i}))+\beta(T(\{j_1\},\{j_2\},\{j_3\};\cD)). \label{eq:easy-recursion-main-beta} 
\end{align}

Additionally, when either of the $r_i$'s are $0$ we show that $C(T(R_1,R_2,R_3;\cD)) \leq 1$, as well as that $\beta(T(R_1,R_2,R_3;\cD))=0$. This along with the equations above is enough to conclude the lemma as we show in the end of the proof.

To prove \eqref{eq:easy-recursion-main-C} and \eqref{eq:easy-recursion-main-beta} fix any affine UG instance $\Phi$ on $T(R_1,R_2,R_3;\cD)$ which has $\delta$-fraction of inconsistent triangles. 

\paragraph{Setting up lists on restrictions:}
For every $a_i \in \supp(\cD^{j_i})$ let $X_{a_i} \in \S_m^{V(G_{a_i})}$ be an assignment on $V(G_{a_i})$ with maximum value, and let $L_{j_i\rightarrow a_i}$ be an ordered list of solutions indexed by permutations in $\S_m$, where $L_{S_i\rightarrow a_i}[\pi] = X_{a_i}\circ \pi$. 

\paragraph{Setting up permutations between restrictions:}
For every two restrictions of different parts, say $a$ of $j_1$ and $b$ of $j_2$ where $(a,b)$ is a valid restriction (that is, where $(a,b)\in {\sf supp}(\cD^{\{j_1,j_2\}})$), let $G_{a,b}$ be the induced subgraph $T(R_1, R_2, R_3;\cD | (X_{j_1} = a, X_{j_2} = b))$. Then by Claim~\ref{claim:shift-partition} there exists a permutation $\pi_{a,b}$ that satisfies the following:
\begin{equation}\label{eq:bad-ij_basic}
\Pr_{v \sim G_{a,b}}[X_a(v) \neq X_b(v)\pi_{a,b}] \leq \frac{\viol(X_a;G_{a,b})+\viol(X_b;G_{a,b})}{1-\lambda(G_{a,b})} 
\lll \viol(X_a;G_{a,b})+\viol(X_b;G_{a,b}),  
\end{equation}
where we used Claim~\ref{lem:eps-product-eigval} to bound $\lambda(G_{a,b})\leq 1/2+\poly(r)\eps\leq 0.51$. Let 
\[
\bad(a,b) = \{v\in G_{a,b}~|~X_a(v)\neq X_b(v)\pi_{a,b}\}.
\]
For any restriction $b$ as above, let $L_b\circ \pi$ denote the list where $L_b\circ \pi[\pi'] = L_b[\pi\pi'] = X_b\circ\pi\pi'$. 
With these notations and the definition of $\bad(a,b)$, we get that $L_a|_{G_{a,b} \cap \overline{\bad(a,b)}} = L_b\circ\pi_{a,b}|_{G_{a,b} \cap \overline{\bad(a,b)}}$. 

\paragraph{Counting bad triangles:}
For any restriction $(a_1,a_2,a_3) \in \supp(\cD^{\{j_1,j_2,j_3\}})$ let $G_{a_1,a_2,a_3}$ denote the graph $T(R_1\setminus j_1, R_2\setminus j_2, R_3 \setminus j_3;\cD|(a_1,a_2,a_3))$. Clearly, we have that,
\[\E_{\substack{a_3 \sim \cD^{j_3}|(a_1,a_2)}}[\mu_{G_{a_1,a_2,a_3}}(\bad(a_1,a_2))] = \mu_{G_{a_1,a_2}}(\bad(a_1,a_2)).\]
Fix $a_1,a_2$. Using Markov's inequality and~\eqref{eq:bad-ij_basic} 
it follows that
\begin{equation}\label{eq:avg-viol_basic}
\Pr_{\substack{a_3 \sim \cD^{j_3}|(a_1,a_2)}}\left[\mu_{G_{a_1,a_2,a_3}}(\bad(a_1,a_2)) > \frac{1}{100}\right]
\lll 
\viol(X_{a_1};G_{a_1,a_2})+\viol(X_{a_2};G_{a_1,a_2}).
\end{equation}
We say a restriction $(a_1,a_2,a_3)$ is bad if $\mu_{G_{a_1,a_2,a_3}}(\bad(a_i,a_k)) > 1/100$ for some $i \neq k$. By \eqref{eq:avg-viol_basic} we get
\begin{align}\label{eq1_basic}
&\E_{(a_1,a_2,a_3) \sim \cD^{\cup j_i}}[\Ind((a_1,a_2,a_3) \text{ is bad})]\notag\\
&\leq \sum_{i \in [3]} \E_{\substack{(a_\ell,a_k) \sim \cD^{j_\ell \cup j_k}\\ \ell,k\neq i}}\E_{a_i \sim \cD^{j_i}|(a_\ell,a_k)}[\Ind(\mu_{G_{a_1,a_2,a_3}}(\bad(a_\ell,a_k)) > \frac{1}{100})]\notag\\
&\lll \sum_{i \in [3]} \E_{\substack{(a_\ell,a_k) \sim \cD^{j_\ell \cup j_k}\\ \ell,k\neq i}}[\viol(X_{a_\ell};G_{a_\ell,a_k})+\viol(X_{a_k};G_{a_\ell,a_k})]\notag\\
&\lll\sum_{i \in [3]} \E_{\substack{a_i \sim \cD^{j_i}}}[\viol(X_{a_i};G_{a_i})],
\end{align}
where in the first inequality we used the union bound, the second inequality we used \eqref{eq:avg-viol_basic}. 
By definition we have that $\viol(X_{a_1};G_{a_1}) \leq C(G_{a_1})\eps_{a_1}+\beta(G_{a_1})$, where $\eps_{a_1}$ is the fraction of violated triangles in $G_{a_1}$. 
Therefore,
\begin{align*}
\eqref{eq1_basic}
\lll
\sum_{i \in [3]} \E_{\substack{a_i \sim \cD^{j_i}}}[C(G_{a_i})\eps_{a_i}+\beta(G_{a_i})] 
\lll  \max_{\substack{i \in [3] \\a_i \in \supp(\cD^{j_i})}}(C(G_{a_i}))\delta+\max_{\substack{i \in [3] \\a_i \in \supp(\cD^{j_i})}}(\beta(G_{a_i})) : =\delta'.    
\end{align*}


\paragraph{Creating a UG instance on graph over restrictions:}
Let $\Psi$ be the following UG instance on $H = T(\{j_1\},\{j_2\},\{j_3\};\cD)$: each edge $(a_i,a_k)$ has the constraint $\pi_{a_i,a_k}^{-1}$. That is, we want to find a solution $A$ that maximizes the fraction of edges satisfying $A(a_i) = \pi_{a_i,a_k}^{-1}A(a_k)$. Note that a triangle $(a_1,a_2,a_3)$ is consistent if $\pi_{a_3,a_1}^{-1}\pi_{a_1,a_2}^{-1}\pi_{a_2,a_3}^{-1} = \text{id}$ or equivalently if $\pi_{a_2,a_3}\pi_{a_1,a_2}\pi_{a_3,a_1} = \text{id}$.

First note that if a triangle $(a_1,a_2,a_3)$ is not bad, then it is consistent. To see this, note that in this case,  by the union bound we have $\mu_{G_{a_1,a_2,a_3}}(\bad(a_1,a_2) \cup \bad(a_2,a_3) \cup \bad(a_1,a_3)) \leq 0.03$, which gives that $\mu_{G_{a_1,a_2,a_3}}(\good(a_1,a_2,a_3)) \geq 0.97 > 0$, where the latter set is defined as those vertices in $G_{a_1,a_2,a_3}$ that don't belong to any of $\bad(a_i,a_k)$. We have that: 
\begin{align*}
L_{a_1}|_{\good(a_1,a_2,a_3)} = L_{a_2}\circ \pi_{a_1,a_2}|_{\good(a_1,a_2,a_3)}\\
L_{a_2}|_{\good(a_1,a_2,a_3)} = L_{a_3}\circ \pi_{a_2,a_3}|_{\good(a_1,a_2,a_3)}\\
L_{a_3}|_{\good(a_1,a_2,a_3)} = L_{a_1}\circ \pi_{a_3,a_1}|_{\good(a_1,a_2,a_3)},
\end{align*}
which implies that: 
\[L_{a_3}|_{\good(a_1,a_2,a_3)} = L_{a_3}\circ \pi_{a_2,a_3}\circ \pi_{a_1,a_2}\circ \pi_{a_3,a_1}|_{\good(a_1,a_2,a_3)} = L_{a_3}\circ\pi_{a_2,a_3}\pi_{a_1,a_2}\pi_{a_3,a_1}|_{\good(a_1,a_2,a_3)}.\] 
Taking any $v \in {\good(a_1,a_2,a_3)}$ we get that, $L_{a_3}[\text{id}][v] = L_{a_3}\circ \pi_{a_2,a_3}\pi_{a_1,a_2}\pi_{a_3,a_1}[\text{id}][v]$, where the former is $X_{a_3}(v)$ and the latter is $X_{a_3}(v)\pi_{a_2,a_3}\pi_{a_1,a_2}\pi_{a_3,a_1}$, which implies $\pi_{a_2,a_3}\pi_{a_1,a_2}\pi_{a_3,a_1} = \text{id}$. 

Using the coboundary expansion of the graph $T(\{j_1\},\{j_2\},\{j_3\};\cD)$ we get that there exists a UG solution to the vertices violating at most $(C(H)\delta'+\beta(H))$-fraction of the edges. Call this solution $A$.

\paragraph{Lifting the solution:} We will now use $A$ to create a highly satisfying solution $B$ to $G = T(R_1,R_2,R_3;\cD)$. To each vertex $u \in R_i$, define  $B(u) = L_{u|_{j_i}}[A(u|_{j_i})][u]$. We will now upper bound the fraction of edges that $B$ violates for $\Phi$. 

Consider an edge $(u_i,u_k) \in G$ between parts $R_i,R_k$, and let $(a_i,a_k)$ denote the edge $(u_i|_{j_i},u_k|_{j_k})$ in the restriction graph $H$.
Let $\good(G)$ be the set of edges where the following three events hold: (1) $A$ satisfies the edge $(a_i,a_k)$, i.e. $A(a_i) = \pi_{a_i,a_k}^{-1} A(a_k)$, (2) $X_{a_i}$ satisfies the edge $(u_i,u_k)$, i.e. $X_{a_i}(u_i) = \pi_{u_i,u_k}X_{a_i}(u_k)$  and (3) $u_k$ does not belong to $\bad(a_i,a_k)$, i.e.  $X_{a_i}(u_k) = X_{a_k}(u_k)\pi_{a_i,a_k}$. It is easy to see that $B$ satisfies every good edge, that is, $B(u_i) = \pi_{u_i,u_k} B(u_k)$. Indeed, using the events (1), (2) and (3), we have that,
\begin{align*}
B(u_i) &= L_{a_i}[A(a_i)][u_i] = X_{a_i}(u_i)A(a_i) \\
&= X_{a_i}(u_i)\pi_{a_i,a_k}^{-1} A(a_k) &\text{by (1)}\\
&= \pi_{u_i,u_k}X_{a_i}(u_k)\pi_{a_i,a_k}^{-1} A(a_k) &\text{by (2)}\\
&= \pi_{u_i,u_k}X_{a_k}(u_k)\pi_{a_i,a_k}\pi_{a_i,a_k}^{-1} A(a_k) &\text{by (3)}\\
&=  \pi_{u_i,u_k}B(u_k).
\end{align*}

So it suffices to upper bound the fraction of edges that are not in $\good(G)$, which we denote by $\bad(G)$. Sampling an edge, we can upper bound the probability it doesn't satisfy at least one of the events. For (1), 
\[\Pr_{(u_i,u_k) \in G}[A(a_i) \neq \pi_{a_i,a_k}^{-1} A(a_k)] = \viol(A) \lll C(H)\delta'+\beta(H).\]
For (2), 
\begin{align*}
\Pr_{(u_i,u_k) \in G}[X_{a_i}(u_i) \neq \pi_{u_i,u_k}X_{a_i}(u_k)]
&= \E_{\substack{i \sim [3]\\a_i \sim \cD^{j_i}}}\E_{\substack{(u_1,u_2,u_3)\sim \cD|(X_{j_i} = a_i)\\k \sim [3]\setminus i}}\left[\Ind(X_{a_i}(u_i) \neq \pi_{u_i,u_k}X_{a_i}(u_k))\right]\\
&\lll \E_{\substack{i \sim [3]\\a_i \sim \cD^{j_i}}}\left[\viol(X_{a_i};G_{a_i})\right] \lll \delta'.
\end{align*}
For (3), using~\eqref{eq:bad-ij_basic} we have
\begin{align*}
\Pr_{(u_i,u_k) \in G}[u_k \not\in \bad(a_i,a_k)]
&\lll\E_{\substack{(a_i,a_k)\sim \cD^{j_i \cup j_k}}}\E_{\substack{(u_i,u_k,u_\ell)\sim \cD|(a_i,a_k)\\r \sim \{i,k\}}}\left[\Ind(u_r \not\in \bad(a_i,a_k))\right]\\  &\lll\E_{\substack{(a_i,a_k)\sim \cD^{j_i \cup j_k}}}\left[\viol(X_{a_i};G_{a_i,a_k})+\viol(X_{a_k};G_{a_i,a_k})\right] \\
&\lll \E_{\substack{a_i\sim \cD^{j_i}}}\left[\viol(X_{a_i};G_{a_i})\right]\\
&\lll \delta'.
\end{align*}
Adding up these probabilities we get,
\begin{align*}
&\Pr_{(u_i,u_k) \in G}[(u_i,u_k) \in \bad(G)] 
\lll \delta' + C(H)\delta'+\beta(H) \\
&\qquad\qquad\qquad\qquad\qquad\lll \left(C(H)\max_{\substack{i \in [3] \\a_i \in \supp(\cD^{j_i})}}(C(G_{a_i}))\right)\delta+\left(C(H)\max_{\substack{i \in [3] \\a_i \in \supp(\cD^{j_i})}}(\beta(G_{a_i}))+\beta(H)\right),    
\end{align*}
which completes the proof of \eqref{eq:easy-recursion-main-C} and~\eqref{eq:easy-recursion-main-beta}.

\paragraph{Proving base cases when $r_i$'s are 0:} Without loss of generality assume that $|R_1| = 0$.
Recall that the graph $G = T(R_1,R_2,R_3;\cD)$ has one vertex $\emptyset$ in its first part that is connected to all the vertices $\supp(\mu^{R_2})$ and $\supp(\mu^{R_3})$ 
in the second and third parts named $P_2,P_3$. We can show that $C(G) \leq 1$ via the following algorithm to get a satisfying solution to $\Phi$. 
We assign the identity permutation to $\emptyset$, and then propagate this solution to all the vertices in parts $P_2$ and $P_3$. Now note that by definition all the edges between $\emptyset$ and $P_2$ or $P_3$ are satisfied. 
For an edge $(u,v)$ between $P_2,P_3$, it is violated only if the triangle $(\emptyset,u,v)$ is violated which happens only with probability $\delta$. Therefore we get a solution violating at most $\delta/3$-fraction of the edges of $G$, which shows that $C(G) \leq 1$. 
The same argument holds when more of the $R_i$'s have size $0$.

\paragraph{Concluding the induction:} Firstly note that \eqref{eq:easy-recursion-main-C},~\eqref{eq:easy-recursion-main-beta} imply that:
\begin{equation}\label{eq:recursion}
C_{r_1,r_2,r_3}(\mu) \lll C_{1,1,1}(\mu)C',~~~\beta_{r_1,r_2,r_3}(\mu) \lll C_{1,1,1}(\mu)\beta'+\beta_{1,1,1}(\mu)    
\end{equation}
where we take $C' = \max(C_{r_1-1,r_2,r_3}(\mu),C_{r_1,r_2-1,r_3}(\mu), C_{r_1,r_2,r_3-1}(\mu))$ and analogously pick $\beta'$ to be $\beta' = \max(\beta_{r_1-1,r_2,r_3}(\mu),\beta_{r_1,r_2-1,r_3}(\mu), \beta_{r_1,r_2,r_3-1}(\mu))$.

To prove the final statement we can then use induction on $\sum_{i=1}^3 r_i$. The base case of induction is when either $\sum r_i = 3$ and each $r_i=1$, or one of the $r_i$'s is $0$. In the first case $C_{r_1,r_2,r_3}(\mu) \leq C_{1,1,1}(\mu) \leq O(C_{1,1,1}(\mu))^3$, $\beta_{r_1,r_2,r_3}(\mu) \leq \beta_{1,1,1}(\mu) \leq O(C_{1,1,1}(\mu))^3\beta_{1,1,1}(\mu)$ and in the second case $C_{r_1,r_2,r_3} \leq 1 \leq O(C_{1,1,1}(\mu))^{\sum r_i},\beta_{r_1,r_2,r_3} = 0 \leq C_{1,1,1}(\mu)^{\sum r_i}\beta_{1,1,1}(\mu)$ too. Therefore now let us assume the lemma statement holds for all $(r_1,r_2,r_3)$ with $\sum r_i \leq t$ and let us prove it for $(r'_1,r'_2,r'_3)$ with $\sum r'_i = t+1$ and each $r'_i > 0$. Applying \eqref{eq:recursion} on $C_{r'_1,r'_2,r'_3}(\mu)$ we get that,
\begin{align*}
C_{r'_1,r'_2,r'_3}(\mu) \leq O(C_{1,1,1}(\mu))C' 
\leq O(C_{1,1,1}(\mu)) O(C_{1,1,1}(\mu))^{r'_1+r'_2+r'_3-1}
= O(C_{1,1,1}(\mu))^{r'_1+r'_2+r'_3},
\end{align*}
where in the second inequality we used the inductive bound on $C_{r'_1-1,r'_2,r'_3}(\mu), C_{r'_1,r'_2-1,r'_3}(\mu),C_{r'_1,r'_2,r'_3-1}(\mu)$. A similar argument for $\beta_{r'_1,r'_2,r'_3}(\mu)$ completes the induction thus proving the lemma.
\end{proof}

\subsection{Subexponential Bounds via Non-Lopsided Induction}\label{sec:exp-bounds}

Throughout this section we fix $k = r^{0.01}$, three $r$-sized pairwise disjoint sets $R_1,R_2,R_3 \subset [d]$, $I=\bigcup_{i \in [3]} R_i$ and $\mu$ over $\prod_{i \in I} X_i$ that satisfy the following assumptions:
\begin{assumption}\label{assumption1}
The following conditions hold:
\begin{enumerate}
\item The measure $\mu$ is an $\eps$-product distribution with $\eps \leq 2^{-r^{12}}$.
\item For all $k \leq r$, $S \subseteq I$ of size at most $3r - 3k$ and restrictions $a_0 \in \supp(\mu^S)$, for all $k$-sized pairwise disjoint sets $A,B,C$ of $I \setminus S$, such that 
\[
\max_{a\in A} a < \min_{b \in B} b\leq \max_{b \in B} b < \min_{c \in C} c,
\]
the graph $T(A,B,C;\mu|X_S = a_0)$ is a $(\poly(r),2^{-r^{12}})$-coboundary expander over $\S_m$.
\item For each interval $I_j = \left\{\frac{jdk^{10}}{r},\ldots,\frac{(j+1)dk^{10}}{r}\right\}$ and 
for each $i \in [3]$ we have that 
\[
k^{10}-k^6
\leq 
|R_i \cap I_j| 
\leq k^{10}+k^6.
\]
In words, the number of 
elements in $R_i$ in the 
interval $I_j$ is roughly
$k^{10}$ (which is the number of points 
a typical interval of that length has).
\end{enumerate}
\end{assumption}

The main result of this section is the following 
lemma, asserting that if $\mu$ satisfies Assumption~\ref{assumption1}, then the corresponding
tripartite graph is a coboundary expander with good parameters. More precisely:
\begin{lemma}\label{lem:exp-bound-mu}
Let $R_1,R_2,R_3\subseteq$ and let $\mu$ be a probability measure over $\prod_{i \in I} X_i$ satisfying Assumption~\ref{assumption1}. Then the graph $T(R_1,R_2,R_3;\mu)$ is a $(\poly(r)^{r/k},2^{-\Omega(r^{12})})$-coboundary expander over $\S_m$.
\end{lemma}

\begin{proof}
Let $I_j \subset [d]$ denote the interval $\left\{\frac{jdk^{10}}{r},\ldots,\frac{(j+1)dk^{10}}{r}\right\}$, and write $|R_i \cap I_j| = k^{10}+c_{i,j}$ where $c_{i,j} \in [-k^6,k^6]$; note that $\sum_j c_{i,j}=0$. We will use induction to prove the lemma. Throughout the induction, we will have three sets $R'_1,R'_2,R'_3$, satisfying:
\begin{enumerate}
    \item $R'_i \subseteq R_i$.
    \item For all $i,j$, $|R'_i \cap I_j| = r_{j}+c_{i,j}$ for some $r_j \leq k^{10}$.
    \item There exist at least three intervals with $r_j \geq k^9$, say $I_{j_1},I_{j_2},I_{j_3}$, with $j_1 < j_2 < j_3$.
\end{enumerate}
Initially, we will take 
$R_1' = R_1$, $R_2' = R_2$, 
$R_3' = R_3$, and later 
steps will take subsets of 
these. Note that in particular (2) above implies that all $R'_i$ are equal in size. Henceforth, we fix such $R_1', R_2', R_3'$. 
Note that by the second and third items, we may find  three $3k$-sized disjoint sets $S_1,S_2,S_3$ satisfying that $S_i \subset I_{j_i}$ for $i=1,2,3$ and furthermore for all $j \in [3]$, $|S_i \cap (R'_j \cap I_{j_i})| = k$, and we fix 
such $S_1, S_2, S_3$ henceforth. 

Consider a distribution $\cD = \mu | X_B=B_0$, where $B \subseteq I \setminus \bigcup_{i \in [3]} R'_i$ and $B_0 \in \supp(\mu^B)$.
The bulk of the argument will
be devoted to proving that
\begin{equation}\label{eq:intermediate-exp-C}
C(T(R'_1,R'_2,R'_3;\cD)) \leq \poly(r)\cdot \max_{\substack{i\in [3],\\a_i \in \supp(\cD^{S_i})}}(C(T(R'_1\setminus S_i,R'_2 \setminus S_i,R'_3\setminus S_i;\cD|X_{S_i}=a_i))),
\end{equation}
\begin{equation}\label{eq:intermediate-exp-beta}
\beta(T(R'_1,R'_2,R'_3;\cD)) \leq \poly(r)\cdot\max_{\substack{i\in [3],\\a_i \in \supp(\cD^{S_i})}}(\beta(T(R'_1\setminus S_i,R'_2 \setminus S_i,R'_3\setminus S_i;\cD|X_{S_i}=a_i))).    
\end{equation}
Once we establish these two inequalities, it is straightforward to conclude the lemma, and we do so in the end of the proof. 
Towards proving \eqref{eq:intermediate-exp-C} and \eqref{eq:intermediate-exp-beta}, fix any UG instance $\Phi$ on $G = T(R'_1,R'_2,R'_3;\cD)$ with $\delta$-fraction of inconsistent triangles. 

\paragraph{Setting up lists on restrictions:}
For every $S_i$ and every restriction $a \in \supp(\cD^{S_i})$, let $G_a$ be the subgraph $T(R'_1\setminus S_i,R'_2 \setminus S_i,R'_3\setminus S_i;\cD|X_{S_i} = a)$. Let $X_a \in \S_m^{V(G_a)}$ be an assignment on $V(G_a)$ with maximum value, and let $L_{a}$ be an ordered list of solutions indexed by permutations in $\S_m$, where $L_{a}[\pi] = X_a\circ \pi$. 

\paragraph{Setting up permutations between restrictions:}
For every $i \neq j$, and every restriction $a$ of $S_i$ and $b$ of $S_j$ where $(a,b) \in \supp(\cD^{S_i \cup S_j})$, let $G_{a,b}$ be the induced subgraph $T(R'_1, R'_2, R'_3;\cD | (X_{S_i} = a,X_{S_j} = b))$. Let $\viol(X_a;G_{a,b})$ denote the fraction of edges in $G_{a,b}$ that are violated by the assignment $X_a$.
Then by Claim~\ref{claim:shift-partition} there exists a permutation $\pi_{a,b}$ satisfying that
\begin{equation}\label{eq:bad-ij}
\Pr_{v \sim G_{a,b}}[X_a(v) \neq X_b(v)\pi_{a,b}] \leq \frac{\viol(X_a;G_{a,b})+\viol(X_b;G_{a,b})}{1-\lambda(G_{a,b})} \lll \viol(X_a;G_{a,b})+\viol(X_b;G_{a,b}),  
\end{equation}
where we used Claim~\ref{lem:eps-product-eigval} to bound $\lambda(G_{a,b})$ by $1/2+\poly(r)\eps < 0.51$. Let 
\[
\bad(a,b) = \{v\in G_{a,b}~|~X_a(v)\neq X_b(v)\pi_{a,b}\}.
\]
For any restriction $b$ as above, let $L_b\circ \pi$ denote the list where $L_b\circ \pi[\pi'] = L_b[\pi\pi'] = X_b\circ\pi\pi'$. With these notations and the definition of $\bad(a,b)$, we get that $L_a|_{G_{a,b} \cap \overline{\bad(a,b)}} = L_b\circ\pi_{a,b}|_{G_{a,b} \cap \overline{\bad(a,b)}}$. 

\paragraph{Counting bad triangles:}
Let $S = \bigcup_i S_i$. For any restrictions $a_1,a_2,a_3$ of $S_1,S_2,S_3$ let $G_{a_1,a_2,a_3}$ denote the graph $T(R'_1\setminus S, R'_2\setminus S, R'_3 \setminus S;\cD|(a_1,a_2,a_3))$ when $(a_1,a_2,a_3)$ is a valid restriction. Clearly, we have thats
\[\E_{\substack{a_3 \sim \cD^{S_3}|(a_1,a_2)}}[\mu_{G_{a_1,a_2,a_3}}(\bad(a_1,a_2))] = \mu_{G_{a_1,a_2}}(\bad(a_1,a_2)).\]
Fix $a_1,a_2$. Using Markov's inequality and~\eqref{eq:bad-ij} we get
\begin{equation}\label{eq:avg-viol}
\Pr_{\substack{a_3 \sim \cD^{S_3}|(a_1,a_2)}}\left[\mu_{G_{a_1,a_2,a_3}}(\bad(a_1,a_2)) > \frac{1}{100}\right]
\lll \viol(X_{a_1};G_{a_1,a_2})+\viol(X_{a_2};G_{a_1,a_2}).
\end{equation}

We say a restriction $(a_1,a_2,a_3)$ is bad if $\mu_{G_{a_1,a_2,a_3}}(\bad(a_i,a_j)) > 1/100$ for some $i \neq j$. By \eqref{eq:avg-viol} we get that,
\begin{align}\label{eq1}
&\E_{(a_1,a_2,a_3) \sim \cD^{\cup S_i}}[\Ind((a_1,a_2,a_3) \text{ is bad})]\notag\\
&\leq \sum_{i \in [3]} \E_{\substack{(a_j,a_k) \sim \cD^{S_j \cup S_k}\\ j,k\neq i}}\E_{a_i \sim \cD^{S_i}|(a_j,a_k)}[\Ind(\mu_{G_{a_1,a_2,a_3}}(\bad(a_j,a_k)) > \frac{1}{100})]\notag\\
&\lll \sum_{i \in [3]} \E_{\substack{(a_j,a_k) \sim \mu^{S_j \cup S_k}\\ j,k\neq i}}[\viol(X_{a_j};G_{a_j,a_k})+\viol(X_{a_k};G_{a_j,a_k})]\notag\\
&= \sum_{i \in [3]} \E_{\substack{a_i \sim \cD^{S_i}}}[\viol(X_{a_i};G_{a_i})],
\end{align}
where in the first inequality we used the union bound and in the second inequality we used \eqref{eq:avg-viol}. 
By definition we have that $\viol(X_{a_i};G_{a_i}) \leq C(G_{a_i})\eps_{a_i}+\beta(G_{a_i})$, where $\eps_{a_i}$ is the fraction of violated triangles in $G_{a_i}$. 
Therefore, 
\[
\eqref{eq1}
\leq
\sum_{i \in [3]} \E_{\substack{a_i \sim \cD^{S_i}}}[C(G_{a_i})\eps_{a_i}+\beta(G_{a_i})]
\lll \max_{\substack{i\in [3]\\a_i \in \supp(\cD^{S_i})}}(C(G_{a_i})\delta+\beta(G_{a_i})):=\delta''.
\]

\paragraph{Creating a UG instance on graph over restrictions:}
Let $\Psi$ be the following UG instance on $H = T(S_1,S_2,S_3;\cD)$: each edge $(a_i,a_j)$ has the constraint $\pi_{a_i,a_j}^{-1}$. That is, we want to find a solution $A$ that maximizes the fraction of edges satisfying $A(a_i) = \pi_{a_i,a_j}^{-1}A(a_j)$. Note that a triangle $(a_1,a_2,a_3)$ is consistent if $\pi_{a_3,a_1}^{-1}\pi_{a_1,a_2}^{-1}\pi_{a_2,a_3}^{-1} = \text{id}$ or equivalently if $\pi_{a_2,a_3}\pi_{a_1,a_2}\pi_{a_3,a_1} = \text{id}$.

First note that if a triangle $(a_1,a_2,a_3)$ is not bad, then it is consistent. The proof is the same as that in Lemma~\ref{lem:easy-recursion} hence we omit it here. Using the coboundary expansion of $H$, we get that there exists a UG solution $A$ to the vertices violating at most $C(H)\delta''+\beta(H)$-fraction of the edges. 

Since $S_i \subset I_{j_i}$ 
for all $i$, we get that $\max_{a \in S_1} a < 
\min_{b \in S_2} b$ and $\max_{b \in S_2} b < \min_{c \in S_3} c$. Therefore by Assumption~\ref{assumption1}, $C(H) \leq \poly(r)$ and $\beta(H) \leq 2^{-r^{12}}$. Therefore $\viol(A) \leq \poly(r)\delta''+2^{-r^{12}}:=\delta'$. 

\paragraph{Lifting the solution:} 
We will now use $A$ to create a highly satisfying solution $B$ to $G = T(R'_1,R'_2,R'_3;\cD)$. For every vertex $v\in \supp(\mu^{R'_i})$ and restriction $a \in \supp(\cD^{S_1}|(X_{R'_i} = v))$, let $g_a(v)$ denote the permutation $L_{a}[A(a)][v]$. We will choose a randomized assignment as follows: to every vertex $u \in \supp(\mu^{R'_i})$, choose a random $s \sim \cD^{S_1}|(X_{R'_i} = u)$ and assign $B(u)= g_{s}(u)$. We will now upper bound the expected fraction of edges that $B$ violates for $\Phi$. 

Consider an edge $(u_i,u_j) \in G$ between parts $R'_i,R'_j$. This edge is satisfied if there exists $s' \in \supp(\cD^{S_1}|(X_{R'_i} = u_i, X_{R'_j} = u_j))$ such that, (1) $B(u_i) = g_{s'}(u_i)$, (2) $B(u_j) = g_{s'}(u_j)$ and (3) The assignment $(g_{s'}(u_i),g_{s'}(u_j))$ satisfies the edge $(u_i,u_j)$ or equivalently $(u_i,u_j)$ is satisfied by the assignment $X_{s'}$. 
To evaluate the probability there is such $s'$, we 
sample $s'\sim \cD^{S_1}|(X_{R'_i} = u_i, X_{R'_j} = u_j)$ and consider each one of the events, starting with event (1).
For it, the probability it doesn't hold is at most
\[\E_B\E_{(u_i,u_j) \sim G}\E_{s' \sim \cD^{S_1}|(X_{R'_i} = u_i, X_{R'_j} = u_j)}[\Ind(B(u_i) \neq g_{s'}(u_i))]
=\E_{u \in G}\E_{s,s' \sim \cD^{S_1}|u}[\Ind(g_{s'}(u) \neq g_s(u))]. 
\]
To calculate this, let us first calculate a bound on:
\[\E_{u \in G}\E_{\substack{i\neq j \in [3]\\(s,s') \sim \cD^{S_i \cup S_j}|u}}[\Ind(g_{s'}(u_i) \neq g_s(u_i))].\]
Fix a vertex $u \in R'_i$ for some $i$. It is easy to check that an if an edge $(s,s') \in T(S_1,S_2,S_3;\cD|u)$ 
is satisfied by $A$ and $u\not\in \bad(s,s')$, then $g_s(u) = g_{s'}(u)$. Thus,
\begin{align*}
&\E_{u \in G}\E_{\substack{i\neq j \in [3]\\(s,s') \sim \cD^{S_i \cup S_j}|u}}[\Ind(g_{s'}(u) \neq g_s(u))]\\
&\leq \E_{u \in G}\E_{\substack{i\neq j \in [3]\\(s,s') \sim \cD^{S_i \cup S_j}|u}}[\Ind((s,s') \text{ not satisfied by $A$})]+\E_{u \in G}\E_{\substack{i\neq j \in [3]\\(s,s') \sim \cD^{S_i \cup S_j}|u}}[\Ind(u \in \bad(s,s'))]\\
&= \viol(A) + \E_{\substack{(s,s') \sim E(T(S_1,S_2,S_3;\cD))\\ u \in G_{s,s'}}}[\Ind(u \in \bad(s,s'))]\\
&=\delta' + 2\E_{u \in T(S_1,S_2,S_3;\cD)}[\viol(X_u)]\\
&\lll \delta'.
\end{align*}

We are ready to bound the probability that event (1) does not happen. Towards this end, for each vertex $u$ define $p_u := \Pr_{\substack{(s,s') \sim \cD^{S_1 \cup S_2}|u}}[g_{s'}(u) \neq g_s(u)]$, so that the above inequality translates to $\E_{u \in G}[p_u] \lll \delta'$. Let $p'_u = \Pr_{\substack{s,s' \sim \cD^{S_1}|u}}[g_{s'}(u) \neq g_s(u)]$. By Lemma~\ref{claim:gll-sing-val}, the second largest singular value of the bipartite graph $A(S_1,S_2;\cD|u)$ is at most $\eps\cdot\poly(r) \leq 0.01$ for all $u$, so by the easy direction of Cheeger's inequality we get that, $p'_u \leq O(p_u)$, which gives us that,
\[\E_{u \in G}\E_{s,s' \sim \cD^{S_1}|u}[\Ind(g_{s'}(u) \neq g_s(u))] \lll \delta',\]
thus bounding the probability of event (1). One can check that the probability of event (2) is the same as event (1), hence let us proceed to event (3). For that we get,
\[\E_{(u_i,u_j) \sim G}\E_{s \sim \cD^{S_1}|(X_{R'_i} = u_i, X_{R'_j} = u_j)}[\Ind((u_i,u_j) \in \viol(X_s))] \lll 
\delta'.
\]
We see that sampling an edge $(u_i,u_j)$, $s$ and $s'$ fails to satisfy at least one of the events (1), (2) and (3) with probability at most $O(\delta')$. Thus, 
with probability at most $O(\delta')$ over the choice of $(u_i,u_j)$ and $s$, there is no $s'$ like that and otherwise we get that 
$s'$ satisfies all of (1), (2) and (3).
This shows that in expectation the assignment $B$ that we get violates at most $O(\delta')
\leq \poly(r)\max_{i\in [3]\\a_i \in \supp(\cD^{S_i})}(C(G_{a_i})\delta+\beta(G_{a_i}))$ fraction of the edges of $G$, which completes the proof of \eqref{eq:intermediate-exp-C},~\eqref{eq:intermediate-exp-beta}.

\paragraph{Concluding the induction:} Let us now use \eqref{eq:intermediate-exp-C},~\eqref{eq:intermediate-exp-beta} to prove the lemma. For any sets $R'_1,R'_2,R'_3$ satisfying the first two items in our assumption about $R'_i$'s
above with $|R'_i| = \ell k$ and distribution $\cD$ as above we will show by induction that $C(T(R'_1,R'_2,R'_3;\cD)) \leq \poly(r)^{\ell}\cdot \poly(r)^{r/k}$ and $\beta(T(R'_1,R'_2,R'_3;\cD)) \leq 2^{-\Omega(r^{12})}$. Let us prove the base case first.  When $\ell=1$, we can use Lemma~\ref{lem:easy-recursion} to bound $C(T(R'_1,R'_2,R'_3;\cD))$ by $O(C_{1,1,1}(\mu))^{3k} \leq \poly(r)^k \leq \poly(r)^{r/k}$, since $C_{1,1,1}(\mu) \leq \poly(r)$ by Assumption~\ref{assumption1}. Our next base case is when the sets $R'_i$'s do not satisfy the third item in the assumptions about $R'_i$'s. That is, for all but at most two intervals $I_j$, $|R'_i \cap I_j| \leq k^9+c_{i,j} \leq k^9+k^6$. This implies that $|R'_i| \leq (r/k^{10})\cdot(k^9+k^6)+2\cdot (k^{10}+k^6) \lll r/k$. In this case, again using Lemma~\ref{lem:easy-recursion} we get that $C(T(R'_1,R'_2,R'_3;\cD)) \leq \poly(r)^{r/k}$ as required.

Let us now show the inductive step. Assume that we have proved the statement for $\ell$ and let us prove it for sets $R'_i$'s satisfying assumptions (1), (2) and (3) with $|R'_i| =(\ell+1)k$. Applying \eqref{eq:intermediate-exp-C} we get that there are $3k$-sized $S_i$'s so that,
\begin{equation*}
C(T(R'_1,R'_2,R'_3;\cD)) \leq \poly(r)\cdot \max_{\substack{i\in [3],\\a_i \in \supp(\cD^{S_i})}}(C(T(R'_1\setminus S_i,R'_2 \setminus S_i,R'_3\setminus S_i;\cD|X_{S_i}=a_i))).
\end{equation*}
One can now check that the sets $(R'_1 \setminus S_1, R'_2 \setminus S_1,R'_3 \setminus S_1)$ satisfy assumptions (1) and (2), therefore $C(T(R'_1 \setminus S_1, R'_2 \setminus S_1,R'_3 \setminus S_1;\cD|a_1)) \leq \poly(r)^\ell \cdot \poly(r)^{r/k}$ by the induction hypothesis. The same holds for $S_2$ and $S_3$, thus giving us that, $C(T(R'_1,R'_2,R'_3;\cD)) \leq \poly(r)^{\ell+1}\poly(r)^{r/k}$ as required. This implies that for $\ell=r/k$, we get that $C(T(R_1,R_2,R_3;\mu))\leq \poly(r)^{r/k}$.

Applying a similar argument on $\beta$ gives that $\beta(T(R_1,R_2,R_3;\mu)) \leq 2^{-\Omega(r^{12})}$.
\end{proof}

\section{Base Case and Extended Base Case for Spherical Buildings}
In this section we establish bounds on the 
coboundary constants that serve as the base 
case presented in this section. The discussion
will be specialized to spherical buildings of 
type A, type C (for technical reasons, in the end
we will also have to discuss tensors of two 
types, but this will be straightforward).

To establish the base case of our induction we will
use the cones \dor{method~\cite{gromov2010singularities,lubotzky2016expansion,KaufmanMass2,kozlov2019quantitative,kaufman2019coboundary,DiksteinD23,ChapmanL,dikstein2023swap}}. The cones method allows one to prove 
coboundary expansion properties for a graph 
given an efficient way \emph{triangulate} certain cycles in our graphs of interest.
\begin{definition}\label{def:triangulation}
A triangulation of a cycle $C$ in a graph $G$ is given by a set of vertices $S_C \subseteq V(G)$ such that the induced subgraph on $S_C \cup V(C)$ breaks into a union of triangles, where any two distinct triangles are either disjoint or share an edge.
\end{definition}

\paragraph{Notations:} throughout this section,
for subspaces $A,B$, we will often use the notation
$(A,B)$ to denote the subspace $A+B$. Similarly,
for a vector $v$ and a subspace $A$, we will denote
by $(v,A)$ the subspace $\spn(\{v\}) + A$.

\subsection{The Base Case for Spherical Buildings of Type A}\label{sec:grass-base-case}
Let $\Gr_d(k_1,k_2,k_3)$ denote the tripartite graph whose vertices are $k_1,k_2,k_3$-dimensional subspaces of $\F_q^d$ and an edge is sampled by sampling a random chain $V_1 \subset V_2 \subset V_3$ and choosing $(V_i,V_j), i \neq j \sim [3]$. We drop the subscript $d$ when clear from context.

\begin{lemma}\label{lem:base-grassmann}
For all $m \in \N$, $k_1<k_2<k_3$ integers, let $K = \max(\lceil\frac{k_2}{k_3-k_2}\rceil, \lceil\frac{k_2}{k_2-k_1}\rceil)$. Then $\Gr_d(k_1,k_2,k_3)$ is an $(O(K^2),\poly(K)/q)$-coboundary expander over $\S_m$.
\end{lemma}
\mitali{We remark that Lemma~\ref{lem:base-grassmann} is a variant of~\cite[Lemma 8.10]{dikstein2023swap}. They bound the coboundary constant of a $4$-partite graph $\Gr_d(k_1,k_2,k_3,k_4)$ where analogous to the above, each part contains $k_i$-dimensional subspaces of $\F_q^d$. We incur an additive error of $\poly(K)/q$ whereas they don't have any additive error. On the other hand our multiplicative constant is polynomial in $K$ whereas theirs is quasi-polynomial in a related parameter. Thus their bound is better in the low error regime and ours is better in the high error regime, though this has little effect on the final result. As far as we know, the proof we give herein is different and provides a good starting point for the analogous case of spherical buildings of type C, which is more complicated.}

We break the proof into two cases according to which term larger, $k_2-k_1$ or $k_3-k_2$. Let us first discuss the case when $k_3-k_2 \leq k_2-k_1$ and then later we discuss how to modify the proof in the other case. For simplicity of notation we will assume that $k_1,k_2,k_3$ are all multiples of $k_3-k_2$, so set $k = k_3-k_2$, $t = k_3/k$ and $t' = k_1/k$. The same proof with some slight modifications works when this is not the case, hence we omit those details.

Let $G$ denote the graph $\Gr_d(k_1,k_2,k_3)$ and $G_i$ denote the vertices of dimension $k_i$. First fix an arbitrary vertex $U \in G_3$ and an arbitrary basis for it: $u_1,\ldots, u_{k_3}$. Let $U_{(1)}$ denote the set of first $k$ vectors, $\{u_1,\ldots, u_{k}\}$, $U_{(2)}$ the second set of $k$ vectors and so on upto $U_{(t)}$. We now fix a set of paths from $U$ to $V$ for most $V \in G$. 
Let $\cB$ be a set of ``block decompositions'' for each subspace $V \in G$, i.e. $\cB$ assigns $V \in G_i$ the blocks, $\cB(V) = (V_{(1)},\ldots,V_{(t)})$ with $V_{(i)}= U_{(i)}$ for all $i \leq t-k_i/k$ and $V = (V_{>t-k_i/k})$.

\subsubsection{Set of Good Vertices and Edges with respect to $\cB$ when $k_3-k_2\leq k_2-k_1$}
\begin{enumerate}
\item Let $\good_i(\cB) \subseteq G_i$ be the set of $V \in G_i$ that satisfy for all $i \in [t], \dim(V_{(1)},\ldots,V_{(i)},U_{>i}) = k_3$. In particular taking $i=t$ this implies that $\dim(U_{\leq t-k_i/k},V) = k_3$.
\item For all $b < a \in [3]$ let $\good_{ab}(\cB) \subseteq E(G_a,G_b)$ be the set of edges $(V,W)$ with $W \subset V$  that satisfy
\begin{enumerate}
\item Both $V,W \in \bigcup_{a' \in [3]} \good_{a'}(\cB)$.
\item For all $j < i \in [t], \dim(V_{\leq j},W_{(j+1)},\ldots,W_{(i)},U_{>i}) = k_3.$
\end{enumerate}
\end{enumerate}

\begin{claim}\label{claim:random-subspace-intersection}
Let $V_1,V_2\subseteq\F_q^d$ be $d_1,d_2$ dimensional subspaces respectively, and suppose that $d_3+d_1 \leq d_2$. Then
\[\Pr_{V_3 \subseteq_{d_3} V_2}[\dim(V_3+V_1) < d_3+d_1] \leq \frac{1}{q-1}.\]
\end{claim}
\begin{proof}
Let $U = V_1 \cap V_2$ and write $V = U + V'_1$ with $V'_1 \cap V_2 = \{0\}$. Then:
\[\Pr_{V_3 \subset_{d_3} V_2}[V_3 \cap U = \{0\}] \geq \prod_{i =0}^{d_3-1}\left(\frac{q^{d_2}-q^{d_1+i}}{q^{d_2}}\right) \geq 1-\frac{1}{q-1}.\]
If $V_3 \cap U =\{0\}$ then $V_3\cap V_1 = \{0\}$ too, hence $\dim(V_3+V_1)=d_3+d_1$ thus completing the proof.
\end{proof}

\begin{lemma}\label{lem:random-block-decomp}
There exists a block decomposition $\cB$ such that for all $b < a \in [3]$: 
\[\Pr_{V \sim G_a}[V \in \good_a(\cB)] \geq 1-O\left(\frac{t}{q}\right)\]
\[\Pr_{\substack{V \sim G_a\\W \subset_b V}}[(V,W) \in \good_{ab}(\cB)] \geq 1-\frac{\poly(t)}{q}.\]
\end{lemma}

\begin{proof}
Let us consider a random block decomposition $\cB$, that for each subspace $V \in G_a$, picks a sequence of random subspaces:
\[V_{(t-k_a/k+1)} \subset_k (V_{(t-k_a/k+1)}, V_{(t-k_a/k+1)}) \subset_{2k} \ldots \subset_{k_a-k} V,\]
and sets $V_{(i)}=U_{(i)}$ for all $i \leq t-k_a/k$. Let us prove that with high probability $V \in \good_a$. Fix some $i \in [t]$ and let $F \subset [t]$ be the set $\{1,\ldots,t-k_a/k\}\cup\{i+1,\ldots,t\}$. For a set $S \subset [t]$, let $U_S$ denote the subspace $\spn(U_{(i)} \mid i \in S)$. Then using Claim~\ref{claim:random-subspace-intersection} we get,
\[\Pr_{V \sim G_a,\cB}[\dim(V_{\leq i},U_{>i})=k_3]=\Pr_{V \sim G_a,\cB}[\dim(V_{\overline{F}},U_{F})=k_3] \geq 1-\frac{1}{q-1},\]
where we used that $V_{\overline{F}}$ is distributed as a uniformly random $k(t-|F|)$-dimensional subspace. Taking a union bound over $i \in [t]$ we get that,
\begin{equation}\label{eq:good-verts-B}
\Pr_{V \sim G_a,\cB}[V \in \good_a(\cB)] \geq 1-\frac{t}{q-1},
\end{equation}
establishing the first item. For the second item, 
fix $b <a \in [3]$, $j<i \in [t]$, 
and let 
\[
F = [t-k_a/k] \cup \{i+1,\ldots,t\} \cup ([t-k_b/k]\cap\{j+1,\ldots,i\}),
~~
R_1=[j]\setminus F,
~~ 
R_2 = \{j+1,\ldots,i\}\setminus F.
\] 
We have:
\begin{align}
&\Pr_{\substack{W \subset_{k_1} V, \cB}}[\dim(V_{\leq j}, W_{(j+1)},\ldots,W_{(i)},U_{>i})=k_3]\notag \\
&=\Pr_{\substack{W \subset V, \cB}}[\dim(V_{R_1}, W_{R_2},U_{F})=k_3]\notag \\
&=\Pr_{\substack{V,\cB}}[\dim(V_{R_1},U_{F})=(|R_1|+|F|)k]\notag\\
&\cdot\Pr_{\substack{W \subset_{k_2} V,\cB}}[\dim(V_{R_1},W_{R_2},U_{F})=k_3 \mid\dim(V_{R_1},U_{F})=(|R_1|+|F|)k].\label{eq:random-grassmann}
\end{align}
The first term is at least $1-\frac{1}{q-1}$ by Claim~\ref{claim:random-subspace-intersection}, 
and we next lower bound the second term. 
By symmetry, it suffices to bound it for a fixed $V_{R_1}$ which satisfies $\dim(V_{R_1},U_{F})=(|R_1|+|F|)k$. We will first show that the distribution over $W_{R_2}$ conditioned on $V_{R_1}$ is $O(1/q)$-close to uniform. Indeed, by Claim~\ref{claim:random-subspace-intersection} with probability $1-O(1/q)$ 
we have that $W \cap T = \{0\}$. Conditioned on this, $W_{R_2}$ is a uniformly random $|R_2|k$-dimensional subspace amongst $|R_2|k$-dimensional subspaces intersecting $V_{R_1}$ at $\{0\}$. The latter distribution is $O(1/q)$-close to a uniformly random $|R_2|k$-dimensional subspace, as by Claim~\ref{claim:random-subspace-intersection} the probability a random $|R_2|k$-dimensional subspace intersects $V_{R_1}$ at $\{0\}$ is $1-O(1/q)$. When $W_{R_2}$ is chosen to be a uniformly random $|R_2|k$-dimensional subspace we can easily bound the probability that $\dim(V_{R_1},W_{R_2},U_F) = k_3$ using Claim~\ref{claim:random-subspace-intersection}.
Therefore we get,
\[\Pr_{\substack{W \subset_{k_2} V,\cB}}[\dim(V_{R_1},W_{R_2},U_{F})=k_3 \mid\dim(V_{R_1},U_{F})=(|R_1|+|F|)k] \geq 1-O(1/q).\]
Plugging this back into \eqref{eq:random-grassmann} and taking a union bound over all $j<i\in [t]$ we get,
\begin{equation}\label{eq:good-edges-B}
\Pr_{(V,W)\sim E(G_a,G_b),\cB}[(V,W) \in \good_{ab}(\cB)] \geq 1-\frac{\poly(t)}{q}.    
\end{equation}
An averaging argument on \eqref{eq:good-verts-B} and \eqref{eq:good-edges-B} now gives us that there is a $\cB$ for which most vertices and edges are good as required.
\end{proof}

Henceforth fix a block decomposition $\cB$ satisfying the conclusions of Lemma~\ref{lem:random-block-decomp}.

\subsubsection{Constructing the Paths when $k_3-k_2\leq k_2-k_1$}
In this section, fixing $U$,
we construct collections of canonical
paths between $U$ and good vertices in our graph.
For each $V\in \bigcup_{i \in [3]}\good_i(\cB)$, 
this is achieved using the block decomposition of
$V$ as follows:
\begin{align*}
P(U,V) = &(U_{(1)},\ldots, U_{(t)}) \rightarrow (U_{(2)},\ldots,U_{(t)}) \rightarrow (V_{(1)},U_{(2)},\ldots, U_{(t)}) \\
&\rightarrow (V_{(1)},U_{(3)},\ldots,U_{(t)}) \rightarrow (V_{(1)},V_{(2)},U_{(3)},\ldots,U_{(t)}) \rightarrow \ldots \\
&\rightarrow (V_{(1)},\ldots,V_{(t-1)},U_{(t)}) \rightarrow (V_{(1)},\ldots,V_{(t-1)}) \rightarrow (V_{(1)},\ldots,V_{(t)}) \rightarrow V,    
\end{align*}
where we omit the last step if $V \in G_3$ since $V = (V_{(1)},\ldots,V_{(t)})$. Note that since $V \in \good_i$, this path alternates between vertices from $S_3$ and $S_2$. When $V \in G_1$ or $G_2$, the subspace $U$ appears more than once on the path, and the path from $U$ to $V$ could be shortened. We use this longer path instead to keep things notationally simpler, since now all paths are of length either $2t$ or $2t+1$.

\subsubsection{Triangulating Cycles when $k_3-k_2\leq k_2-k_1$}
Having constructed paths from $U$ to all good 
vertices, we notice that if $(W,V)$ is an edge
in the graph, then together with the paths from $U$
we have a cycle $C(U,V,W)$. In the following lemma,
we show how to triangulate this cycle using a small
number of triangles.
\begin{lemma}\label{lem:triangulation-gr1}
When $k_3-k_2 \leq k_2-k_1$, for every edge $(V,W) \in \bigcup_{b<a \in [3]}\good_{ab}(\cB)$, the cycle $C(U,V,W) = U \xrightarrow[]{P(U, V)} V \rightarrow W \xrightarrow[]{P(W, U)} U$ has a triangulation of size $O((\frac{k_3}{k_3-k_2})^2)$. 
\end{lemma}

\begin{proof}
To make notation simpler we will give a triangulation $T(U,V,W)$ of the cycle with some repeating vertices. We put in edges between two vertices that are the same, and label them with the identity permutation. These are called ``equality edges''. This introduces triangles that might have at least two identical vertices, but it is easy to see that such a triangle is consistent, hence we can use these triangles essentially for free. 
 
\paragraph{Tiling by 8-cycles when $V \in G_3$ and $W \in G_2$:} Let $W' = (W_{(1)},\ldots,W_{(t)})$ with $W_{(1)}=U_{(1)}$, $W = (W_{\geq 2})$, and $V = (V_{(1)},\ldots, V_{(t)})$. To get a triangulation we first create paths between the $(2i)^{th}$ vertex on the path $P(U,V)$ and the $(2i)^{th}$ vertex on $P(U,W)$ for all $i \in \{1,\ldots,t\}$. Recall that $P_{2i}(U,V) = (V_{\leq i}, U_{>i})$ and $P_{2i}(U,W) = (W_{\leq i}, U_{>i})$. For $i \in [t]$ we take the obvious path $R^i(U,V,W)$ between $P_{2i}(U,W)$ and $P_{2i}(U,V)$ that flips a block of $W$ to a block of $V$ one at a time:
\begin{align*}
R^i(U,V,W) := &P_{2i}(U,W) \rightarrow (W_{(2)},\ldots, W_{(i)},U_{>i}) 
\rightarrow (V_{(1)},W_{(2)},\ldots, W_{(i)}, U_{>i}) \rightarrow \ldots \rightarrow P_{2i}(U,V),
\end{align*}
that alternates between $k_3$ and $k_2$ dimensional vertices. Here we used the fact that $(V,W)\in\good_{32}(\cB)$ which implies that the intermediate odd vertices -- $(V_{\leq j},W_{(j+1)},\ldots,W_{(i)},U_{>i})$, on the path $R^i$ above are of dimension $k_3$. 

Let us henceforth drop the notation $U,V,W$ in $R^i(U,V,W)$ since $U,V,W$ are fixed. First note that by creating the paths $R^{i}$, we have broken the original cycle $C(U,V,W)$ into $O(t)$ cycles of the form:     
\begin{align*}
C^i = P_{2i}(U,W) \xrightarrow[]{R^{i}} P_{2i}(U,V) \rightarrow P_{2i+1}(U,V) \rightarrow P_{2i+2}(U,V) \xrightarrow[]{R^{i+1}} P_{2i+2}(U,W) 
&\rightarrow P_{2i+1}(U,W)\\
&\rightarrow P_{2i}(U,W),
\end{align*}
for $i \in [t-1]$ and 
\[C^{t} = W' = (W_{(1)},\ldots,W_{(t)}) \xrightarrow[]{R^{t}} V=(V_{(1)},\ldots,V_{(t)}) \rightarrow W \rightarrow V.\] 

We will tile each $C^i$ by 8-cycles and triangles, starting with $i \in [t-1]$. For all $j \in [1,i+1]$, let $R^i_j$ denote the $(2j-1)^{th}$ vertex on the path $R^i$, that is, 
\[R^i_j = (V_{<j},W_{(j)},\ldots, W_{(i)},U_{> i}).\] 
It is easy to see that for all $j \in [1,i+1]$, the vertices $R^i_{j}$ and $R^{i+1}_j$ are connected via a path of length two: 
\[R^i_j \rightarrow (V_{<j},W_{(j)},\ldots, W_{(i)},U_{\geq i+2},\ldots, U_{(t)}) \rightarrow R^{i+1}_j.\] 
As for the last part of the cycle, the vertex $P_{2i+1}(U,V)$, occurs as the middle vertex in all the following length 2 paths: 1) $R^i_{i+1} = P_{2i}(U,V) \rightarrow R^{i+1}_{i+1}$, 2)$P_{2i}(U,V) \rightarrow P_{2i+2}(U,V)$ and $R^{i+1}_{i+1} \rightarrow R^{i+1}_{i+2}$. Thus using equality edges between these vertices we can check that the whole cycle has been broken into $O(t)$ 8-cycles and $O(1)$ triangles (see figure below). 
\begin{center}
\includegraphics[scale=0.20]{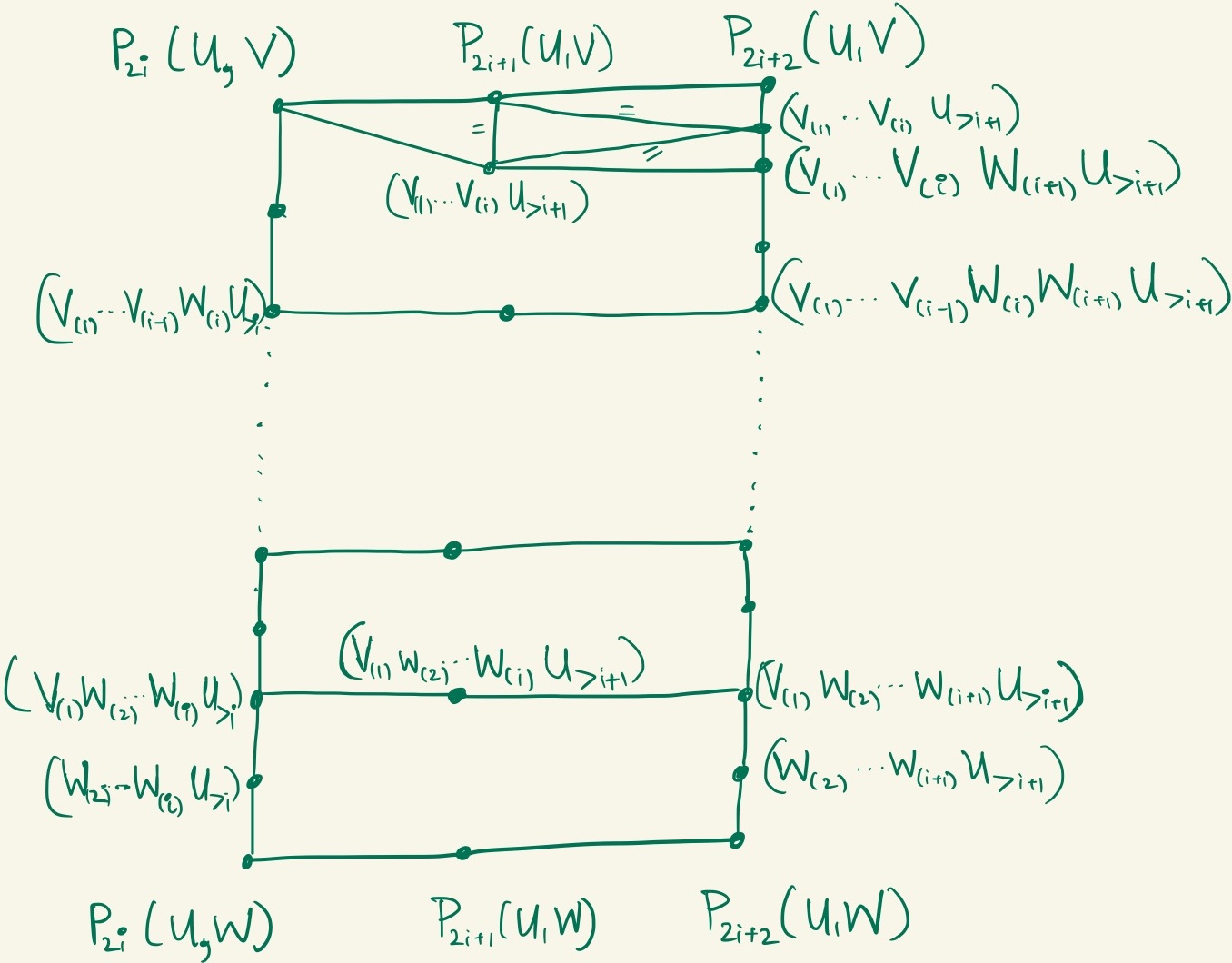}
\end{center}

Let us now tile the cycle $C^{t}$. We have that the first three vertices of $R^t$ are $W'$, $(W_{\geq 2})$ and $(V_{(1)},W_{\geq 2})$ which are all connected to $W$. Additionally including $(V_{(1)},W_{\geq 2})$, every other odd vertex on $R^t$, i.e. for all $j \in [2,t+1]$, $R^{t}_j = (V_{<j},W_{(j)},\ldots, W_{(t)})$
is contained (as a subspace) in $V$ since both $W$ and $V'$ are in $V$. Since the dimensions match it means that these are all equal to $V$ and therefore we can put in equality edges between them. Similarly every intermediate vertex ($k_2$-dimensional) on $R^t$ is contained in $V$. This means that $C^t$ is broken into $O(t)$ triangles. So overall $C(U,V,W)$ has been broken into $O(t^2)$ 8-cycles and $O(t)$ triangles.

\paragraph{Tiling by 8-cycles when $V \in G_3$ and $W \in G_1$:} Let $W' = (W_{(1)},\ldots,W_{(t)})$, where $W = (W_{>t-k_1/k+1}) \subset V$. Now we create the same paths $R^i$ of length $2i+1$ between $P_{2i}(U,W) \rightarrow P_{2i}(U,V)$ for all $i \in [t]$, and each vertex on these paths is in $S_3$ or $S_2$ because $(V,W) \in \good_{31}(\cB)$. This breaks the cycle into the cycles $C^i$, for $i \in [t]$. The tiling of $C^i, i \in [t-1]$ proceeds identical to the first case, therefore let us discuss the tiling of the last cycle: 
\[C^{t} = W' = (W_{(1)},\ldots,W_{(t)}) \xrightarrow[]{R^{t}} V = (V_{(1)},\ldots,V_{(t)}) \rightarrow W \rightarrow W'.\] 
As in the case above, we have that the first $2(t-k_1/k)+1$ vertices of $R^t$ are all connected to $W$, since they are of the form $(V_{<j},U_{(j)},\ldots,U_{(t-k_1/k)},W)$. 
After that all subsequent odd vertices on $R^t$, i.e. for all $j \in [t-k_1/k+1,t]$, $R^{t}_j = (V_{<j},W_{(j)},\ldots, W_{(t)})$ are connected to $V$. 
This means that $C^t$ is broken into $O(t)$ triangles. So overall $C(U,V,W)$ has been broken into $O(t^2)$ 8-cycles and $O(t)$ triangles.

\paragraph{$(V,W)$ is an edge between $G_2$ and $G_1$:} We can show that $C(U,V,W)$ in this case too can be broken into $O(t^2)$ 8-cycles and $O(t)$ triangles. The proof for the cycles $C^{i}$ for $i \in [t-1]$ is the same, so we only discuss the tiling of $C^t$. We have that,
\[C^{t} = W' = (W_{(1)},\ldots,W_{(t)}) \xrightarrow[]{R^{t}} V' = (V_{(1)},\ldots,V_{(t)}) \rightarrow V \rightarrow W \rightarrow W',\] 
where $W = (W_{>t-k_1/k}), V = (V_{>t-k_2/k})$. Again we have that the first $2(t-k_1/k)+1$ vertices of $R^t$ are all connected to $W$, since they are of the form $(V_{<j},U_{(j)},\ldots,U_{(t-k_1/k)},W)$. Then the vertex $(V_{\leq t-k_1/k},W)$ is additionally connected to $V',V$. After that all subsequent vertices on $R^t$, i.e. for all $j \in [t-k_1/k+1,t-1]$, $R^{t}_j = (V_{\leq j},W_{(j)},\ldots, W_{(t)})$ and the intermediate vertices in $G_1$ are all connected to $V'$, by the same reasoning as the above two cases. 
This means that $C^t$ is broken into $O(t)$ triangles. 

\paragraph{Triangulating the 8-cycles:} Having shown that the cycle $C(U,V,W)$ can always be tiled by at most $O(t^2)$ 8-cycles and $O(t)$ triangles, it suffices to show that each one of the resulting 8-cycles can be triangulated individually. Towards this end, we first notice that the $8$-cycles we formed consist of edges between subspaces of dimension $k_3$ and $k_2$. Thus, to triangulate them we will have
to use auxiliary vertices of dimension $k_1$.

Note that each $8$-cycle in the above tiling is of the following form for some $1 \leq j < i \in [t]$:
\begin{align*}
(W_{(j)},U_{(i)},X) \rightarrow (W_{(j)},X) \rightarrow (W_{(j)},W_{(i)},X) \rightarrow (W_{(i)},X)\notag \\
\rightarrow (V_{(j)},W_{(i)},X) \rightarrow (V_{(j)},X) \rightarrow (V_{(j)},U_{(i)},X) \rightarrow (U_{(i)},X) \rightarrow (W_{(j)},U_{(i)},X),
\end{align*}
with $X = (V_{\leq j-1},W_{(j+1)},\ldots,W_{(i-1)},U_{(> i)})$. We know that $\dim(X) = k_2-k = 2k_2-k_3$ and since $k_2-k_1 \geq k_3-k_2$, $\dim(X) \geq k_1$. This implies that there is some $k_1$-dimensional subspace $Y \subseteq X$ that is contained within all the vertices of this 8-cycle. Putting this vertex in the middle of the 8-cycle completes the triangulation of this cycle. 
\paragraph{Size of Triangulation:} In total we used  $O((k_3/k_3-k_2)^2)$ 8-cycles, each of which used $O(1)$ triangles thus completing the proof.
\end{proof}

\subsubsection{The Case that $k_3-k_2 > k_2-k_1$}
We now move on to discuss the case that $k_2-k_1\leq k_3-k_2$ and let $k=k_2-k_1, t=k_2/k$ henceforth. We will need to use a different set of paths, and towards this end we fix an arbitrary vertex $U \in G_2$ and an arbitrary basis for it: $u_1,\ldots, u_{k_2}$. Let $U_{(1)}$ denote the set of first $k$ vectors, $\{u_1,\ldots, u_{k}\}$, $U_{(2)}$ the second set of $k$ vectors and so on up to $U_{(t)}$. We now fix a set of paths from $U$ to $V$ for most $V \in G$. We will choose a block decomposition $\cB$ for each subspace $V \in G$, 
namely:
\begin{enumerate}
\item For $V \in G_2$, $\cB(V) = (V_{(1)},\ldots,V_{(t)}) = V$. 
\item For every $V \in G_1$, let $V' = (V_{(1)},\ldots,V_{(t)})$ for $V_{(1)}= U_{(1)}$ and $V = (V_{\geq 2})$.
\item For $V\in G_3$, the decomposition $\cB$ first associates with $V$ a $k_2$-dimensional subspace $V'\subseteq V$ along with its block decomposition $V' = (V_{(1)},\ldots,V_{(t)})$.
\end{enumerate}
As before, we will need the block decomposition $\cB$ 
to satisfy several genericness properties, that we explain next.
\paragraph{Set of Good Vertices and Edges with respect to $\cB$ when $k_3-k_2 > k_2-k_1$:}
\begin{enumerate}
\item Let $\good_i(\cB) \subseteq G_i$ be the set of $V \in G_i$ that satisfy for all $i \in [t], \dim(V_{(1)},\ldots,V_{(i)},U_{>i}) = k_2$. 
\item For all $b<a \in [3]$ let $\good_{ab}(\cB) \subseteq E(G_a,G_b)$ be the set of edges $(V,W)$ with $W \subset_b V$ that satisfy,
\begin{enumerate}
\item Both $V,W \in \bigcup_{a \in [3]} \good_a(\cB)$.
\item For all $j < i \in [t], \dim(V_{\leq j},W_{(j+1)},\ldots,W_{(i)},U_{>i}) = k_2.$
\end{enumerate}
\end{enumerate}

Analogously to Lemma~\ref{lem:random-block-decomp} one can show that there is a $\cB$ for which $1-\poly(t)/q$-fraction of vertices and edges are good, and we fix such 
$\cB$ henceforth.

\subsubsection{Constructions the Paths when $k_3-k_2 > k_2-k_1$} 
For a vertex $V \in \bigcup_{i \in [3]}\good_i(\cB)$, we first construct a path $P(U,V')$, where $V' = (V_{(1)},\ldots,V_{(t)})$, by flipping a block of $U$ to a block of $V'$ one at a time. One can check that this path alternates between $k_2$ and $k_1$-dimensional vertices since $V$ is good. If $V \in G_2$, then $V' = V$ and we are done, else in the last step we go from $V' \rightarrow V$.

\subsubsection{Triangulating Cycles when $k_3-k_2 > k_2-k_1$}
Having formed the paths $P(U,V)$, we now note that
if $(V,W)$ is an edge, then $P(U,V), (V,W), P(U,W)$ 
form a cycle. We call this cycle $C(U,V,W)$ as before,
and show that it can be triangulated using a small
number of triangles.
\begin{lemma}\label{lem:triangulation-gr2}
Suppose that $k_2-k_1 \leq k_3-k_2$, and let $(V,W) \in \bigcup_{b<a \in [3]} \good_{ab}(\cB)$ be an edge. 
Then cycle $C(U,V,W) = U \xrightarrow[]{P(U, V)} V \rightarrow W \xrightarrow[]{P(W, U)} U$ has a triangulation of size $O((\frac{k_2}{k_2-k_1})^2)$. 
\end{lemma}
\begin{proof}
We only give a proof sketch here since the details are exactly the same as Lemma~\ref{lem:triangulation-gr1}. To create the paths $R^i$ we take the obvious path between $P_{2i}(U,W)$ and $P_{2i}(U,V)$ by flipping one block at a time, alternating between $k_2$ and $k_1$-dimensional vertices (instead of $k_2,k_3$-dimensional ones). Then putting in the length two paths between $R^i_j$ and $R^{i+1}_j$ (that are two $k_2$-dimensional vertices, that can be connected using one $k_1$-dimensional vertex) we break the cycle into $O(t^2)$ $8$-cycles and $O(t)$ triangles. 

Note that each $8$-cycle in the above tiling is of the following form for some $1 \leq j < i \in [t]$:
\begin{align*}
&(W_{(j)},U_{(i)},X) \rightarrow (W_{(j)},X) \rightarrow (W_{(j)},W_{(i)},X) \rightarrow (W_{(i)},X)\notag \\
&\rightarrow (V_{(j)},W_{(i)},X) \rightarrow (V_{(j)},X) \rightarrow (V_{(j)},U_{(i)},X) \rightarrow (U_{(i)},X)\rightarrow (W_{(j)},U_{(i)},X),
\end{align*}
with $X = (V_{\leq j-1},W_{(j+1)},\ldots,W_{(i-1)},U_{(> i)})$, with $\dim(X)=k_1-k=2k_1-k_2$. Since $k_3-k_2 \geq k_2 -k_1$, we get that $k_3 \geq k_2+k$. Therefore to tile this cycle we can use the vertices: $X_0 = (V_{(j)},W_{(j)},X,Y_0)$, where $Y_0$ is chosen so that $\dim(X_0) = k_2$, $X_1 = (X_0,U_{(i)},Y_1)$, and $X_2 = (X_0,W_{(i)},Y_2)$ where $Y_1,Y_2$ are chosen so that $\dim(X_1) = \dim(X_2) = k_3$. This is possible since $\dim(X_0,W_{(i)}),\dim(X_0,U_{(i)})\leq k_2+k\leq k_3$.
The figure below shows that after adding these in the cycle breaks into triangles. 

\begin{center}
\includegraphics[scale=0.20]{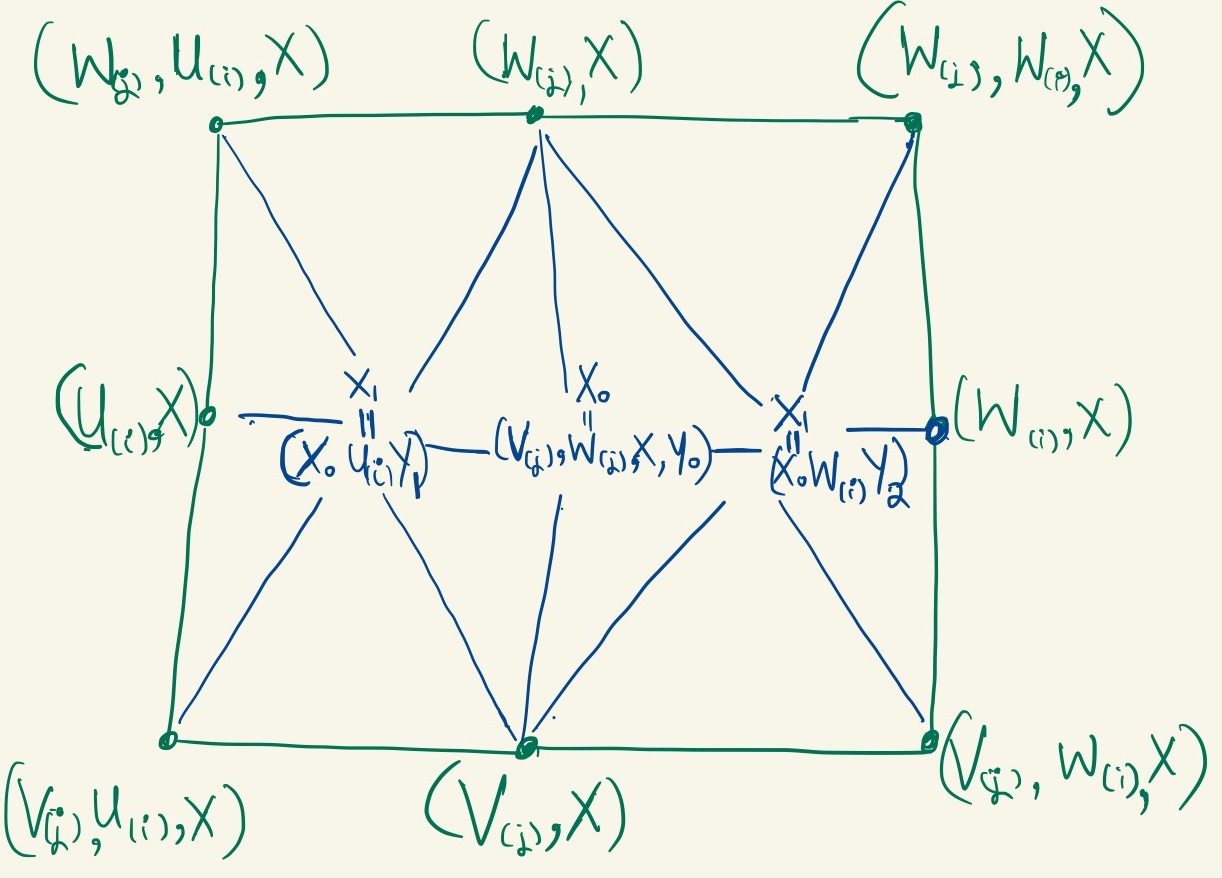}
\end{center}

As for the size of the triangulation: we had $O(t^2)$ 8-cycles, each tiled by $O(1)$ triangles, which gives an overall triangulation size of $O((k_2/k_2-k_1)^2)$ as required.
\end{proof}

\subsubsection{Proof of Lemma~\ref{lem:base-grassmann}}
Given the paths $P(U,V)$ and triangulations $T(U,V,W)$ for every good edge $(V,W)$, it is easy to complete the proof of Lemma~\ref{lem:base-grassmann} using the fact that $\GL_d(\F_q)$ acts transitively on the triangles of $\Gr(k_1,k_2,k_3)$. 

For each invertible linear transformation $L \in \GL_d(\F_q)$, and a subspace $V$ let $L(V)$ denote the subspace $\text{span}(Lv \mid v \in V)$ and let $L^{-1}(V)$ denote the subspace $W$ such that $L(W) = V$. For every $V \in \bigcup_{i\in [3]}\good_i(\cB)$ let $P_L(L(U), L(V))$ denote the path from $L(U) \rightarrow L(V)$ where at the $i^{th}$-step we have the vertex $L(P_i(U,V))$. It is easy to see that this is a valid path from $L(U)$ to $L(V)$. For a triangle $\Delta$ let $L(\Delta)$ denote the triangle whose vertices are $L(U_i), \forall U_i \in \Delta$. In fact for every good edge $(V,W)$ we can let $T_L(L(U),L(V),L(W))$ be the triangulation of the cycle $L(U) \xrightarrow[]{P_L(L(U),L(V))} L(V) \rightarrow L(W) \xrightarrow[]{P_L(L(W),L(U))} L(U)$ where a triangle in this triangulation is given by $L(\Delta)$ for $\Delta \in T(U,V,W)$. Again it is easy to see that this is a valid triangulation of the cycle for every edge $(V,W) \in \bigcup_{b<a}\good_{ab}(\cB)$ with the same size as $T(U,V,W)$.

We have the following randomized algorithm to get a highly satisfying UG solution to an arbitrary UG instance $\Phi$. 
\begin{mdframed}
\begin{algo}[$\Phi = (\Gr(k_1,k_2,k_3),\Pi)$]\mbox{}\label{algo:prop}\\
Input: UG instance $\Phi$ on $\Gr(k_1,k_2,k_3)$.\\
Output: A function $f: V(\Gr(k_1,k_2,k_3)) \rightarrow \S_m$.
\begin{enumerate}
\item Choose a random linear transformation $L \in \GL_d(F_q)$ and set $f(L(U)) = \text{id}$.
\item For each subspace $V \in \cup_{i \in [3]}\good_i(\cB)$, assign $f_L(L(V))$ the label obtained by propagating the label of $L(U)$ to $L(V)$ via the path $P_L(L(U),L(V))$, chosen appropriately according to whether $k_3-k_2$ or $k_2-k_1$ is larger.
\item For every subspace $V \notin \cup_{i \in [3]}\good_i(\cB)$ choose an arbitrary label for $L(V)$.
\end{enumerate}
\end{algo}
\end{mdframed}

We now complete the proof of Lemma~\ref{lem:base-grassmann} via the following lemma:
\begin{lemma}\label{lem:expected-sol-grassmann}
Let $\Phi$ be any UG instance over $\S_m$ with $\incons(\Phi) =\delta$. Then in expectation over $L \sim \GL_d(\F_q)$, the algorithm violates at most $O(K^2)\delta+\poly(K)/q$-fraction of edges, where 
\[
K = \max\left(\left\lceil\frac{k_2}{k_3-k_2}\right\rceil, \left\lceil\frac{k_2}{k_2-k_1}\right\rceil\right).
\]
\end{lemma}

\begin{remark}\label{rem:cones-grassmann}
\dor{The proof of Lemma~\ref{lem:expected-sol-grassmann} below uses the triangulations we constructed to get a high-valued solution to any locally consistent UG instance. This is a standard technique known as the cones method that has been used many times in the context of Abelian coboundary expansion. We mention that~\cite{DiksteinD23} were the first to use this method for getting non-Abelian coboundary expansion and this was later also used in~\cite{dikstein2023swap} (see Lemma 4.3 therein). They did it in the case \emph{every} cycle $C(U,V,W)$ has a small triangulation, namely the case without additive error. 
In the case of the current 
paper, the additive error is due to the fact that we constructed a small triangulation for \emph{most} (but not all) cycles $C(U,V,W)$, and this requires very minor modifications (the proof is included for completeness).}
\end{remark}

\begin{proof}
Suppose the propagation algorithm chooses a linear transformation $L$. Let $f_L: V(\Gr(k_1,k_2,k_3)) \rightarrow \S_m$ denote the assignment outputted by Algorithm~\ref{algo:prop} in this case. Let $E$ denote $E(\Gr(k_1,k_2,k_3))$ and let $\good(E) = \bigcup_{b<a}\good_{ab}(\cB)$. For every $(V,W) \in \good(E)$, the edge $(L(V),L(W))$ is satisfied by $f_L$ if the cycle $L(U) \xrightarrow[]{P_L(L(U),L(V))} L(V) \rightarrow L(W) \xrightarrow[]{P_L(L(W),L(U))} L(U)$ is consistent. Furthermore this is true if every triangle in $T_L(L(U),L(V),L(W))$ 
is consistent. Recall that this is the set of triangles $L(\Delta), \Delta \in T(U,V,W)$. So we get that,
\begin{align*}
\viol(f_L) &\leq \Pr_{(V,W) \sim E}[(V,W)\notin \good(E)] +\E_{(V,W)\sim \good(E)}[\Ind(\exists \Delta \in T(L(U),L(V),L(W)) \cap \incons(\Phi)]\\
&\leq \frac{\poly(K)}{q} + \max_{(V,W) \in \good(E)}(|T(U,V,W)|)\E_{(V,W) \sim \good(E)}\E_{\Delta \in T(U,V,W)}[\Ind(L(\Delta) \in \incons(\Phi)]\\
&\leq \frac{\poly(K)}{q} + O(K^2)\E_{(V,W) \sim \good(E)}\E_{\Delta \in T(U,V,W))}[\Ind(L(\Delta) \in \incons(\Phi)],
\end{align*}
where in the second inequality we used Lemma~\ref{lem:random-block-decomp} and the last one we used Lemmas~\ref{lem:triangulation-gr1} and~\ref{lem:triangulation-gr2} to bound the size of the triangulation.
Now taking an expectation over $L \sim \GL_d(\F_q)$ we get:
\begin{align*}
\E_L[\viol(f_L)] &\leq \frac{\poly(K)}{q}+O(K^2)\E_{(V,W) \sim \good(E)}\E_{\Delta \in T(U,V,W))}\E_L[\Ind(L(\Delta) \in \incons(\Phi)]\\
&\leq\frac{\poly(K)}{q}+ O(K^2)\delta,
\end{align*}
which completes the proof.   
\end{proof}

\subsection{The Base Case for Spherical Buildings of Type C}\label{sec:symp-base-case}
Let $\mu$ be the uniform distribution over chains of isotropic subspaces $(V_1 \subset \ldots V_d)$ of $\F_q^{2d}$, with $\dim(V_i) = i$. For $k_1,k_2,k_3 \leq d$ let $S_d(k_1,k_2,k_3)$ denote the tripartite graph $T(\{k_1\},\{k_2\},\{k_3\};\mu)$, where we drop the subscript $d$ when clear from context. We will prove that:
\begin{lemma}\label{lem:base-symplectic}
For all $m \in \N$, $0 < k_1,k_2,k_3 \leq d$, let $K = \max(\lceil\frac{k_2}{k_3-k_2}\rceil, \lceil\frac{k_2}{k_2-k_1}\rceil)$. Then $S(k_1,k_2,k_3)$ is an $(O(K^2),\poly(K)/q)$-coboundary expander over $\S_m$.
\end{lemma}

The proof of this statement will follow closely along the lines of Lemma~\ref{lem:base-grassmann}. 
\subsubsection{Auxiliary Claims}
\begin{claim}\label{claim:isotropic-path}
For $t < k < d \in \N$, given isotropic subspaces $W_1, W_2 \subseteq \F_q^{2d}$ with $\dim(W_1) = d_0-k$ and $\dim(W_2) = d_0$, there exists an isotropic subspace $W \subseteq W_2$ satisfying $\dim(W) = k$, and $W \cap (W_1 \cap W_2) = \{0\}$ such that $W+W_1$ is a $d_0$-dimensional isotropic subspace.
\end{claim}
\begin{proof}
Define $U = W_1 \cap W_2$ and denote $d_1 =\dim(U)$. We can write $W_1 = U+W'_1$ and $W_2 = U+W'_2$ with $W_1' \cap W_2' = \{0\}$.
Let $V$ be the subspace of vectors $v$ satisfying $\omega(v, w') = 0$ for all $w' \in W'_1$. Note that $V$ has dimension $2d -\dim(W'_1) = 2d-d_0+k+d_1$. Also 
\[\dim(V \cap W'_2) = \dim(V)+\dim(W'_2)-\dim(V+W'_2) \geq (2d-d_0+k+d_1)+(d_0-d_1)-2d = k.\]
Take $W \subseteq V \cap W'_2$ 
of dimension $k$.
Since $W\subseteq W_2'$ we get that $W$ is isotropic, and as $W\subseteq V$ we get that $W+W'_1$ is isotropic. Next, note that 
$W\cap W_1\cap W_2\subseteq W\subseteq W_2'\cap U = \{0\}$.
Finally, now that
\[
\dim(W + W_1)
=k+(d_0-k) - \dim(W\cap W_1)
=d_0,
\]
as $W\cap W_1\subseteq W\cap W_1\cap W_2'\subseteq W\cap W_1\cap W_2 = \{0\}$.
\end{proof}

For an isotropic subspace $U$, let $\symp(U)$ denote the subspace that is symplectically orthogonal to $U$, i.e. for all $w \in \symp(U), u \in U, \omega(w,u) = 0$. We note the 
standard facts that $\dim(\symp(U))=2d-\dim(U)$ and $\symp(U\oplus V) = \symp(U)\cap \symp(V)$. We also have that $\dim(U \cap \symp(W))=\dim(U)-\dim(W)+\dim(\symp(U)\cap W)$. In particular, when $\dim(U) = \dim(W)$ we get that $\dim(U \cap \symp(W))=\dim(\symp(U)\cap W)$.
\begin{claim}\label{claim:intersection-symp}
The following facts are true for sufficiently large $q$:
\begin{enumerate}
\item Fix a subspace $U$, let $1\leq i\leq d$ and choose an 
isotropic subspace $W$ of dimension $i$ uniformly at random. Then:
\begin{enumerate}
    \item With probability at least $1-\frac{4}{q}$ we have that \[
    \dim(U \cap \symp(W)) = 
    \max(0,\dim(U) - \dim(W)).
    \]
    \item With probability at least 
    $1-\frac{4}{q}$ we have that 
    \[
    \dim(\symp(U) \cap W) = 
    \max(0,\dim(W) - \dim(U)).
    \]
\end{enumerate}
\item Choosing maximal isotropic subspaces $U,W$ of $\mathbb{F}_q^{2d}$ randomly, we have that $U\cap W = \{0\}$ with probability at least $1-\frac{4}{q}$.

\end{enumerate}    
\end{claim}
\begin{proof}
    We prove each item separately, 
    and we begin with the first item. We start with (a), and assume that $\dim(U)\geq i$;
    otherwise we replace $U$ by 
    a subspace of dimension $i$ that contains it.
    Choose $w_1,\ldots,w_{i}$
    uniformly conditioned on $w_j$
    being symplectically orthogonal 
    to $w_1,\ldots,w_{j-1}$ for all $j$. Define the vectors $w_1',\ldots,w_i'$ by 
    $w_j(s) = -w_j(s+d)$ for 
    $s\leq d$, and $w_j(s) = w_j(s-d)$ for $s>d$. Thus, the space 
    $U\cap \symp(\{w_1,\ldots,w_{i}\})$ corresponds
    to all vectors in $U$ that are
    orthogonal to $w_1',\ldots,w_i'$.
    Note that for each $1\leq j\leq i$ and $\alpha_1,\ldots,\alpha_{j-1}\in\mathbb{F}_q$, the marginal distribution of $w_j'+\sum\limits_{\ell=1}^{j-1}\alpha_{\ell}w_{\ell}'$ 
    is uniform, and hence it is
    orthogonal to $U$ with probability at most $q^{(2d-\dim(U))-2d}=q^{-\dim(U)}$. It follows that no vector in $\spn(\{w_1',\ldots,w_i'\})$ is orthogonal to $U$ with probability at least
    \[
    1-\sum\limits_{j=1}^{i}q^{j-1}q^{-\dim(U)}
    \geq 1-\frac{2}{q^{(i-1)-\dim(U) }}
    \geq 1-\frac{2}{q},
    \]
    where we used the fact that 
    $\dim(U)\geq i$.
    Next, we note that the probability that $w_1,\ldots,w_i$
    is linearly independent is at 
    least 
    \[
    1-\sum\limits_{j=1}^{i}\frac{q^{2j}}{q^{2d}}
    \geq 1-2\frac{q^{2i}}{q^{2d}},
    \]
    in which case $w_1',\ldots,w_i'$
    are linearly independent too.   It follows from the union bound that with probability at least 
    $1-\frac{3}{q}$, the set $w_1,\ldots,w_i'$ is linearly independent and no vector in 
    it is orthogonal to $U$, in which case 
    the subspace in $U$ of vectors 
    orthogonal to $w_1,\ldots,w_i'$
    has dimension $\dim(U)-i$, and (a) follows.

    We move on to proving (b). Note 
    that
    \begin{align*}
    \dim(\symp(U)\cap W)
    &=\dim(\symp(U)) + \dim(W)
    -\dim(\symp(U)\oplus \dim(W))\\
    &=\dim(\symp(U)) + \dim(W)
    -2d + \dim(U\cap \symp(W))\\
    &=i-\dim(U) + \dim(U\cap \symp(W)),
    \end{align*}
    and the result follows from (a).
    
    The second item follows immediately from the first item with $i = d$.
\end{proof}

\begin{claim}\label{claim:random-path-symp}
Let $A,B$ be such that $\dim(A) = \dim(B) = k$. Then there exists a randomized algorithm to choose $A' \subseteq_{k-1} A, b \in B$ such that $A'+\spn(\{b\})$ is an isotropic subspace and for all $d_0$-dimensional $C$ where $(A+C)$ and $(B+C)$ are two $d_0+k$-dimensional isotropic subspaces,
\[\Pr_{A',b}[\dim(A'+\spn(\{b\})+C) = d_0+k] \geq 1-O\left(\frac{1}{q}\right).\]
\end{claim}
\begin{proof}
We know that $A \cap C = B \cap C = \{0\}$, and there are two cases:
\begin{enumerate}
    \item If there is no $b \in B$ that is isotropic to all of $A$, then choose $A' \subset_{k-1} A$ arbitrarily, and let $b \in B$ so that $A'+\spn(\{b\})$ is a $k$-dimensional isotropic subspace by Claim~\ref{claim:isotropic-path}. In this case, we know that $B \cap (A+C) = \{0\}$, because every vector in $A+C$ is in $\symp(A)$. Therefore $b \notin A'+C$, which gives that,
\[\dim(A'+\spn(\{b\})+C) = \dim(A'+C)+\dim(\spn(\{b\})) = \dim(A')+\dim(C)+1 = d_0+k,\]
as required.
\item If there is some non-zero $b_0 \in B \cap \symp(A)$, then choose $b = b_0$ and choose a uniformly random $A' \subset_{k-1} A$. Note that $A'+\spn(\{b\})$ is an isotropic subspace. Furthermore $\dim((C+\spn(\{b_0\})) \cap A) \leq 1$ since $\dim(C \cap A) = 0$. If $\dim((C+\spn(\{b_0\})) \cap A) = 0$, then $\dim(C+\spn(\{b_0\})+A') =d_0+k$, as required. Else, $(C+\spn(\{b_0\})) \cap A = \spn(\{v\})$ for some $v\neq 0$; by Claim~\ref{claim:random-subspace-intersection} we know that $A' \cap \spn(\{v\}) = \{0\}$ with probability $1-O(1/q)$, in which case $(C+\spn(\{b_0\})) \cap A' = \emptyset$ and $\dim(C+\spn(\{b_0\})+A) = d_0+k$ as required. 
\qedhere
\end{enumerate}
\end{proof}

\subsubsection{Constructing Paths}
In this section, we set up paths 
between vertices in the spherical 
building that will be convenient 
for us to triangulate. Let $S_i$ denote the vertices of dimension $k_i$, and let $k = k_2-k_1$. 
For simplicity of notation we will assume that $k_2$ is a multiple of $k$ and that $k_1\geq k$ (and so $k_1$ is also a multiple of $k$). The same proof with some slight modifications works when this is not the case, and in Section~\ref{sec:edge_cases} we explain the necessary modifications. 

Fix an arbitrary vertex $U \in S_2$, and pick an arbitrary basis for $U$: $\{u_1,\ldots, u_{k_2}\}$. Let $U_{(1)}$ denote the set of first $k$ vectors, $\{u_1,\ldots, u_{k}\}$, $U_{(2)}$ the second set of vectors and so on up to $U_{(t)}$ for $t = k_2/k$. Note that $k_1 = k(t-1)$. Define:
\begin{enumerate}
\item $\good_1 \subseteq S_1$ to be the set of subspaces $V$ that satisfy the following: $V \cap U = \{0\}$, and for all $i \in [t-1]$, $\dim(V \cap \symp(U_{>i})) = (i-1)k$ 
and $\dim(\symp(V)\cap \symp(U_{>i})) = 2d - k_1 - k_2 + ik$.
\item $\good_2 \subseteq S_2$ to be the set of subspaces $V$ that satisfy the following: $V \cap U = \{0\}$, and for all $i \in [t-1]$, $\dim(V \cap \symp(U_{>i})) = ik$.
\item $\good_3 \subseteq S_3$ to be the set of subspaces $V$ that satisfy the following: $V \cap U = \{0\}$ and $\dim(V \cap \symp(U)) = k_3-k_2$.
\end{enumerate}

\begin{lemma}\label{lem:good-verts-symp}
For all $i \in [3]$:
\[\Pr_{V \sim S_i}[V \in \good_i] \geq 1-O\left(\frac{t}{q}\right). \] 
\end{lemma}
\begin{proof}
Immediate from Claim~\ref{claim:intersection-symp} 
and the union bound.
\end{proof}

\paragraph{Setting up Paths from $U$:} 
We will now be interested in constructing 
paths from $U$ to other 
vertices $V$ in the graph.
Towards this end, for fixed $U$ and $V$ we will associate with $V$ a vertex $V'$. In the case that $\dim(V) = k_2$, we will take $V' = V$, and otherwise $V'$ will be an appropriately chosen subspace or superspace of $V$.
\begin{enumerate}
\item For a vertex $V$ of dimension $k_2$, using Claim~\ref{claim:isotropic-path} on $V$ and $U_{\geq 2}$ we find $V_{(1)} \subseteq_{k} V$ such that $(V_{(1)}, U_{(2)},\ldots, U_{(t)}) \in S_2$. Applying this claim iteratively, we find $V_{(2)}$ such that 
\[
(V_{(1)}, V_{(2)},U_{(3)},\ldots, U_{(t)}) \in S_2,
\] 
and so on. 
Then consider the following path from $U \rightarrow V$ which flips a block of $U$ to a block of $V$ one at a time: 
\begin{align*}
P(U,V) = &(U_{(1)},\ldots, U_{(t)}) \rightarrow U_{\geq 2} \rightarrow (V_{(1)},U_{\geq 2}) 
\rightarrow (V_{(1)},U_{\geq 3}) \rightarrow (V_{(1)},V_{(2)},U_{\geq 3}) \rightarrow \ldots \\
&\rightarrow (V_{< t},U_{(t)}) \rightarrow V_{<t} \rightarrow V.    
\end{align*}
Note that this path alternates between vertices of $S_2$ and $S_1$.
We set $V' = V$.
\item For a vertex $V$ of dimension $k_3$, we know that $\dim(V \cap \symp(U)) = k_3-k_2$, and thus we may choose $V' \subset_{k_2} V$ such that $V' \cap \symp(U) = \{0\}$. Consider its block decomposition $V' = (V_{(1)},\ldots, V_{(t)})$ such that for all $i \leq t$, $(V_{(1)},\ldots, V_{(i)},U_{(i+1)},\ldots,U_{(t)}) \in S_2$. As shown above, we know that such a decomposition is possible using Claim~\ref{claim:isotropic-path} iteratively. Then consider the following path from $U \rightarrow V$ which flips a block of $U$ to a block of $V$ one at a time similar to the above: 
\begin{align*}
P(U,V) = &(U_{(1)},\ldots, U_{(t)}) \rightarrow U_{\geq 2} \rightarrow (V_{(1)},U_{\geq 2}) 
\rightarrow (V_{(1)},U_{\geq 3}) \rightarrow (V_{(1)},V_{(2)},U_{\geq 3}) \rightarrow \ldots \\
&\rightarrow V' \rightarrow V.
\end{align*}
 \item We only give a path from $U$ to $V \in \good_1$. For such a $V$, we first find $V_{(1)}$ of dimension $k$ such that $(V_{(1)},V)$ and $(V_{(1)},U_{\geq 2})$ are in $S_2$ and $V_{(1)} \cap U_{(1)} = 0$. Such a subspace exists because 
\[
\dim(\symp(U_{\geq 2}) \cap \symp(V)) 
= 2d-k_1-k_2+k 
\geq 2(d-k_2) + 2k 
\geq 2k.\]
By Claim~\ref{claim:intersection-symp} picking a random $k$-dimensional isotropic subspace $V_{(1)}$ from the intersection will satisfy that $V_{(1)} \cap \symp(U_{(1)}) = \{0\}$, since $\symp(V)\cap \symp(U)$
is a co-dimension $k$ subspace of $\symp(U_{\geq 2}) \cap \symp(V)$. Additionally both $(V_{(1)},V)$ and $(V_{(1)},U_{\geq 2})$ are isotropic. Since $V \in \good_1$, $\symp(V) \cap U_{\geq 2} = V \cap \symp(U_{\geq 2}) = \{0\}$, which gives us that $V_{(1)} \cap V = V_{(1)} \cap U_{\geq 2} = \{0\}$ implying that $(V_{(1)},V)$ and $(V_{(1)},U_{\geq 2})$ are in $S_2$.

Now set $V' = (V_{(1)}, V)$ and consider the a path from $U \rightarrow V$. This path is the 
result of the above process for $k_2$ dimensional vertices 
applied on $V'$, where we already picked $V_{(1)}$, followed by
a final step from $V'$ to $V$.
\begin{align*}
P(U,V) = &U\rightarrow U_{\geq 2} \rightarrow (V_{(1)},U_{\geq 2}) 
\rightarrow (V_{(1)},U_{\geq 3}) \rightarrow (V_{(1)},V_{(2)},U_{\geq 3}) \rightarrow \ldots \\
&\rightarrow (V_{(1)},\ldots,V_{(t)}) \rightarrow V.
\end{align*}
Note that this path too alternates between vertices from $S_2$ and $S_1$.
\end{enumerate}

The block decomposition above satisfies the following properties:
\begin{claim}\label{claim:u_i-w_i}
Fix $U$ and take $V \in \bigcup_{i \in [3]} \good_i$, let $V'$ be the associated vertex with $V$ and write $V' = (V_{(1)},\ldots,V_{(t)})$ the block decomposition chosen by the path from $U$ as above. Then for all $i \in [t]$, $V_{(i)} \cap \symp(U_i) = \{0\}$.
\end{claim}
\begin{proof}
We split the proof to cases.
\paragraph{The case that $\dim(V) = k_2$:} in that case, by construction 
and definition of $\good_2$ 
we have that for $i\geq 2$ it holds that $V\cap \symp(U_{\geq i})= V_{\leq i-1}$. By this equality for $i+1$ instead of $i$, it follows that 
each $v\in V_{(i)}$ is symplectic to 
$U_{\leq i+1}$, and we conclude that
\[
V_{(i)}\cap \symp(U_i)
=V_{(i)}\cap \symp(U_{\geq i})
=V_{(i)}\cap V\cap \symp(U_{\geq i})
=V_{(i)}\cap V_{\leq i-1}
=\{0\}.
\]
\paragraph{The case that $\dim(V) = k_3$.} In that case, we picked 
$V'\subseteq V$ of dimension $k_2$ such that $V'\cap \symp(U) = \{0\}$, 
and the proof is exactly as in the 
previous case.

\paragraph{The case that $\dim(V) = k_1$:} by construction 
$V_{(1)}\cap \symp(U_1) = \{0\}$. For $i\geq 2$, 
the fact that 
$V_{(i)}\cap \symp(U_i) = \{0\}$ follows from 
the same argument above,
as the path from $U$ to 
$V'$ is constructed in 
the same way.
\end{proof}

\subsubsection{Constructing the Triangulations: Handling Special $8$-cycles}
With the paths $P(U,V)$, we would like to construct triangulations of cycles that they form. Towards this  end, we will have to triangulate $8$-cycles of a special form as in the following lemma:
\begin{lemma}\label{lem:8-cycle-symp}
Let $k \leq k_1 < k_2 < k_3 \leq d$ with $k_3 \geq k_2+k$, let $A,B,C,D$ be $k$-dimensional isotropic subspaces, let $X$ be a $k_2-2k$-dimensional isotropic subspace and let $X_1,X_2 \subseteq X$ be of dimension $k_1-k$. Suppose that $C \cap \symp(D) = \symp(C) \cap D = \{0\}$. Then the following 8-cycle, whose odd vertices are in $S_2$ and even ones are in $S_1$,    
\begin{align*}
C_8 &= (A,C,X) \rightarrow (A,X_1) \rightarrow (A,D,X) \rightarrow (D,X_2) \rightarrow (B,D,X) \\
&\rightarrow (B,X_1) \rightarrow (B,C,X) \rightarrow (C,X_2) \rightarrow (A,C,X),
\end{align*}
has a triangulation of size $O(1)$.
\end{lemma}

\begin{proof}
We break the proof into two cases. We will use the fact that $d \geq k_3 \geq k_2+k$.

\paragraph{Case 1--$A \subseteq \symp(B)$:}
This implies that $B \subseteq \symp(A)$ too.
Let $Z = (A,B,X)$ and $\dim(Z) = k_2-\ell$ for some $\ell \in [0,k]$. Let $V$ be a subspace such that: $V \subseteq \symp(Z)$ and $V \cap Z = \{0\}$.  We get that $\dim(V) \geq 2d - 2(k_2-\ell) \geq 2k+2\ell$. Let $G_1 \subseteq V$ be such that $G_1 \subseteq \symp(C) \cap V$ and $G_2 \subseteq V \cap \symp(D)$. We get that $\dim(G_1), \dim(G_2) \geq \dim(V) - k$, therefore: 
\[\dim(G_1 \cap G_2) \geq \dim(V)- 2k \geq 2\ell.\]
Pick an arbitrary $\ell$-dimensional isotropic subspace $Y \subset G_1 \cap G_2$ and add in the isotropic subspace $Z_0 = (Z,Y)$. As $Y \subset V$ it follows that $Y \cap Z = \{0\}$, and therefore $\dim(Z_0) = k_2$. Then add in the vertices $Z_1 = (Z,Y,C,Y_1)$ and $Z_2 = (Z,Y,D,Y_2)$,  where $Y_1,Y_2$ are chosen such that $Z_1,Z_2$ are $k_3$-dimensional isotropic subspaces. This is possible since $(Z,Y,C),(Z,Y,D)$ are subspaces of dimension $\leq \dim(Z,Y)+k \leq k_3$ and they are isotropic since $Y \subset G_1 \cap G_2$. The figure below shows that adding in $Z_0,Z_1,Z_2$ breaks the 8-cycle into triangles.

\begin{center}
\includegraphics[scale=0.20]{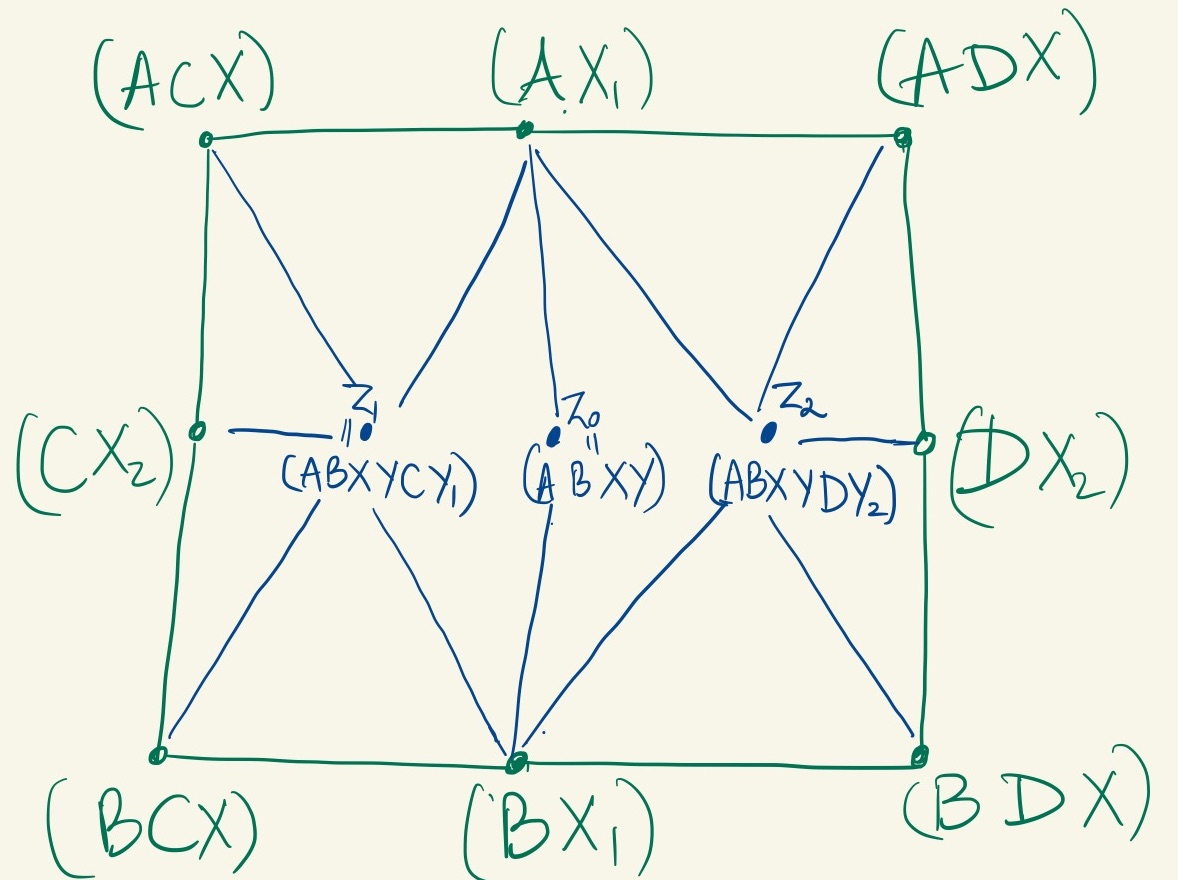}
\end{center}
\paragraph{Case 2 -- $A$ may not be in $\symp(B)$:}
Write $A = A_1 + A_2$, where $A_1 \subset \symp(B)$ and $A_2 \cap \symp(B) = 0$. Similarly write $B = B_1+B_2$, where $B_1 \subset \symp(A)$ and $B_2 \cap \symp(A) = 0$. We know that $\dim(A_1) = \dim(B_1)$, and say they are equal to $k-\ell$ for some $\ell \in [0,k]$; thus $\dim(A_2) = \dim(B_2)=\ell$. We know that $\dim(A_1,B_1,X) \leq 2(k-\ell)+k_2-2k = k_2-2\ell$, and so we write $\dim(A_1,B_1,X)=k_2-2\ell-z$ for some $z \geq 0$. Let us find a subspace $V'$ that satisfies: 
\[V' \subseteq \symp(A_2)\cap\symp(B_2)\cap \symp(C)\cap\symp(D)\cap \symp(A_1,B_1,X),~~~ V' \cap (A_1,B_1,X) = \{0\},\]

By dimension counting we get,
\begin{align*}
\dim(V') &\geq 2d-\dim(A_2)-\dim(B_2)-\dim(C)-\dim(D) - 2\dim(A_1,B_1,X) \\
&= 2d-2\ell-2k-2(k_2-2\ell-z)\\
&= 2(d-k_2-k) + 2(\ell+z)\\
&\geq 2(\ell+z).
\end{align*}
We can now pick any $(\ell+z)$-dimensional isotropic subspace $V$ inside $V'$. Also fix an arbitrary $\ell$-dimensional subspace $\wt{V}$ inside $V$. 

Below we elaborate on the vertices we add inside the cycle. It is easy to check that by construction each of these vertices are isotropic subspaces, therefore we will only check that they are of the correct dimensions ($k_1,k_2$ or $k_3$). Let $A' = (B_1,\wt{V})$. 
\begin{enumerate}
\item $(A',X_1)$: By construction $V \cap (B_1,X_1) = \{0\}$ and $\wt{V}\subseteq V$, hence $(B_1,X_1)\cap \wt{V} = \{0\}$ and by the dimension formula we get
\[\dim(A',X_1) = \dim(B_1,X_1)+\dim(\wt{V})=\dim(B_1)+\dim(X_1)+\dim(\wt{V})=k-\ell+k_1-k+\ell = k_1.\]
\item $(A,B_1,V,X)$: This is equal to $(A_2,A_1,B_1,X,V)$. We know that $A_2 \cap (A_1,B_1,V,X) = \{0\}$, since the latter is symplectically orthogonal to $B$ and $A_2\cap \symp(B) = \{0\}$. Also, by construction $V \cap (A_1,B_1,X)=\{0\}$, 
and applying the dimension formula twice we get that:
\[\dim(A,B_1,V,X) =\dim(A_2)+\dim(A_1,B_1,X)+\dim(V)=\ell+k_2-2\ell-z+\ell+z = k_2.\]
\item $(A,B_1,V,C,X,Y_1)$ where $Y_1$ is chosen so that the whole subspace is $k_3$-dimensional. This is possible since:
\[\dim(A,B_1,V,C,X)=\dim(A_2)+\dim(C)+\dim(A_1,B_1,X)+\dim(V)= k_2+k \leq k_3,\] 
where we used that $A_2 \cap (A_1,B_1,V,C) = \{0\}$ since the latter is symplectically orthogonal to $B$ 
and $A_2\cap \symp(B) = \{0\}$, and further $C \cap (A_1,B_1,V,X)$, since the latter is symplectically orthogonal to $D$ and $C \cap \symp(D) = \{0\}$ by assumption.
\item $(A',C,X)$: This equals $(B_1,\wt{V},C,X)$ and like the above, we get
\[\dim(A',C,X)=\dim(C)+\dim(B_1)+\dim(X)+\dim(\wt{V})=k_2.\]
\item $(A,B_1,V,D,X,Y_2)$ where $Y_2$ is chosen so that the whole subspace is $k_3$-dimensional, where the proof that this is possible is the same as the third item above.
\item $(A',D,X)$: This is a $k_2$-dimensional isotropic subspace like the fourth item above.
\end{enumerate}
The figure below shows that this breaks the cycle into triangles and another 8-cycle corresponding to the subspaces $A',B,C,D,X,X_1,X_2$. 

\begin{center}
\includegraphics[scale=0.20]{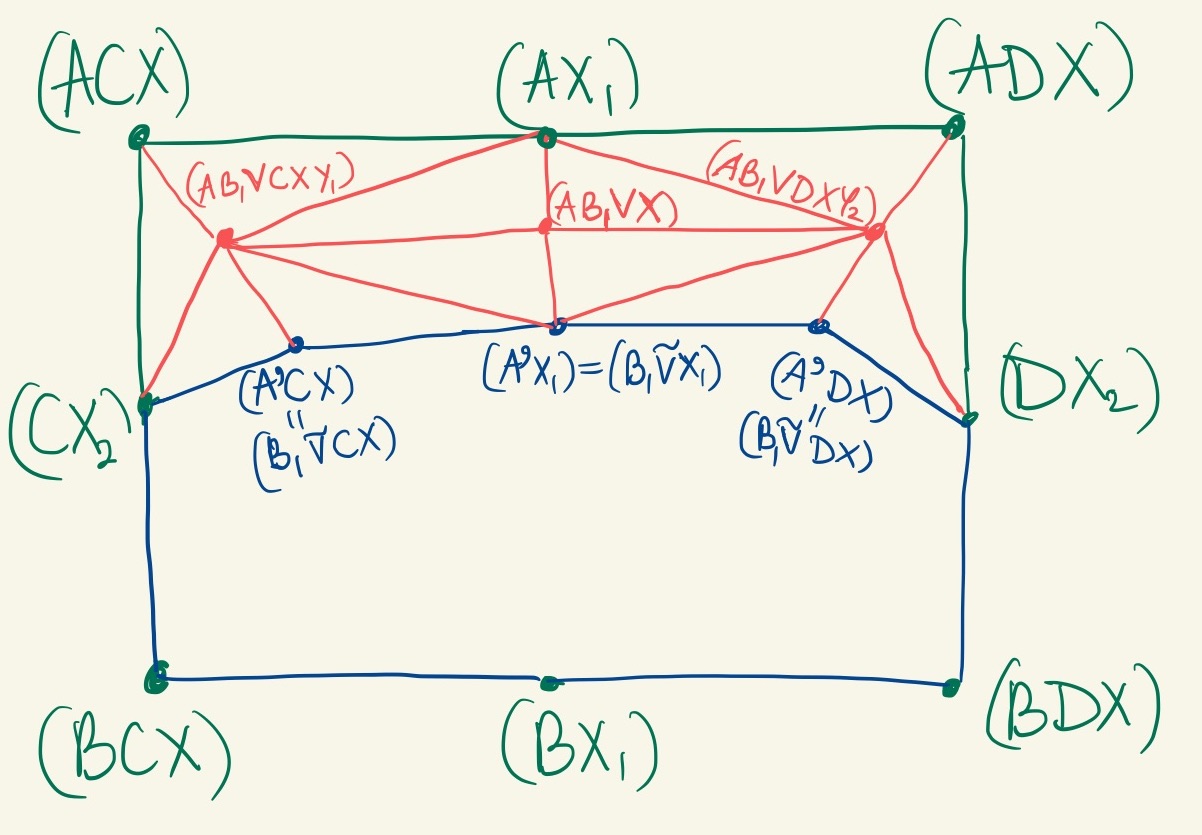}
\end{center}

Now note that $A' = (B_1,\wt{V})$ satisfies $A' \subseteq \symp(B)$. Therefore we can triangulate the remaining 8-cycle (in blue above) using Case 1, which gives a triangulation of the original cycle using $O(1)$ triangles.
\end{proof}


\subsubsection{Triangulating General Cycles}
\begin{lemma}\label{lem:triangulation-symp1}
For every edge $(V,W) \in S(k_1,k_2,k_3)$, where $V,W \in \bigcup_{i=1}^{3} \good_i$,
the cycle $C(U,V,W) = U \xrightarrow[]{P(U, V)} V \rightarrow W \xrightarrow[]{P(W, U)} U$ has a triangulation of size $O(K^2)$, where $K = \max\left(\frac{k_2}{k_2-k_1},\frac{k_2}{k_3-k_2}\right)$. 
\end{lemma}

\begin{proof}
Without loss of generality we assume that $\dim(V) > \dim(W)$, meaning that $W \subset V$. First, we break the cycle $C(U,V,W)$ into $O(t^2)$ cycles of length $8$ and $O(t)$ triangles. This step is slightly different depending on which type of edge we have, and we proceed by 
case analysis.

\paragraph{Tiling by $8$-cycles when $V \in S_3, W \in S_2$:} Let $V' \subset V$ be the vertex chosen by the path from $U$, with $V' = (V_{(1)},\ldots,V_{(t)})$ being the corresponding block decomposition. Let $(W_{(1)},\ldots,W_{(t)})$ be the block decomposition of $W$. For all $0 \leq i < j \leq t$ we can check that: 
\begin{equation}\label{eq:intermediate-verts-symp}
(V_{(1)},\ldots, V_{(i)},W_{(i+1)},\ldots, W_{(j)},U_{(j+1)},\ldots,U_{(t)}) \in S_2.
\end{equation}
Indeed, fix $i,j$. First note that this vertex is an isotropic subspace because (1) we know that $(V_{\leq i},U_{>i}) \in S_2$ so $V_{\leq i}$ is symplectically orthogonal to $U_{> j}$, (2) $(W_{\leq j},U_{>j}) \in S_2$ implying $(W_{(i+1)},\ldots,W_{(j)})$ is symplectically orthogonal to $U_{> j}$, and (3) $W, V' \subset V$ implying that $(W_{(i+1)},\ldots,W_{(j)}),V_{\leq i}$ are symplectically orthogonal. Therefore all vectors in the subspace in~\eqref{eq:intermediate-verts-symp} are symplectically orthogonal to each other. 

Second, we check that the dimension of the space in~\eqref{eq:intermediate-verts-symp} is $tk$. Since $W \in \good_2$, $\dim(W \cap \symp(U_{>i})) = \dim(W) - \dim(U_{>i}) = ik$. Since $\dim(W_{\leq i}) = ik$, and by construction $W_{\leq i}$ is symplectic to $U_{> i}$, it follows that $W\cap \symp(U_{>i}) = W_{\leq i}$, and so nothing in $W_{>i}$ is symplectically orthogonal to $U_{>i}$. Thus, as 
$V_{\leq i} \subset \symp(U_{> i})$, we conclude that $V_{\leq i} \cap W_{>i} = \{0\}$. We also argue that 
\[
(V_{\leq i}, W_{(i+1)},\ldots,W_{(j)})\cap(U_{>j})
=\{0\}.
\]
Indeed, if $v+w\in (V_{\leq i}, W_{(i+1)})$ and $w\neq 0$, then $v$ is symplectic to $U_{>j}$ and $w$ is not,
so $v+w$ is not symplectic to $U_{>j}$ 
and hence $v+w\not\in U_{>j}$. 
If $w = 0$, then the claim follows
as $V_{\leq i}\cap U_{>j} = \{0\}$
for $i\leq j$.
Combining everything, we get from 
the dimension formula that
\begin{align*}
\dim(V_{\leq i},W_{(i+1)},\ldots, W_{(j)},U_{>j}) &= \dim(V_{\leq i},W_{(i+1)},\ldots, W_{(j)})+\dim(U_{>j}) \\
&= \dim(V_{\leq i})+\dim(W_{(i+1)},\ldots, W_{(j)})+\dim(U_{>j}) \\
&= tk.
\end{align*}

Recall that $P_{2i}(U,V) = (V_{\leq i},U_{>i})$ and same for $W$. We will now create paths $R^i$ of length $2i+1$ between $P_{2i}(U,W)$ and $P_{2i}(U,V)$ for all $i \in [t]$ using the intermediate vertices from~\eqref{eq:intermediate-verts-symp}. These paths are constructed as:
\begin{align*}
R^i := &P_{2i}(U,W) \rightarrow (W_{(2)},\ldots, W_{(i)}, U_{>i}) 
\rightarrow (V_{(1)},W_{(2)},\ldots, W_{(i)}, U_{>i}) \rightarrow \ldots 
&\rightarrow P_{2i}(U,V).    
\end{align*}
Note that this is a valid path because the odd vertices (starting from $P_{2i}(U,W)$) are in $S_2$ by \eqref{eq:intermediate-verts-symp}, which also implies that the even vertices are in $S_1$. By creating the paths $R^{i}$, we have broken the original cycle $C(U,V,W)$ into $O(t)$ cycles of the form: 
\begin{align*}
    C^i = P_{2i}(U,W) \xrightarrow[]{R^{i}} P_{2i}(U,V) \rightarrow P_{2i+1}(U,V) \rightarrow P_{2i+2}(U,V) \xrightarrow[]{R^{i+1}} P_{2i+2}(U,W) 
&\rightarrow P_{2i+1}(U,W) \\
&\rightarrow P_{2i}(U,W),
\end{align*}
for $i \in [1,t-1]$ and 
\[C^{t} = W \xrightarrow[]{R^{t}} V' = (V_{(1)},\ldots,V_{(t)}) \rightarrow V \rightarrow W.\] 

We will tile each $C^i$ by 8-cycles and triangles, starting with $i \in [t-1]$. For all $j \in [1,i+1]$, let $R^i_j$ denote the $(2j-1)^{th}$ vertex on the path $R^i$, that is, 
\[R^i_j = (V_{<j},W_{(j)},\ldots, W_{(i)},U_{> i}).\] 
It is easy to see that for all $j \in [1,i+1]$, the vertices $R^i_{j}$ and $R^{i+1}_j$ are connected via a path of length two: 
\[R^i_j \rightarrow (V_{<j},W_{(j)},\ldots, W_{(i)},U_{\geq i+2},\ldots, U_{(t)}) \rightarrow R^{i+1}_j.\] 
Additionally the middle vertex in the length 2 path is equal to $P_{2i+1}(U,V)$ in the path from $R^i_{i+1} = P_{2i}(U,V) \rightarrow R^{i+1}_{i+1}$, $P_{2i}(U,V) \rightarrow P_{2i+2}(U,V)$ and $R^{i+1}_{i+1} \rightarrow R^{i+1}_{i+2}$, thus showing that the whole cycle has been broken into $O(t)$ 8-cycles and $O(1)$ triangles.

Let us now tile the cycle $C^{t}$. We have that for all $j \in [1,t+1]$, $R^{t}_j = (V_{<j},W_{(j)},\ldots, W_{(t)})$ is contained in $V$ since both $W$ and $V'$ are in $V$. This means that $C^t$ is broken into $O(t)$ triangles. So overall $C(U,V,W)$ has been broken into $O(t^2)$ 8-cycles and $O(t)$ triangles.

\paragraph{Tiling by $8$-cycles when $V \in S_3, W \in S_1$:} Let $V' \subset V$ be the vertex chosen by the path from $U$, with $V' = (V_{(1)},\ldots,V_{(t)})$ being the corresponding block decomposition. Let $W_{(1)}$ be the additional vertex used in the path from $U$ to $W$, with the block decomposition $W' = (W_{(1)},\ldots,W_{(t)})$, where $W = (W_{(2)},\ldots, W_{(t)}) \subset V$. For all $0 \leq i < j \in [t]$,
using the same argument as in~\eqref{eq:intermediate-verts-symp} we have that 
\begin{equation}\label{eq:intermediate-verts-symp2}
(V_{(1)},\ldots, V_{(i)},W_{(i+1)},\ldots, W_{(j)},U_{(j+1)},\ldots,U_{(t)}) \in S_2.
\end{equation}


Now we create the same paths $R^i$ of length $2i+1$ between $P_{2i}(U,W) \rightarrow P_{2i}(U,V)$ for all $i \in [t]$, and each vertex on these paths is in $S_2$ or $S_1$ because of \eqref{eq:intermediate-verts-symp2}. This breaks the cycle into the cycles $C^i$, for $i \in [t]$. The tiling of $C^i, i \in [t-1]$ proceeds identical to the first case, therefore let us discuss the tiling of the last cycle: 
\[C^{t} = W' = (W_{(1)},\ldots,W_{(t)}) \xrightarrow[]{R^{t}} V' = (V_{(1)},\ldots,V_{(t)}) \rightarrow V \rightarrow W \rightarrow W'.\] 
We have that the first three vertices of $R^t$ are $W'$, $(W_{\geq 2})$ and $(V_{(1)},W_{\geq 2})$ which are connected to $W$. Additionally including $(V_{(1)},W_{\geq 2})$, every other vertex on $R^t$, i.e. for all $j \in [2,t+1]$, $R^{t}_j = (V_{<j},W_{(j)},\ldots, W_{(t)})$ and the intermediate vertices in $S_1$, is contained in $V$ since both $W$ and $V'$ are in $V$. This means that $C^t$ is broken into $O(t)$ triangles. So overall $C(U,V,W)$ has been broken into $O(t^2)$ 8-cycles and $O(t)$ triangles.

\paragraph{Tiling by $8$-cycles when $V \in S_2, W \in S_1$:} We can show that $C(U,V,W)$ in this case too can be broken into $O(t^2)$ 8-cycles and $O(t)$ triangles. The proof for the cycles $C^{i}$ for $i \in [t-1]$ is the same, so we only discuss the tiling of $C^t$. We have that,
\[C^{t} = W' = (W_{(1)},\ldots,W_{(t)}) \xrightarrow[]{R^{t}} V = (V_{(1)},\ldots,V_{(t)}) \rightarrow W \rightarrow W',\] 
where $W = (W_{(2)},\ldots,W_{(t)})$. The first and second vertex on $R^t$, $W',W$ are connected to $W$, and then one can check that every vertex on $R^t$ except for the first one, is connected to $V$. This breaks $C^t$ into $O(t)$ triangles.

\paragraph{Triangulating the $8$-cycles:}
Having shown that the cycle $C(U,V,W)$ 
can always be tiled by at most $O(t^2)$
$8$-cycles and $O(t)$ triangles, it suffices
to show that each one of the resulting 
$8$-cycles can be triangulated individually.
Towards this end, we first notice that the 
$8$-cycles we formed consist of edges between isotropic subspaces of dimension $k_1$ and isotropic subspaces of dimension 
$k_2$. Thus, to triangulate them we will have
to use auxiliary vertices of dimension $k_3$.
Intuitively, the difference $k_2-k_1$ measures how ``different'' adjacent vertices
in the cycle are, and it stands to reason
that the closer these vertices are, the 
easier time we will have triangulating it. 
In the proof below we handle two
cases separately: the case $k_2-k_1\leq k_3-k_2$, and the case that $k_2-k_1 > k_3-k_2$, 
and we begin with the former easier case.

\paragraph{Triangulating the 8-cycles when $k_3-k_2\geq k_2-k_1$:} 
Note that each $8$-cycle in the above tiling is of the following form for some $1 \leq j < i \in [t]$:
\begin{align}
(W_{(j)},U_{(i)},X) \rightarrow (W_{(j)},X) \rightarrow (W_{(j)},W_{(i)},X) \rightarrow (W_{(i)},X)\notag \\
\rightarrow (V_{(j)},W_{(i)},X) \rightarrow (V_{(j)},X) \rightarrow (V_{(j)},U_{(i)},X) \rightarrow (U_{(i)},X) \rightarrow (W_{(j)},U_{(i)},X),\label{eq:8-cycle-symp}
\end{align}
with $X = (V_{\leq j-1},W_{(j+1)},\ldots,W_{(i-1)},U_{(> i)})$. Let $A=W_{(j)},B=V_{(j)},C = U_{(i)},D = W_{(i)}$.
By Claim~\ref{claim:u_i-w_i}, $C \cap \symp(D) = \{0\}$. Furthermore each of these blocks are of size $k = k_2-k_1$ and our assumption reads that $k_3 \geq k_2+k$. This cycle therefore satisfies the assumptions of Lemma~\ref{lem:8-cycle-symp} and can thus be triangulated with $O(1)$ triangles.

\paragraph{Triangulating the 8-cycles when $k_3-k_2\leq k_2-k_1$:} 
This case is more difficult, and we first start
with an $8$-cycle of the form we created, and  transform it into an $8$-cycle that can 
be triangulated using Lemma~\ref{lem:8-cycle-symp}.
Fix some $i_0 <j_0 \in [t]$ consider an 8-cycle as in~
\eqref{eq:8-cycle-symp}, namely:
\begin{align*}
C_8 &= (A,C,X) \rightarrow (A,X) \rightarrow (A,D,X) \rightarrow (D,X) \rightarrow (B,D,X) \\
&\rightarrow (B,X) \rightarrow (B,C,X) \rightarrow (C,X) \rightarrow (A,C,X),
\end{align*}
where $A,B,C,D,X$ are defined appropriately as in the above paragraph. We have the property that $C \cap \symp(D) = D \cap \symp(C) = \{0\}$.

Picture the 8-cycle as a square, with the $k_2$-dimensional vertices at the 4 corners. We will construct two ``horizontal'' paths: $(A,C,X) \rightarrow (A,C,D)$ and $(B,C,X) \rightarrow (B,D,X)$ and two vertical paths: 
$(A,C,X) \rightarrow (B,C,X)$ and $(A,C,X) \rightarrow (B,D,X)$. To do so first we apply the randomized algorithm from Claim~\ref{claim:random-path-symp} on the subspaces $C_1:=C, D:=D_1$ to get $C_2 \subset_{k-1} C_1, d_1 \in D$ such that $(d_1,C_2)$ is an isotropic subspace. We can then write $D = D_2 +d_1$, such that $D_2 \cap \spn(\{d_1\}) = \emptyset$, and apply Claim~\ref{claim:random-path-symp} again to get $d_2 \in D_2, C_3 \subset_{k-2} C_2$ such that $(d_2,C_3)$ is an isotropic subspace and in fact $(d_1,d_2,C_3)$ is also isotropic. Similarly applying the claim $k$ times we get a sequence of subspaces:
\begin{equation*}\label{eq:sequence-cd}
C, ~~~ (d_1,C_2), ~~~ (d_1,d_2,C_3),\ldots, (d_1,\ldots,d_{k-1},C_{k}), ~~~ D.    
\end{equation*}
We know that $C_{k} \subset_1 C_{k-1} \subset_2 \ldots C_2 \subset_{1} C$. Similarly we create the following sequence between $A$ and $B$:
\begin{equation*}\label{eq:sequence-ab}
A, ~~~ (b_1,A_2), ~~~ (b_1,b_2,A_2),\ldots, (b_1,\ldots,b_{k-1},A_{k}), ~~~ B.    
\end{equation*}
We will show that for all $i,j \in [k]$:
\begin{equation}\label{eq:symp-grid-vertices}
\Pr_{a,b,c,d}[\dim(b_{\leq i},A_{i},d_{\leq j},C_{j},X) = 2k] \geq 1-\frac{\poly(k)}{q}.   
\end{equation}
This is because:
\begin{align*}
\Pr_{a,b}[\dim(b_{\leq i},A_{i},X)= k] &= \prod_{i'=1}^i \Pr_{a,b}[\dim(b_{\leq i'-1},b_{i'},A_{i'},X)= k \mid \dim(b_{\leq i'-1},A_{i'-1},X) = k]\\
&\geq \left(1 - O\left(\frac{1}{q}\right)\right)^i \geq 1 - O\left(\frac{k}{q}\right),
\end{align*}
where for each term in the product we used Claim~\ref{claim:random-path-symp} with the subspaces $C=(b_{\leq i'-1},X)$, $A = A_{i'-1}$ and $B = B_{i'-1}$. To prove \eqref{eq:symp-grid-vertices} we use the same trick again:
\begin{align*}
&\Pr_{a,b,c,d}[\dim(b_{\leq i},A_{i},d_{\leq j},C_j,X)=2k] \\
&=\Pr_{a,b}[\dim(b_{\leq i},A_{i},X)=k]\Pr_{a,b,c,d}[\dim(b_{\leq i},A_{i},d_{\leq j},C_j,X)=2k \mid \dim(b_{\leq i},A_{i},X)=k]\\ 
&\geq \left(1 - O\left(\frac{k}{q}\right)\right)\prod_{j'=1}^j \Pr_{a,b,c,d}[\dim(d_{\leq j'-1},b_{\leq i},A_{i},d_{j'},C_{j'},X)= 2k \mid \dim(d_{\leq j'-1},C_{j'-1},b_{\leq i},A_{i},X) = 2k]\\
&\geq 1 - O\left(\frac{\poly(k)}{q}\right),
\end{align*}
where for each term in the product we used Claim~\ref{claim:random-path-symp} with the subspaces $C=(d_{\leq j'-1},b_{\leq i},A_{i},X)$, $A = C_{j'-1}$ and $B = D_{j'-1}$, which proves \eqref{eq:symp-grid-vertices}. 
Since $q > \poly(k)$ we can now union bound over all choices of $i,j \in [k]$ to get that there exists a choice of $a,b,c,d$'s such that for all $i,j \in [k]$:
\begin{equation}\label{eq:grid-vertices-full-rank}
\dim(b_{\leq i},A_{i},d_{\leq j},C_{j},X) = 2k.
\end{equation}
Henceforth fix such a choice of $a,b,c,d$'s.

Denote $k' = k_3-k_2$ and write $k = t'k'+r'$, with $r' < k'$ and $t' \geq 1$. For all $i \in [0,t']$, set $P_i(C,D)= (d_{\leq ik'},C_{ik'+1})$ and $P_{t'+1}(C,D) = D$. Then we have the path
\begin{align*}
R^{1} &= (A,C,X) \rightarrow (A,X) \rightarrow (A,P_1(C,D),X) \rightarrow \ldots \\
&\rightarrow (A,X) \rightarrow 
(A,P_{t'}(C,D),X) \rightarrow (A,X) \rightarrow (A,D,X)  
\end{align*}
that alternates between $S_2$ and $S_1$ by \eqref{eq:grid-vertices-full-rank}. Similarly for all $i \in [0,t']$, let $P_i(A,B)$ denote the subspace, $(b_{\leq ik'},A_{ik'+1})$ and $P_{t'+1}(A,B) = B$. We get the following parallel path between $(B,C,X)$ and $(B,D,X)$:
\begin{align*}
R^{2} &= (B,C,X) \rightarrow (B,X) \rightarrow (B,P_1(C,D),X) \rightarrow \ldots \\
&\rightarrow (B,X) \rightarrow 
(B,P_{t'}(C,D),X) \rightarrow (B,X) \rightarrow (B,D,X). 
\end{align*}
We now create horizontal paths between $(A,X,P_i(C,D))$ and $(B,X,P_i(C,D))$ for all $0\leq i \leq t'+1$:
\begin{align*}
H^{i} &= (A,P_i(C,D),X) \rightarrow (P_i(C,D),X) \rightarrow (P_1(A,B),P_i(C,D),X) \rightarrow \ldots \\
&\rightarrow (P_i(C,D),X) \rightarrow 
(P_{t'}(A,B),P_i(C,D),X) \rightarrow (P_i(C,D),X) \rightarrow (B,P_i(C,D),X),  
\end{align*}
where every vertex is either in $S_2$ or $S_1$ by \eqref{eq:grid-vertices-full-rank}. Let $H^i_j$ denote the vertex $(P_j(A,B),P_i(C,D),X)$ -- this is connected to $H^{i+1}_j$ via the vertex $(P_j(A,B),X) \in S_1$. Therefore we have broken the whole cycle into multiple triangles on the periphery and a grid of 8-cycles of the form:
\begin{align*}
&(P_{j}(A,B),P_{i}(C,D),X) \rightarrow (P_{j}(A,B),X) \rightarrow (P_{j}(A,B),P_{i+1}(C,D),X)\\ 
&\rightarrow (P_{i+1}(C,D),X) \rightarrow (P_{j+1}(A,B),P_{i+1}(C,D),X) 
\rightarrow (P_{j+1}(A,B),X) \\
&\rightarrow (P_{j+1}(A,B),P_{i}(C,D),X) \rightarrow (P_{i}(C,D),X) \rightarrow (P_{j}(A,B),P_{i}(C,D),X),
\end{align*}
for $i,j \in [t']$. 
To rewrite this cycle in a more convenient form, let $A_{jk'+1} = A_{(j+1)k'+1} + A'$, for some $k'$-dimensional $A'$ satisfying $A' \cap A_{(j+1)k'+1} = \{0\}$. Let $B' = (b_{jk'+1},\ldots,b_{(j+1)k'})$, and $Y_{AB} = (b_{\leq jk'},A_{(j+1)k'+1})$. In these notations we have that
\[P_j(A,B)=(R_A,Y_{AB}) ~~~ P_{j+1}(A,B)=(b_{jk'+1},\ldots,b_{(j+1)k'},Y_{AB}).\]

Let $D' = (d_{ik'+1},\ldots,d_{(i+1)k'})$, so that $D_{ik'+1} = D_{(i+1)k'+1}+D'$. Let $C' = \symp(D_{(i+1)k'+1}) \cap C_{ik'+1}$, and $Y_{CD} = (d_{\leq ik'},C_{(i+1)k'+1})$. We will now show that $C' \cap \symp(D') = 0$ and $\dim(C') = k'$.
First note that for all $x \in [t']$, $C_x \cap \symp(D_x) = \symp(C_x) \cap D_x = \{0\}$. This is because $D = \spn(d_{\leq x-1})+D_x$ and $C_x \subset \symp(d_{\leq x-1})$ therefore if any vector $v \in C_x$ was symplectically orthogonal to $D_x$, then $v$ would be symplectically orthogonal to $D$ and we would get a contradiction to $C \cap \symp(D) = \{0\}$; thus $C_x\cap\symp(D_x) = \{0\}$, and the proof
for $\symp(C_x)\cap D_x = \{0\}$ is analogous.
Since $C_{(i+1)k'+1} \cap \symp(D_{(i+1)k'+1}) = \{0\}$, we get that $C_{(i+1)k'+1} \cap C' = 0$. By dimension counting we know that $\dim(C') \geq k'$,
which gives us the decomposition:
\[C_{ik'+1} = C_{(i+1)k'+1} + C',\]
and $\dim(C') = k'$. Since $C_{ik'+1} \cap \symp(D_{ik'+1}) = \{0\}$ and $C_{(i+1)k'+1} \subset \symp(D')$, we conclude that $C' \cap \symp(D') = \{0\}$. Finally we get that,
\[P_i(C,D)=(C',Y_{CD}) ~~~ P_{i+1}(C,D)=(d_{ik'+1},\ldots,d_{(i+1)k'},Y_{CD}).\]

Letting $Y = (Y_{AB},Y_{CD})$, we can rewrite the cycle as:
\begin{align*}
&(A',C',Y) \rightarrow (A',Y_{AB}) \rightarrow (A',D',Y) \rightarrow (D',Y_{CD}) \rightarrow (B',D',Y) \\
&\rightarrow (B',Y_{AB}) \rightarrow (B',C',Y_{AB}) \rightarrow (C',Y_{CD}) \rightarrow (A',C',Y).
\end{align*}
We have shown that each of $A',B',C',D'$ is of size $k' = k_3-k_2$, with $C' \cap \symp(D') = 0$ and  $k_3 \geq k_2+k'$. Thus the cycle satisfies the assumptions of Lemma~\ref{lem:8-cycle-symp} and can be triangulated in $O(1)$ triangles. 

\paragraph{Size of triangulation:} Since we showed a tiling of $C(U,V,W)$ by $O(t^2)$ 8-cycles, and a tiling of each 8-cycle by $O(\lceil k_2-k_1/k_3-k_2\rceil^2)$ triangles, $C(U,V,W)$ has a triangulation of size $O(K^2)$ as required. 
\end{proof}

\subsubsection{Using Triangulations for Coboundary Expansion}
Given the triangulations from Lemma~\ref{lem:triangulation-symp1} and as is
standard in applications of the cones method, we will use the property that $\sp_{2d}(\F_q)$ acts transitively on the triangles of $S(k_1,k_2,k_3)$ to complete the proof of Lemma~\ref{lem:base-symplectic}.

For a matrix $M \in \sp_{2d}(\F)$, and a subspace $V$ recall that $M(V)$ denotes the subspace $\spn(Mv \mid v \in V)$. Let $M^{-1}(V)$ denote the subspace $W$ such that $M(W) = V$. Given a path $P(U,V)$ from 
$U$ to $V$, let $P_M(M(U),M(V))$ denote the path from $M(U)$ to $M(V)$ where at the $i^{th}$-step we have the vertex $M(P_i(U,V))$. It is easy to see that this is a valid path from $M(U)$ to $M(V)$ if $V \in \bigcup_{i \in [3]}\good_i$. For a triangle $\Delta$ let $M(\Delta)$ denote the triangle whose vertices are $M(U_i), \forall U_i \in \Delta$. In fact we can let $T_M(M(U),M(V),M(W))$ be the triangulation of the cycle \[
M(U) \xrightarrow[]{P_M(M(U),M(V))} M(V) \rightarrow M(W) \xrightarrow[]{P_M(M(W),M(U))} M(U)
\]
where a triangle in this triangulation is given by $M(\Delta)$ for $\Delta \in T(U,V,W)$. Again it is easy to see that this is a valid triangulation of the cycle.

We have the following randomized algorithm to get a good solution to an arbitrary UG instance $\Phi$. 
\begin{mdframed}
\begin{algo}[$\Phi = (S(k_1,k_2,k_3),\Pi)$]\mbox{}\label{algo:prop2}\\
Input: UG instance $\Phi$ on $S(k_1,k_2,k_3)$.\\
Output: A function $f: V(S(k_1,k_2,k_3)) \rightarrow \S_m$.
\begin{enumerate}
\item Choose a random linear transformation $M \sim \sp_{2d}(\F_q)$ and set $f(M(U)) = \text{id}$.
\item For each subspace $V \in \bigcup_i \good_i$, assign $M(V)$ the label obtained by propagating the label of $M(U)$ to $M(V)$ via the path $P_M(M(U),M(V))$. For $V \notin \bigcup_i \good_i$, assign an arbitrary label to $M(V)$.
\end{enumerate}
\end{algo}
\end{mdframed}

We now complete the proof of Lemma~\ref{lem:base-symplectic} via the lemma below. \dor{The proof of this lemma is essentially the same as that of Lemma 4.3 in~\cite{dikstein2023swap} and Lemma~\ref{lem:expected-sol-grassmann} and we give it here for the sake of completeness (see also Remark~\ref{rem:cones-grassmann}).}
\begin{lemma}
Let $\Phi$ be any UG instance over $\S_m$ with $\incons(\Phi) =\delta$. Then in expectation over $M \sim \sp_{2d}(\F_q)$, the algorithm violates at most $O(K^2)\delta+\poly(K)/q$-fraction of edges where 
\[
K = \max\left(\left\lceil\frac{k_2}{k_3-k_2}\right\rceil, \left\lceil\frac{k_2}{k_2-k_1}\right\rceil\right).
\]
\end{lemma}

\begin{proof}
Suppose the propagation algorithm chooses a linear transformation $M$. Consider the assignment $f_M: V(S(k_1,k_2,k_3)) \rightarrow \S_m$ that Algorithm~\ref{algo:prop2} outputs in this case. For $V,W \in \bigcup_i \good_i$, an edge $(M(V),M(W))$ is satisfied if the cycle $M(U) \xrightarrow[]{P_M(M(U),M(V))} M(V) \rightarrow M(W) \xrightarrow[]{P_M(M(W),M(U))} M(U)$ is consistent, i.e. the permutations on the edges product to $\text{id}$. Furthermore this is true if every triangle in $T_M(M(U),M(V),M(W))$ is consistent. Recall that this is the set $M(\Delta)$ as $\Delta$ ranges over $T(U,V,W)$. Let $E$ denote $E(S(k_1,k_2,k_3))$ and $\good(E)$ denote the set of edges $(V,W)$ for $V,W \in \bigcup_i \good_i$. We get that,
\begin{align*}
\viol(f_M) &\leq \hspace{-0.5ex}\Pr_{(V,W) \sim E}[(V,W)\notin \good(E)] +\hspace{-2.4ex}\E_{(V,W)\sim \good(E)}[\Ind(\exists \Delta \in T_M(M(U),M(V),M(W)) \cap \incons(\Phi)]\\
&\leq O\left(\frac{K}{q}\right) + \max_{(V,W) \in \good(E)}(|T(U,V,W)|)\E_{(V,W) \sim \good(E)}\E_{\Delta \in T(U,V,W)}[\Ind(M(\Delta) \in \incons(\Phi)]\\
&\leq O\left(\frac{K}{q}\right) + O(K^2)\E_{(V,W) \sim \good(E)}\E_{\Delta \in T(U,V,W))}[\Ind(M(\Delta) \in \incons(\Phi)],
\end{align*}
where in the second inequality we used Lemma~\ref{lem:good-verts-symp} and the last one we used Lemmas~\ref{lem:triangulation-symp1} to bound the size of the triangulation.

Recall that $M$ acts transitively on the set of $t$-dimensional isotropic subspaces, therefore the distribution $M(\Delta)$ over $M \sim \sp_{2d}(\F_q)$ is uniform over the triangles in $S(k_1,k_2,k_3)$. Using this fact, we can now take an expectation over $M \sim \sp_{2d}(\F_q)$ for the above equation to get:
\begin{align*}
\E[\viol(f_M)] &\leq O\left(\frac{K}{q}\right)+O(K^2)\E_{(V,W) \sim \good(E)}\E_{\Delta \in T(U,V,W))}\E_M[\Ind(M(\Delta) \in \incons(\Phi)]\\
&\leq O\left(\frac{K}{q}\right)+O(K^2)\delta,
\end{align*}
which completes the proof.
\end{proof}

\subsubsection{Modifications in Triangulations in Edge Cases}\label{sec:edge_cases}
In this section, we explain the necessary modifications to the triangulation argument 
in the case that $k_2$ and $k_1$ are not multiples of $k = k_2-k_1$. Let $k_2=(t-1)k+a$ for some $2 \leq t \in \N$ and $a \in \{0,\ldots,k-1\}$. Then $k_1 = (t-2)k+a$. Fix a vertex $U \in S_2$ and let $U = (U_{(1)},\ldots,U_{(t)})$ where each block has dimension $k$, except for the last one which has dimension $a$. We can define the set $\bigcup_i \good_i$ accordingly which gives the following set of paths from $U$.

\paragraph{Set of Paths from $U$:}
The paths from $U$ are the same for all but the second to last step. For every $V$ we will associate with $V$ a vertex $V'$. In the case that $\dim(V) = k_2$, we will take $V' = V$, and otherwise $V'$ will be an appropriately chosen subspace or superspace of $V$. We will also associate with $V$ an $a$-dimensional subspace $V'_{(t-1)}$ which is a subset of $V_{(t-1)}$ that is obtained while creating a path from $U$ to $V$.
\begin{enumerate}
\item For a vertex $V$ of dimension $k_2$, using Claim~\ref{claim:isotropic-path} on $V$ and $U_{\geq 2}$ we find $V_{(1)} \subseteq_{k} V$ such that $(V_{(1)}, U_{(2)},\ldots, U_{(t)}) \in S_2$. Applying this claim iteratively, we find $V_{(2)}$ such that 
\[
(V_{(1)}, V_{(2)},U_{(3)},\ldots, U_{(t)}) \in S_2,
\] 
and so on. Let $V'_{(t-1)}$ be a random $a$-dimensional subspace of $V_{(t-1)}$. Then consider the following path from $U \rightarrow V$ which flips a block of $U$ to a block of $V$ one at a time: 
\begin{align*}
P(U,V) = &(U_{(1)},\ldots, U_{(t)}) \rightarrow U_{\geq 2} \rightarrow (V_{(1)},U_{\geq 2}) 
\rightarrow (V_{(1)},U_{\geq 3}) \rightarrow (V_{(1)},V_{(2)},U_{\geq 3}) \rightarrow \ldots \\
&\rightarrow (V_{< t},U_{(t)}) \rightarrow (V_{<t-1},V'_{(t-1)}) \rightarrow V.    
\end{align*}
Note that this path alternates between vertices of $S_2$ and $S_1$. We set $V' = V$.
\item We do the same for vertices in $\good_1,\good_3$. The paths remain the same for all but the step where we put in the vertex $(V_{<t-1},V'_{(t-1)})$ and one can show that these are valid paths alternating between $S_2$ and $S_1$.
\end{enumerate}

\paragraph{Triangulating the cycles $C(U,V,W)$:} Let us focus on the case where $(V,W)$ is an edge for $V \in \good_2, W\in \good_1$. The other two cases follow analogously. We create the same paths $R^i$ from $P_{2i}(U,W)$ to $P_{2i}(U,V)$. The only path that is slightly different is the path $R^t$:
\begin{align*}
R^t := P_{2i}(U,W) \rightarrow (W_{(2)},\ldots, W_{(t)}) 
\rightarrow (V_{(1)},W_{(2)},\ldots, W_{(t)}) \rightarrow \ldots \rightarrow (V_{<t}, W_{(t)}) &\rightarrow (V_{<t-1}, V'_{(t-1)})\\
&\rightarrow P_{2i}(U,V).    
\end{align*}
One can now tile each cycle $C^i$ created using 8-cycles and triangles in an identical manner when $i \in [1,t-2] \cup \{t\}$. Let us discuss the tiling of $C^{t-1}$. Putting in horizontal length two paths between $R^{t-1}_j$ and $R^t_j$ we can break $C^{t-1}$ into 8-cycles of the form:
\begin{align}
&(W_{(1)},U_{(t)},X,Y) \rightarrow (W_{(1)},X) \rightarrow (W_{(1)},W_{(t)},X,Y) \rightarrow (W_{(t)},X,Y)\notag \\
&\rightarrow (V_{(1)},W_{(t)},X,Y) \rightarrow (V_{(1)},X) \rightarrow (V_{(1)},U_{(t)},X,Y) \rightarrow (U_{(t)},X,Y) \rightarrow (W_{(1)},U_{(t)},X,Y),\label{eq:8-cycle-symp-modified}
\end{align}
with $X = (W_{(2)},\ldots,W_{(t-2)},W'_{(t-1)})$ and $Y = W''_{(t-1)})$, where $W''_{(t-1)}+W'_{(t-1)} = W_{(t-1)}$ and $W''_{(t-1)} \cap W'_{(t-1)} = \{0\}$. The only difference in this 8-cycle from the one in~\eqref{eq:8-cycle-symp} is that the vertices $(W_{(1)},U_{(t)},X,Y)$ and $(W_{(1)},W_{(t)},X,Y)$ have an intersection $(W_{(1)},X,Y)$ which is larger than the $k_1$-dimensional subspace $(W_{(1)},X)$, and the same holds for the vertices $(V_{(1)},U_{(t)},X,Y)$ and $(V_{(1)},W_{(t)},X,Y)$. Intuitively this is only easier to tile than~\eqref{eq:8-cycle-symp} where every two adjacent points on the square intersect in a $k_1$-dimensional subspace. Indeed this holds, and in both the cases $k_3-k_2 \geq k_2-k_1$ and $k_3-k_2 \leq k_2-k_1$ we can use the same strategy to tile this 8-cycle with triangles.

Given this tiling of $C(U,V,W)$, the rest of the proof remains the same and we omit the details.

\subsection{The Extended Base Cases}
In this section we use Lemmas~\ref{lem:base-grassmann} and~\ref{lem:base-symplectic} to show that the spherical buildings of type A and C, and their tensor products, satisfy the Assumptions~\ref{assumption1}. That is, we will prove a bound on the coboundary constant of tripartite graphs $T(A,B,C;\mu|X_S = A_0)$, where $\max_{a \in A} a< \min_{b \in B} b$ and $\max_{b \in B} b < 
\min_{c \in C}c$, and
$\mu$ is the uniform distribution over the maximal faces of one of these complexes.

\subsubsection{The Extended Base Case for Restrictions of Type A and Type C}
To establish Assumption~\ref{assumption1}
we proceed in two steps. 
First, we prove an auxiliary
lemma handling the case 
that $\mu$ has a product structure in the sense that $\mu^{S_1\cup S_2} = \mu^{S_1}\times \mu^{S_2}$
and $\mu^{S_1\cup S_3} = \mu^{S_1}\times \mu^{S_3}$.
Let $\diam(G)$ denote the diameter of a graph $G$.
\begin{lemma}\label{lem:product-cbdry}
Let $\mu$ be a distribution over $\prod_{i \in [d]} X_i$, and let $G$ be a group acting on $\prod X_i$ such that for all $g \in G$ it holds that $g(X_i) = X_i$ and $g(\supp(\mu)) = \supp(\mu)$, as well as for all $X \in \supp(\mu)$, the distribution over $g(X)$ for a uniformly chosen $g\in G$ is $\mu$. Then for all pairwise disjoint sets $S_1,S_2,S_3 \subset [d]$ that satisfy $\mu^{S_1 \cup S_2} = \mu^{S_1} \times \mu^{S_2}$ and $\mu^{S_1 \cup S_3} = \mu^{S_1} \times \mu^{S_3}$, we get that 
\[
C(T(S_1,S_2,S_3;\mu)) \lll \diam(A(S_2,S_3;\mu)).
\]
\end{lemma}

\begin{proof}
We first construct a set of paths and triangulations for propagation. Let $P_i$ denote the set of vertices corresponding to the set $S_i$ in $H = T(S_1,S_2,S_3;\mu)$. Since $\mu^{S_1,S_2}$ and $\mu^{S_1,S_3}$ are both product distributions, we have an edge between $U,W$ for all $U \in P_1$ and $W \in P_3$, as well as between any $U \in P_1$ and $V \in P_2$. Fix two vertices $U_0 \in P_1, V_0 \in P_2$. For every $V \in P_2$ fix the path $P(U_0, V) = U_0 \rightarrow V$ and similarly for $W \in P_3$ fix the path $P(U_0,W) = U_0 \rightarrow W$. For any $U \in P_1$ fix the path $P(U_0, U) = U_0 \rightarrow V_0 \rightarrow U$. Now for an edge $(V,W)$ between $P_2$ and $P_3$ the cycle $(U_0,V,W)$ is already a triangle, therefore has a triangulation of size $1$. For an edge $(U,V)$ we need to tile the cycle $U_0 \rightarrow V_0 \rightarrow U \rightarrow V \rightarrow U_0$. To do so we use a path $R$ of length $\diam(A(S_2,S_3;\mu))$ from $V_0 \rightarrow V$ that alternates between vertices of $P_2$ and $P_3$. Every vertex in this path is connected to both $U$ and $U_0$ therefore this cycle has a triangulation of size $\lll \diam(A(S_2,S_3;\mu))$. We can follow the same strategy for triangulating the cycle $U_0 \rightarrow W \rightarrow U \rightarrow V_0 \rightarrow U_0$ corresponding to the edge $(V,W)$ -- we build a path $R$ of length $\diam(A(S_2,S_3;\mu))$ from $V_0$ to $W$, and given that every vertex in the path is connected to both $U_0$ and $U$ this gives us a triangulation of size $\lll \diam(A(S_2,S_3;\mu))$. We can now use the fact that the triangles in $H$ are transitive under the action of $G$ to get that the coboundary constant of $H$ is $\lll \diam(A(S_2,S_3;\mu))$. The proof of this fact is the same as that of Lemma~\ref{lem:expected-sol-grassmann}, wherein the group $G$ was $\GL_d(\F_q)$, and therefore we omit the details.     
\end{proof}

Armed with Lemma~\ref{lem:product-cbdry}, we now prove that the restrictions of type A 
spherical building give tripartite graphs that are
coboundary expanders.
\begin{lemma}\label{lem:grassmann-extended-base-case}
Let $\mu = SB^A_{d}(\F_q)$. 
For all $t \leq d/3$, $S \subseteq I$ of size at most $d - 3t$ and restrictions $A_0 \in \supp(\mu^S)$, for all $t$-sized pairwise disjoint sets $A,B,C$ of $I \setminus S$, with $i = \max_{a \in A}a, j = \min_{b \in B}b, j' = \max_{b \in B} b$ and $k = \min_{c \in C} c$ satisfying $i < j \leq j' < k$, the graph $T(A,B,C;\mu|X_S = A_0)$ is a $(\poly(K),\poly(K)/q)$-coboundary expander over $\S_m$, for $K = \max\left(\frac{k}{j-i},\frac{k}{k-j'}\right)$.
\end{lemma}
\begin{proof}
In this proof we will apply Lemma~\ref{lem:product-cbdry} multiple times. To do so, we will use the fact that for all $S \subseteq [d], A_0 \in \supp(\mu^S)$, the subgroup $H$ of $\GL_d(\F_q)$ that fixes $A_0$, has the property that for all $X \in \supp(\cD)$, the distribution $h(X), h \sim H$ is equal to $\cD$ for $\cD = \mu|X_S = A_0$. 

We first consider the case of $t=1$.
Fix a set $S \subset [d]$ of size at most $d-3$, $A_0 \in \supp(\mu^S)$ and consider coordinates $i,j,k \in [d] \setminus S$ with $i < j < k$. Let $\cD$ denote the distribution $\mu|X_S=A_0$. We divide the proof into two cases:
\paragraph{There exists $\ell \in S$ such that $\ell \in (i,j)$ or $\ell \in (j,k)$:} Let us assume that there is $\ell \in S$ with $\ell \in (i,j)$. 
Then note that $\cD^{\{i\},\{j\}} = \cD^{\{i\}}\times \cD^{\{j\}}$ and 
$\cD^{\{i\},\{k\}} = \cD^{\{i\}}\times \cD^{\{k\}}$. 
Therefore we can apply Lemma~\ref{lem:product-cbdry} to get that $C(T(\{i\},\{j\},\{k\};\cD)) \lll \diam(A(\{j\},\{k\};\cD))$, which we know is at most $O(k/(k-j))$. 

The same proof works when $\ell \in (j,k)$ to give a coboundary constant $O(j/j-i)$.

\paragraph{There is no $\ell \in S$ with $\ell \in (i,k)$:} In this case, let $d_1 \in S$ be the largest index smaller than $i$ and let $d_2 \in S$ be the smallest index larger than $k$. We know that $G$ is isomorphic to $\Gr(i-d_1,j-d_1,k-d_1)$ over the ambient space $\F_q^{d_2-d_1}$. Therefore we can use Lemma~\ref{lem:base-grassmann} to conclude that our graph is a $(\poly(K'),\poly(K')/q)$-coboundary expander for 
\[
K'=\max\left(\frac{j-d_1}{(j-d_1)-(i-d_1)},
\frac{j-d_1}{(k-d_1)-(j-d_1)}\right) 
\leq 
\max
\left(
\frac{k}{k-j},
\frac{k}{j-i}
\right)
\leq K,
\]
as required.
\skipi
We now move on to the case that $t>1$. Namely fix a set $S \subset [d]$ of size at most $d-3t$, $A_0 \in \mu^S$ and consider three pairwise disjoint $t$-sized sets $A,B,C \subset [d] \setminus S$ with $i < j < j' < k$. Let $\cD$ denote the distribution $\mu|X_S=A_0$. Using the argument in Lemma~\ref{lem:easy-recursion}, and more specifically~\eqref{eq:easy-recursion-main-C}, we get that
\begin{align*}
&C(T(A,B,C;\cD)) \leq C(T(\{i\},\{j\},\{k\};\cD))\cdot \\
&\qquad\qquad\max_{\substack{a \in \supp(\cD^i)\\b \in \supp(\cD^j)\\c \in \supp(\cD^k)}}(C(T(A\setminus \{i\},B,C;\cD|a)),C(T(A,B\setminus \{j\},C;\cD|b)),C(T(A,B,C\setminus \{k\};\cD|c))) \end{align*}
The first term above
is at most $K$ using
the case $t=1$ from above.
As for the second term, 
consider for instance 
$C(T(A\setminus\{i\}, B, C; \cD|a)$, and consider 
$\cD' = \cD|a$ and $A' = A\setminus \{i\}$. 
Note that $(\cD')^{A'\cup B} = (\cD')^{A'}\times (\cD')^{B}$ and 
$(\cD')^{A'\cup C} = (\cD')^{A'}\times (\cD')^{C}$, so applying Lemma~\ref{lem:product-cbdry} we get that the second term above 
is at most $\diam(A(S_2,S_3; \cD'))$.
Note that the diameter of
the graph $A(S_2,S_3; \cD')$
is at most a constant times
the diameter of 
$A(\{j'\}, \{k\}; \cD')$, 
which is easily seen to 
be at most $O\left(\frac{k}{k-j'}\right)\leq O(K)$. 
Combining the two bounds, 
we conclude that 
$C(T(A,B,C; \cD))\lll K^2$. Similarly using~\eqref{eq:easy-recursion-main-beta} It is easy to check that the additive error $\beta$ is at most $\poly(K)/q$.
\end{proof}

Next, we handle restrictions of type C spherical buildings. 
\begin{lemma}\label{lem:symp-extended-base-case}
Let $\mu = SB^C_{d}(\F_q)$. 
For all $t \leq d/3$, $S \subseteq I$ of size at most $d - 3t$ and restrictions $A_0 \in \supp(\mu^S)$, for all $t$-sized pairwise disjoint sets $A,B,C$ of $I \setminus S$, with $i = \max_{a \in A}a, j = \min_{b \in B}b, j' = \max_{b \in B} b$ and $k = \min_{c \in C} c$ satisfying $i < j \leq j' < k$, the graph $T(A,B,C;\mu|X_S = A_0)$ is a $(\poly(K),\poly(K)/q)$-coboundary expander over $\S_m$, for $K = \max\left(\frac{k}{j-i},\frac{k}{k-j'}\right)$.
\end{lemma}
\begin{proof}
The proof is very similar to the proof of 
Lemma~\ref{lem:grassmann-extended-base-case}, and 
we will need to use the following facts 
(which were already used in Lemma~\ref{lem:eps-product-C}): 
\begin{enumerate}
\item Let $V_{i'} \subset V_{k'}$ be two $i'$ and $k'$-dimensional isotropic subspaces. For $i'<i<j<k<k'$, the tripartite graph over $i,j$ and $k$-dimensional isotropic subspaces that are contained in $V_{k'}$ and contain $V_{i'}$ is isomorphic to $\Gr_{k'-i'}(i-i',j-i',k-i')$ as defined in Section~\ref{sec:grass-base-case}.
\item Let $V_{i'}$ be some $i'$-dimensional isotropic subspace. For $i'<i<j<k$, the tripartite graph over $i,j$ and $k$-dimensional isotropic subspaces that contain $V_{i'}$, is isomorphic to $S_{d-i'}(i-i',j-i',k-i')$ as defined in Section~\ref{sec:symp-base-case}.
\end{enumerate}

We will again apply Lemma~\ref{lem:product-cbdry} and to do so, we will use the fact that for all $S \subseteq [d], A_0 \in \supp(\mu^S)$, the subgroup $H$ of $\sp_{2d}(\F_q)$ that fixes $A_0$, has the property that for all $X \in \supp(\cD)$, the distribution $h(X), h \sim H$ is equal to $\cD$ for $\cD = \mu|X_S = A_0$. 

With these facts in mind, we begin with the proof in the case that $t=1$. Let $\cD =\mu|X_S =A_0$. We now prove that $G = T(\{i\},\{j\},\{k\}\};\cD)$ is a $(\poly(K),\poly(K)/q)$-coboundary expander, and there are a few cases to consider depending on the set of coordinates $S$ that we restricted.
\paragraph{There exists $\ell \in S$ such that $\ell \in (i,j)$ or $\ell \in (j,k)$:} 
Let us assume that there is $\ell \in S$ with $\ell \in (i,j)$. 
Then note that $\cD^{\{i\},\{j\}} = \cD^{\{i\}}\times \cD^{\{j\}}$ and $\cD^{\{i\},\{k\}} = \cD^{\{i\}}\times \cD^{\{k\}}$. 
Therefore we can apply Lemma~\ref{lem:product-cbdry} to get that $C(T(\{i\},\{j\},\{k\};\cD)) \lll \diam(A(\{j\},\{k\};\cD))$, which we can check is at most $O(k/(k-j))$. 
The same proof works when $\ell \in (j,k)$ to give a coboundary constant $O(j/j-i)$.

\paragraph{There is no $\ell \in S$ with $\ell \in (i,k)$ and there is $k' \in S, k' >k$:} In this case, let $k' \in S$ be the smallest index larger than $k$ and let $i' \in S$ be the largest index smaller than $i$ (set it to $0$ if no such index). By the above, $G$ is isomorphic to $\Gr_{k'-i'}(i-i',j-i',k-i')$ which is a $(\poly(K),\poly(K)/q)$-coboundary expander by Lemma~\ref{lem:base-grassmann}.

\paragraph{There is no $\ell \in S$ with $\ell \in (i,k)$ and no $k' \in S, k' >k$:}
In this case, again let $k' \in S$ be the smallest index larger than $k$ and let $i' \in S$ be the largest index smaller than $i$ ($i'=0$ if no such index).
By Fact 2 above, we know that $G$ is isomorphic to 
$S_{d-i'}(i-i',j-i',k-i')$ which is a $(\poly(K),\poly(K)/q)$-coboundary expander by Lemma~\ref{lem:base-symplectic}.

\skipi

We now move on to handle $t>1$. 
Using the argument in Lemma~\ref{lem:easy-recursion}, and more specifically~\eqref{eq:easy-recursion-main-C}, we get that
\begin{align*}
&C(T(A,B,C;\cD)) \leq C(T(\{i\},\{j\},\{k\};\cD))\cdot \\
&\qquad\qquad\max_{\substack{a \in \supp(\cD^i)\\b \in \supp(\cD^j)\\c \in \supp(\cD^k)}}(C(T(A\setminus \{i\},B,C;\cD|a)),C(T(A,B\setminus \{j\},C;\cD|b)),C(T(A,B,C\setminus \{k\};\cD|c)))    
\end{align*}
The first term above
is at most $K$ using
the case $t=1$ from above.
As for the second term, 
consider for instance 
$C(T(A\setminus\{i\}, B, C; \cD|a)$, and consider 
$\cD' = \cD|a$ and $A' = A\setminus \{i\}$. Note 
that $(\cD')^{A'\cup B} = (\cD')^{A'}\times (\cD')^{B}$ and 
$(\cD')^{A'\cup C} = (\cD')^{A'}\times (\cD')^{C}$, so applying Lemma~\ref{lem:product-cbdry} we get that the second term above 
is at most $\diam(A(S_2,S_3; \cD'))\lll K$. 
Combining the two bounds, 
we conclude that 
$C(T(A,B,C; \cD))\lll K^2$. Similarly using~\eqref{eq:easy-recursion-main-beta} It is easy to check that the additive error $\beta$ is at most $\poly(K)/q$.
\end{proof}

\subsubsection{The Extended Base Case for Tensors of Type A and C}
We now extend the base case from the 
previous section to a base case regarding \emph{tensors}
of type A and type C complexes.
\begin{definition}
Let $X = (X(0),X(1),\ldots,X(d_1))$ 
and $Y = (Y(0),Y(1),\ldots,Y(d_2))$ be simplicial complexes. We define the tensored simplicial complex $X\otimes Y$ as the complex $(Z(0),\ldots,Z(d_1+d_2))$
where for each $1\leq \ell\leq d_1+d_2$
we have
\[
    Z(\ell) = 
    \{A\cup B~|A\in X, B\in Y, |A|+|B| = \ell~\}.
\]
\end{definition}
We note that if we let $\mu,\nu$ be the uniform distributions on top dimensional faces of $X,Y$ respectively, then the distribution
$\mu\otimes\nu$ is uniform over the top dimensional faces
of $X\otimes Y$. Thus, it will often be convenient for
us to discuss tensors of complexes using the language
of associated probability distributions.

\begin{lemma}\label{lem:tensor-extended-base-case}
Let $\mu_1,\mu_2$ be two distributions over the domains $\prod_{i\in[d_1]}X_i$ and $\prod_{i=d_1+1}^{d} X_i$, with $\mu_1$ being either $SB^A_{d_1}(\F_q)$ or $SB^C_{d_1}(\F_q)$ and $\mu_2$ being either $SB^A_{d-d_1}(\F_q)$ or $SB^C_{d-d_1}(\F_q)$. Let $\mu = \mu_1 \times \mu_2$, $t\leq d/3$, $S \subseteq [d]$ be of size at most $d - 3t$, $A_0 \in \supp(\mu^S)$ a restriction of $S$, 
and let $A,B,C\subseteq [d]\setminus S$ be of size $t$ such that 
\[
i = \max_{a\in A} a
<
j = \min_{b\in B} b
\leq 
j' = \max_{b\in B} b
<
k = \min_{c\in C} c.
\]
Then the graph $G = T(A,B,C;\mu|X_S = A_0)$ is a $(\poly(K),\poly(K)/q)$-coboundary expander over $\S_m$ for $K = \max\left(\frac{k}{k-j'},\frac{k}{j-i}\right)$.  
\end{lemma}

\begin{proof}
Let $G_1 = \GL_{d_1}(\F_q)$ or $\sp_{2d_1}(\F_q)$ depending on whether $\mu_1$ is $SB^A_{d_1}(\F_q)$ or 
$SB^C_{d_1}(\F_q)$, and similarly define $G_2$ as  $\GL_{d-d_1}(\F_q)$ or $\sp_{2(d-d_1)}(\F_q)$. The set 
$\supp(\mu)$ is transitive under the action of $G_1 \times G_2$ and moreover for all $S \subseteq [d], A_0 \supp(\mu^S)$, the subgroup $H$ of $G_1 \times G_2$ that fixes $A_0$, has the property that for all $X \in \supp(\cD)$, the distribution $h(X), h \sim H$ is equal to $\cD$ for $\cD = \mu|X_S = A_0$. 

Given this fact, we begin with the case that $t=1$. Fix a set $S \subset [d]$ of size at most $d-3$, $A_0 \in \supp(\mu^S)$ and consider coordinates $i,j,k \in [d] \setminus S$ with $i < j < k$. Let $\cD = \mu|X_S = A_0$, $I_1=[d_1]$ and $I_2=\{d_1+1,\ldots,d\}$. We divide the proof into 4 cases:

\paragraph{The case that $i\in I_1$ and $j,k \in I_2$:}
First note that $\cD^{\{i,j\}}$ and $\cD^{\{i,k\}}$ are both product distributions. Hence we can apply Lemma~\ref{lem:product-cbdry} to get that the coboundary constant of $G$ is $\lll \diam(A(\{j\},\{k\};\cD))\lll K$. 

\paragraph{The case that $i,j\in I_1$ and $k \in I_2$:}
This case is similar to the above and we get a bound of $O(j/(j-i))\lll K$ on the coboundary constant.

\paragraph{All three coordinates in $I_2$:} Depending on whether $\mu_1$ is $SB^A_{d_1}$ or $SB^C_{d_1}$ we can use Lemma~\ref{lem:grassmann-extended-base-case} or~\ref{lem:symp-extended-base-case} to get that $G$ is a $(\poly(K),2^{-\Omega(r^{12})})$-coboundary expander.

\paragraph{All three coordinates in $I_1$:} This case follows analogously to the third case above.

\skipi
We now move on to the case that $t>1$, and the argument
here is the same as in Lemma~\ref{lem:grassmann-extended-base-case} and~\ref{lem:symp-extended-base-case}. 
Using the argument in Lemma~\ref{lem:easy-recursion}, and more specifically~\eqref{eq:easy-recursion-main-C}, we get that
\begin{align*}
&C(T(A,B,C;\cD)) \leq C(T(\{i\},\{j\},\{k\};\cD))\cdot \\
&\qquad\qquad\max_{\substack{a \in \supp(\cD^i)\\b \in \supp(\cD^j)\\c \in \supp(\cD^k)}}(C(T(A\setminus \{i\},B,C;\cD|a)),C(T(A,B\setminus \{j\},C;\cD|b)),C(T(A,B,C\setminus \{k\};\cD|c)))    
\end{align*}
The first term above
is at most $K$ using
the case $t=1$ from above.
As for the second term, 
consider for instance 
$C(T(A\setminus\{i\}, B, C; \cD|a)$, and consider 
$\cD' = \cD|a$ and $A' = A\setminus \{i\}$. Note 
that $(\cD')^{A'\cup B} = (\cD')^{A'}\times (\cD')^{B}$ and 
$(\cD')^{A'\cup C} = (\cD')^{A'}\times (\cD')^{C}$, so applying 
Lemma~\ref{lem:product-cbdry} we get that the second term above 
is at most $\diam(A(S_2,S_3; \cD'))\lll K$.
Combining the two bounds, 
we conclude that 
$C(T(A,B,C; \cD))\lll K^2$. Similarly using~\eqref{eq:easy-recursion-main-beta} It is easy to check that the additive error $\beta$ is at most $\poly(K)/q$.
\end{proof}

\section{UG coboundary expansion of Spherical Buildings}
We can now use the base case of induction along with the local to global lemma to get that the spherical buildings of type A and C are UG coboundary expander.

\begin{claim}\label{claim:poisson}
There is a constant $C>0$ such that 
if $\alpha(r)$ is a function of $r$ such that $\alpha(r)\geq 6$, and suppose that $d\geq 10r\alpha(r)$.
Then
\[\Pr_{i_1,\ldots,i_r \sim [d]}\left[\min_{j \neq k \in [r]}|i_j - i_k| \geq \frac{d}{r\alpha(r)}\right] \geq 2^{-r/\alpha(r)C}.\]
\end{claim}
\begin{proof}
Denote $r' = r\alpha(r)$ 
for simplicity. We will show 
that the event holds even if 
we pick elements from $[d]$ with repetitions, which is clearly stronger.
To do so, we first choose $a_1$, then calculate the probability that $a_2$ lies outside the $d/r'$ interval around $a_1$, then further calculate the probability that $a_3$ lies outside the $d/r'$ interval around $a_1$ and $a_2$ and so on. Formally it suffices to lower bound:
\[\Pr_{a_i \sim [d]}
\left[\min_{a \in \{a_1,\ldots,a_{i-1}\}}|a - a_i| \geq \frac{d}{r'} ~\Big|~ a_1,\ldots a_{i-1}\right],\]
for all $i \in [k]$.

Let us start by bounding the first expression, and let us fix an $i$. There are a total of at most 
$(i-1) \frac{d}{r'}$ elements
that are $\frac{d}{r'}$-close to some $a_{i'}$ for $1\leq i'\leq i-1$. Thus, 
the probability $a_i$ is chosen 
among them is at most $\frac{i-1}{r'}$, so the first expression is at least $1-\frac{i-1}{r'}$.

As the probability of the event we are interested in the product of the expressions over $i=1,\ldots,r$, we get that
\[\Pr_{\substack{a_1,\ldots,a_r \sim [d]}}\left[\min_{a \neq b \in \cup a_i}|a-b| \geq \frac{d}{r'}\right]  \geq 
\prod_{i \in [r]}\left(1-\frac{i-1}{r'}\right) 
\geq \left(1-\frac{r}{r'}\right)^r \geq 2^{-O(r/\alpha(r))}.
\qedhere\]
\end{proof}


\begin{definition}
Let $r \leq 3d$ and $\cD$ be a distribution over $\prod_{i \in [d]} X_i$. We define a ${d \choose r}$-partite graph $G_r(\cD)$ as follows: for every $S \subset_r [d]$, let the part $P_S$ contain the vertices $\supp(\cD^S)$ and to sample an edge we sample $S, S' \subset_r [d]$ with $S \cap S' =\emptyset$, $X \sim \cD$ and output $(X_{S},X_{S'})$. We associate $G_r(\cD)$ with the distribution on triangles that samples disjoint sets $S_1,S_2,S_3 \subset_r [d]$, $X \sim \cD$ and outputs $(X_{S_1},X_{S_2},X_{S_3})$.
\end{definition}

The next lemma shows that the graph $G_{r}(\cD)$ is 
a coboundary expander over $S_m$ with strong parameters
if $\cD$ is any restriction of 
a spherical building of type A, of type C 
or a tensor of them.
\begin{lemma}\label{lem:ug-cobdry-spherical}
Let $m, r,d,q\in \N$ with $r \ll d \ll q$ and let $\mu$ be the distribution $SB^A_d(\F_q)$, $SB^C_d(\F_q)$ or one of their tensor products. Then for all $S \subset_{\leq r} [d]$ and $A_0 \in \supp(\mu^S)$, $G_r(\mu|X_S = A_0)$ is a $(2^{O(r^{0.99}\log r)}, 2^{-\Omega(r^{12})})$-coboundary expander over $\S_m$.
\end{lemma}
\begin{proof}
Let us fix $\mu$ to be the distribution corresponding to the spherical building of type C. The proof when $\mu$ equals $SB^A_d(\F_q)$ is identical except that we use Lemmas~\ref{lem:eps-product-A} and \ref{lem:grassmann-extended-base-case} for type A instead of Lemmas~\ref{lem:eps-product-C} and~\ref{lem:symp-extended-base-case} for type C that are used below. Similarly, 
in the case that $\mu$ is one of the various possible tensor products, we use the fact that $\mu$ is an $O(1/\sqrt{q})$-product distribution (being the product of two $O(1/\sqrt{q})$-product distributions), and Lemma~\ref{lem:tensor-extended-base-case} instead of 
Lemma~\ref{lem:symp-extended-base-case}.

Consider $\cD = \mu|(X_S = A_0)$ for $S \subset [d]$ of size at most $r$ and $A_0 \in \supp(\mu^S)$. Let $\Phi = (G_r(\cD),\Pi)$ be a UG instance over $\S_m$ with $\incons(\Phi) = \delta$. Let $\cS$ be the set of tuples $(S_1,S_2,S_3)$ of subsets of $[d]\setminus S$ of size $r$ that are pairwise disjoint, and let $\cS'\subseteq \cS$ be the set of tuples $(S_1,S_2,S_3) \in \cS$ where $\cup S_i$ is $\Omega(d/r^2)$-separated and satisfies the third item
in Assumption~\ref{assumption1}. Using Claim~\ref{claim:poisson} we have
\begin{equation}\label{eq:apply_chernoff}
\Pr_{(S_1,S_2,S_3) \in \cS}[(S_1,S_2,S_3) \in \cS'] 
\geq \Omega(1) - \Pr_{(S_1,S_2,S_3) \in \cS}[\exists j\in [r^{0.9}], i\in [3] \text{ such that }||S_i\cap I_j| - r^{0.1}|\geq r^{0.06}].
\end{equation}
By the union bound, symmetry and the fact that 
$r\leq \sqrt{d}$ we have that
\[
\Pr_{(S_1,S_2,S_3) \in \cS}[\exists j\in [r^{0.9}], i\in [3] \text{ such that }||S_i\cap I_j| - r^{0.1}|\geq r^{0.06}]
\lll r \Pr_{S_1}[||S_1\cap I_1| - r^{0.1}|\geq r^{0.06}].
\]
We bound the latter probability using the multiplicative Chernoff bound (that holds for sampling without replacement too): 
\[\Pr_{S_1 \subset_r [d]}[||S_1\cap I_1| - r^{0.1}|\geq r^{0.06}] \leq 2^{-\Omega(r^{0.1} (r^{0.06}/r)^2)}=2^{-\Omega(r^{0.02})}.\]
Overall, plugging in these estimates into~\eqref{eq:apply_chernoff} we get that
\[
\Pr_{(S_1,S_2,S_3) \in \cS}[(S_1,S_2,S_3) \in \cS']
\ggg 1.
\]
Let $\cT$ be the set of $S_4 \subset_r [d]\setminus S$ with $S_4 \cap (\cup_{i \in [3]} S_i) = \emptyset$, 
and note that
\begin{equation}\label{eq:incons1}
\E_{(S_1,S_2,S_3) \sim \cS'}\E_{(a_1,a_2,a_3) \sim \mu^{\cup S_i}}[(a_1,a_2,a_3) \in \incons(\Phi)] \lll\delta,    
\end{equation}
\begin{equation}\label{eq:incons2}
\E_{(S_1,S_2,S_3) \sim \cS'}\E_{\substack{S_4 \sim \cT}}\E_{(a_1,a_2,a_4) \sim \mu^{S_1 \cup S_2 \cup S_4}}[(a_1,a_2,a_4) \in \incons(\Phi)] \lll\delta, 
\end{equation}
and
\begin{equation}\label{eq:incons3}
\E_{(S_1,S_2,S_3) \sim \cS'}\E_{\substack{S_4,S_5 \sim \cT:\\ S_4 \cap S_5 = \emptyset}}\E_{(a_1,a_4,a_5) \sim \mu^{S_1 \cup S_4 \cup S_5}}[(a_1,a_4,a_5) \in \incons(\Phi)] \lll\delta, 
\end{equation}
where we used the fact that the expectations of the same event under $(S_1,S_2,S_3) \sim \cS$ is equal to $\delta$. Using Markov's inequality we can henceforth fix a tuple $(S_1,S_2,S_3) \in \cS'$ where both the above expectations are bounded by $O(\delta)$.  Let $I = \cup_{i \in [3]} S_i$. We verify that $\cD^I$ and $I$ satisfy the Assumptions~\ref{assumption1}. Indeed,
\begin{enumerate}
\item By Lemma~\ref{lem:eps-product-C}, the measure $\cD$ is an $\eps$-product distribution with $\eps \lll 1/\sqrt{q}$ which can be made at most $2^{-r^{12}}$ by taking $q$ large.
\item The second item in Assumption~\ref{assumption1} 
follows from Lemma~\ref{lem:symp-extended-base-case}.
\item The third item holds by the definition of $\cS'$.
\end{enumerate}
This means that we may apply Lemma~\ref{lem:exp-bound-mu}, 
and indeed we do so.
\paragraph{Solving $\Phi$ on $H = T(S_1,S_2,S_3;\cD^I)$:} Since $\cD^I$ and $I$ satisfy Assumption~\ref{assumption1} we can apply Lemma~\ref{lem:exp-bound-mu} to get that $H$ is an $(2^{O(r^{0.99}\log r)},2^{-\Omega(r^{12})})$-coboundary expander over $\S_m$. In particular if $\Phi|_H$ is the restricted UG instance on $H$ then there exists a solution $A$ to the vertices of $H$ with: 
\[\viol(A) \leq 2^{O(r^{0.99}\log r)}\incons(\Phi|_H)+2^{-\Omega(r^{12})}) \leq  2^{O(r^{0.99}\log r)}\delta+2^{-\Omega(r^{12})}.\]

\paragraph{Lifting the solution:} We will now use $A$ to create a highly satisfying solution $B$ to $G = G_r(\cD)$. This proof is similar to the lifting proof from Lemma~\ref{lem:exp-bound-mu}, but we give the full proof here for completeness. For every $S_4 \in \cT$, every vertex $u \in \supp(\cD^{S_4})$ and every restriction $s \in \supp(\cD^{S_i}|X_{S_4} = u), i \in [3]$, let $g_s(v)$ denote the permutation $\pi(v,s)A(s)$. We will choose a randomized assignment as follows: for every set $S_4 \in \cT$, every vertex $u \in \supp(\cD^{S_4})$, choose a random $s \sim \cD^{S_1}|(X_{S_4} = u)$ and assign $B(u)= g_{s}(u)$. We will now upper bound the expected fraction of edges that $B$ violates for $\Phi$. 

Consider an edge $(u,v) \in G$ between two parts $S_4,S_5 \in \cT$ (with $S_4 \cap S_4=\emptyset$). This edge is satisfied if there exists $s' \in \supp(\cD^{S_1}|(X_{S_4} = u, X_{S_5} = v))$ such that, (1) $B(u) = g_{s'}(u)$, (2) $B(v) = g_{s'}(v)$ and (3) The triangle $(s',u,v)$ is consistent. 

To evaluate the probability there is such $s'$, we 
sample $s'\sim \cD^{S_1}|(X_{S_4} = u, X_{S_5} = v)$ 
and consider each one of the events, starting with event (1).
For it, the probability it doesn't hold is at most
\[\E_B\E_{(u,v) \sim \cD^{S_4 \cup S_5}}\E_{s' \sim \cD^{S_1}|(X_{S_4} = u, X_{S_5} = v)}[\Ind(B(u) \neq g_{s'}(u))]
=\E_{u \sim \cD^{S_4}}\E_{s,s' \sim \cD^{S_1}|u}[\Ind(g_{s'}(u) \neq g_s(u))]. 
\]
To calculate this, let us first calculate a bound on:
\[\E_{u \sim \cD^{S_4}}\E_{\substack{(s,s') \sim \cD^{S_1 \cup S_2}|u}}[\Ind(g_{s}(u) \neq g_{s'}(u))].\]
It is easy to check that an if an edge $(s,s') \in T(S_1,S_2,S_3;\cD|u)$ is satisfied by $A$, and the triangle $(s,s',u)$ is consistent then $g_s(u) = g_{s'}(u)$. Thus,
\begin{align*}
&\E_{u \in \cD^{S_4}}\E_{\substack{(s,s') \sim \cD^{S_1 \cup S_2}|u}}[\Ind(g_{s'}(u) \neq g_s(u))]\\
&\leq \E_{u \in G}\E_{\substack{(s,s') \sim \cD^{S_1 \cup S_2}|u}}[\Ind((s,s') \text{ not satisfied by $A$})]+\E_{u \in \cD^{S_4}}\E_{\substack{(s,s') \sim \cD^{S_1 \cup S_2}|u}}[\Ind((s,s',u) \in \incons(\Phi))]\\
&\lll \viol(A) + \E_{\substack{(s,s',u) \sim \cD^{S_1 \cup S_2 \cup S_4}}}[\Ind((s,s'u) \in \incons(\Phi))] :=\delta'(S_4).
\end{align*}

We are ready to bound the probability that event (1) does not happen. Towards this end, for each vertex $u \in \supp(\cD^{S_4})$ define $p_u := \Pr_{\substack{(s,s') \sim \cD^{S_1 \cup S_2}|u}}[g_{s'}(u) \neq g_s(u)]$, so that the above inequality translates to $\E_{u \sim \cD^{S_4}}[p_u] = \delta'(S_4)$. Let $p'_u = \Pr_{\substack{s,s' \sim \cD^{S_1}|u}}[g_{s'}(u) \neq g_s(u)]$. By Lemma~\ref{claim:gll-sing-val}, the second largest singular value of the bipartite graph $A(S_1,S_2;\cD|u)$ is at most $\eps\cdot \poly(r) \leq 0.01$ for all $u$, so by the easy direction of Cheeger's inequality we get that, $p'_u \leq O(p_u)$, which gives us that,
\[\E_{u \sim \cD^{S_4}}\E_{s,s' \sim \cD^{S_1}|u}[\Ind(g_{s'}(u) \neq g_s(u))] \lll \delta'(S_4),\]
thus bounding the probability of event (1). One can check that the probability of event (2) is the same as event (1), hence let us proceed to event (3). For that we get,
\[
\E_{(u,v) \sim \cD^{S_4 \cup S_5}}\E_{s' \sim \cD^{S_1}|(X_{S_1} = u, X_{S_2} = v)}[\Ind((s',u,v) \in \incons(\Phi))] = \E_{(s',u,v) \sim \cD^{S_1 \cup S_4 \cup S_5}}[\Ind((s',u,v) \in \incons(\Phi))]. 
\]
Adding up the probabilities that events (1),(2) or (3) do not happen and taking an expectation over $S_4,S_5 \sim \cT$ with $S_4 \cap S_5 = \emptyset$ we get that,
\[\E_B\E_{\substack{S_4,S_5 \sim \cT:\\S_4 \cap S_5}}\E_{(u,v)\sim \cD^{S_4 \cup S_5}}\E_{s' \sim \cD^{S_1}|(u,v)}[\Ind[\text{events (1),(2), or (3) don't hold}]] \lll 2^{O(r^{0.99}\log r)}\delta+2^{-\Omega(r^{12})}:=\delta',\]
where we used \eqref{eq:incons2} and \eqref{eq:incons3}.

We see that sampling sets $S_4,S_5 \sim \cT$ and an edge $(u,v) \sim \cD^{S_4 \cup S_5}$, the restrictions $s_u,s_v$ chosen by the randomized assignment $B$ and $s'$ fails to satisfy at least one of the events (1), (2) and (3) with probability $\lll \delta'$. Thus, 
with probability at most $O(\delta')$ over the choice of $S_4,S_5$, $(u,v)$, $s_u, s_v$, there is no $s'$ that satisfies all events and otherwise we get that 
$s'$ satisfies all of (1), (2) and (3).
This shows that in expectation the assignment $B$ that we get violates $\lll \delta'$-fraction of the edges of $G$ that are between disjoint $S_4,S_5 \in \cT$. Noting that this is a $1-o(1)$-fraction of all the edges of $G$ (since $r \ll d$) completes the proof.
\end{proof}

\section{Construction of Sparse UG Coboundary Expander}
We show that the complex from~\cite{ChapmanL} is a sufficiently strong UG coboundary expander. As an immediate corollary of~\cite{BafnaMinzer} we get that the Chapman-Lubotzky complex admits direct product testers.

\subsection{Vanishing Cohomology for $G_1[X]$ over $\S_m$}\label{sec:vc}
The main goal of this section 
is to present the Chapman-Lubotzky
complex. In particular, we need the statement that for all $m\in\mathbb{N}$, 
for an appropriate choice of parameters, 
this complex has vanishing $1$-cohomology. 
This statement was communicated to us by Dikstein, 
Dinur and Lubotzky~\cite{DDLpersonal}, and below we
give an alternative proof.

\begin{definition}
We say a graph $G$ has vanishing cohomology with respect to $\S_m$ if the following holds. Let $\Phi$ be a Unique-Games instance on $G$ over $\S_m$ in which every triangle in $G$ is consistent. Then there exists a solution $S \in \S_m^{V(G)}$ that satisfies all the constraints of $\Phi$. A complex $X$ has vanishing $1$-cohomology over $\S_m$ if $G_1[X]$ has vanishing cohomology over $\S_m$. 
\end{definition}
We note that the above definition is equivalent to the standard topological definition of vanishing 
$1$-cohomology. Therein, given a function $f\colon E\to H$ defined
on the edges where $H$ is some finite group, one 
defines the coboundary map 
$\partial f$ on triangles via
\[
\partial f(u,v,w) = f(u,v)f(v,w)f(w,u).
\]
With these notations, we care about functions $f$
such that $\partial f\equiv \text{id}$. If $f(u,v) = g(u)g(v)^{-1}$ for all edges $(u,v)\in E$ for 
some $g\colon V\to H$, then $\partial f \equiv \text{id}$
clearly. The first cohomology group of $G$ with coefficients in $H$ is defined via
\[
H^{1}(G, H) = 
\{f~|~\partial f \equiv \text{id}\}
\backslash
\{f~|~\exists g\colon V\to H, f(u,v) = g(u)g(v)^{-1}\},
\]
and note that the fact that $G$ has vanishing cohomology 
over $S_m$ with respect to the definition above 
is equivalent to $H^{1}(G,H) = 0$. We will 
use this language henceforth in this section.

First, we need the following general lemma 
that relates the cohomology of 
a complex to the homomorphisms 
of groups.
\begin{lemma}\label{lem:cohomology_hom}
    Let $G$ be a group acting transitively on a contractible topological space $X$, let $\Gamma$ be a discrete subgroup of $G$ that acts simply on $X$, and let $H$ be a finite group. Then $H^1(\Gamma\backslash X, H) \cong \mathrm{Hom}(\Gamma, H)$. 
\end{lemma}
\begin{proof}
    We may identify functions on $\Gamma \backslash X$ with functions on $X$ that are invariant under $\Gamma$. For a cycle $x$ let us write $[x]$ for its equivalence class modulo the boundaries. Given $[f] \in H^1(\Gamma \backslash X, H)$ we obtain that $\partial(f) = \text{id}$. Since $X$ is contractible, we may find a function $g$ on $X$ with $\partial(g) = f.$ Without loss of generality $g(x) = \text{id}$ for some $x\in X$. We now set $\varphi(\gamma) = g(\gamma (x))$ for each 
    $\gamma \in \Gamma$. We assert that $\varphi$ is a homomorphism. Indeed, for each $\gamma$ in $\Gamma$ choose a path from $\gamma x$ to $x$.  By hypothesis, $f = \partial(g)$ is invariant under left-multiplication by $\Gamma$. Therefore, the value of $\partial(g)$ on the path is $\varphi(\gamma) = g (\gamma x)g(x)^{-1} = g(\gamma \tau x)g(\tau (x))^{-1} = \varphi(\gamma \tau) \varphi(\tau)^{-1}.$ This yields that $\varphi$ is indeed a homomorphism. To show that the homomorphism is well defined we need to show that if $g$ is $\gamma$-invariant, then $\varphi=1$, which of  course holds. To construct the homomorphism in the reverse direction we choose a representative $x$ for each coset $\Gamma x$ and define a function $g$ on $X$ by setting its value on $\gamma x$ to be $\varphi(\gamma).$ We then obtain back an element $[\partial(g)]\in H^1(\Gamma\backslash X,H).$ It is easy to verify that the maps that we defined are inverses of each other.        
\end{proof}

Recall that a quarternion algebra over $\mathbb{F}$ is a field extension $\mathbb{F}[i,j,k]$ with $i^2= a, j^2 = b, k^2 =c$ and $k=ij = -ji$ for $a,b,c\in \mathbb{F}$. Let $\mathbb{Q}_p$ be the $p$-adic rationals and $\mathbb{Q}_{\infty}$ be $\mathbb{R}$. Let $\nu$ be either a prime $p$ or infinity. A quarternion algebra $D$ over $\mathbb{Q}$ is said to be split at $\nu$ if $D\otimes \mathbb{Q}_p$ is isomorphic to the algebra of $2\times2$ matrices over $\mathbb{Q}_{\nu}$. Otherwise, it is a division ring and it is said to be \emph{ramified} or unsplit.
Given a quarternion algebra $D$ there exists a natural involution $\tau$ sending each of $i,j,k$ to its negation. For a matrix $A$ with coordinates in $D$ let us write $A^* = (\tau(a_{ji})).$ Then the group $SU(n,D)$ consists of all the matrices with coordinates in $D$, such that $A^*A = I$. The first result from number theory used by Chapman and Lubotzky is the following.

\begin{fact}
For every $p_0$ there exits a quarternion algebra that is ramified (non-split) only over $p_0$ and $\infty.$
\end{fact}

The Chapman-Lubotzky high dimensional expanders consist of $\Gamma \backslash B$, where $\Gamma$ is a lattice and $B$ is the Bruhat-Tits building of type $\tilde{C}_n$ over $\mathbb{Q}_p$. The lattice $\Gamma$ is given by taking a lattice inside the set $SU(n,D)$, where $D$ is a quarternion algebra splitting over $p_0,\infty$ for some $p_0,\infty$. In our case we choose the prime $p_0$ to be larger than $m$. The lattice is obtained as follows: we first embed $SU(n,D)$ diagonally inside the product $\prod_{\nu \ne p_0,\infty} SU(n,D\otimes Q_\nu)$. We then intersect it with a product of the compact groups $K_{\nu}$, where each $K_{\nu}$ can be an arbitrary compact open subgroup, which we choose to be $SP_{2n}(\mathbb{Z}_\nu)$ for each $\nu \ne p_0$, and when $\nu =p_0$, $K_{\nu}$ can be taken to be a pro $p_0$-group, which means that all its finite quotients have order which is a power of $p_0$:
\begin{fact}
    The subgroup $SU(n,D)\otimes Q_p$ contains a compact open pro-$p$ group for each $p$.
\end{fact}
Following Chapman-Lubotzky we set $K =\prod_{\nu \ne p,\infty}K_{\nu}$ and $\Gamma = K\cap SU(n,D)$. 

Finally, using the congruence subgroup property and the strong approximation theorem, Chapman--Lubotzky showed that $K$ is the profinite completion of $\Gamma$, which means that every homomorphism from $\gamma$ to a finite group can be extended uniquely into $K$. The final number theoretic facts that we need is the following. Recall that a discrete subgroup $L \le G$ is said to be a \emph{lattice} if there exists a $G$-invariant probability measure on $L\backslash G$. 

\begin{fact}$\Gamma$ is a lattice inside $SP_{2n}(\mathbb{Q}_p)$.
\end{fact}

We are now ready to state the main
lemma of this section.
\begin{lemma}\label{lem:CL-vc}
For all $m\in\mathbb{N}$, choosing 
$n$ and $p$ to be sufficiently 
large, the Chapman-Lubotzky $\Gamma \backslash B$ has vanishing cohomology. Namely, $H^{1}(\Gamma/B, S_m) = 0$.
\end{lemma}
The rest of this section is devoted 
to the proof of Lemma~\ref{lem:CL-vc}.
First, we note that Lemma~\ref{lem:cohomology_hom}, 
to show that $H^{1}(\Gamma\backslash B, S_m) = 0$ it suffices to show that $\mathrm{Hom}(\Gamma ,S_m) = 0$. 
By the above, it is sufficient to establish that $\mathrm{Hom}(K_{\nu} ,S_m) = 0$ for each $\nu\ne p,\infty$. 

When $\nu = p_0$, this follows from the fact that $K_{\nu}$ is a pro $p_0$-group and $p_0>m$. This implies that its quotients are $p$-groups and therefore cannot be embedded inside $S_m$ when $m<p$.

For $\nu \ne p_0$ this follows from the fact that the minimal dimension of a sub-representation of $\mathrm{SP}_{2n}(\mathbb{Z}_p)$ goes to infinity with $n$, and so if $n$ is sufficiently large $\mathrm{Hom}(\mathrm{SP}_{2n}(\mathbb{Z}_p), S_m) = 0$ as each such permutation representation would give rise to a complex representation of dimension $m$, which would therefore have to be trivial. The lower bound on the dimension can be easily deduced by adapting the method due to Howe and Gurevich of $U$-rank, and we give
the details below.
\footnote{A stronger form of the lemma below will appear in a future paper of Evra, Kindler, Lifshitz, and Pirani that will extend the theory of $U$-rank over $p$-adic groups.}

\begin{lemma}
The minimal dimension of a non-trivial representation of $SP_{2n}(\mathbb{Z}_p)$ tends to infinity as a function of $n$ uniformly over all $p$.
\end{lemma}
\begin{proof}
Let $\chi = \mathrm{tr
}\circ \rho $ be the character of a nontrivial representation. Let $\mathcal{B}$ be the Abelian group of symmetric matrices over $\mathbb{Z}_p$ and let $G = \mathrm{GL}_n(\mathbb{Z_p}),$ which acts on the group $\mathcal{B}$ via $A^B = B^t A B.$ The group $\mathrm{SP}_{2n}(\mathbb{Z_p})$. Recall that the group $SP_{2n}(\mathbb{Z_p})$ conists of the matrices that preserves a symplectic form. They are generated by the matrices of the form 
$\begin{pmatrix} I & A \\
O & I \end{pmatrix}$, where $A$ is symmetric matrix with entries in $\mathbb{Z}_p$, those matrices of the form 
$\begin{pmatrix} C& O\\ 
O & \left(C^{t}\right)^{-1}
\end{pmatrix}$, where $C$ is in 
$\mathrm{GL}_n(\mathbb{Z}_p)$ together with the element 
$w = \begin{pmatrix} O& I\\
-I& O\end{pmatrix}$. 

Let us denote the first group of elements by $B$ and the second group of elements $G.$ Then the product $BG$
is a semi direct product $B\rtimes G$ and it known as the \emph{Siegel parabolic} subgroup. 

Now the normal subgroup that $B$ generates can easily be seen to be all of $\mathrm{SP}_{2n}(\mathbb{Z}_p).$ This shows that the restriction of $\rho$ to $B$ is also nontrivial. Indeed,  if the restriction of $\rho$ to $B$ is trivial this would imply that $\rho
$ on all the conjugates of $B$ and therefore also on the normal subgroup that it generates. 

Therefore, the restriction of $\chi$ to $B$ is a a non-constant function. Let us denote the restriction by  $f$.  Then since $f$ is invariant under the conjugation action of 
$G$ on $B$ we have $f(B^t A B) = f(A)$ for all $B\in G$. This implies that for each character $\chi_X$ in the Pontryagin dual of $B$ all its equivalence classes $\chi_{BXB^{t}}$ also lies in the support of $f$. Now the dimension of $\chi$ is equal to $f(1)$ and standard representation theory implies that the Fourier coefficients of $f$ are nonnegative integers. 

We may therefore lower bound the dimension $\chi(1)$ by the minimal size of an orbit of a nontrivial character in the Pontryagin dual of $B$. Each character in the 
Pontryagin dual corresponds to a symmetric matrix $X$ over $\mathbb{Z}/p^k$ and $\chi_X(A) = \omega_p^{\mathrm{tr}(
XA)}$. Now the orbit of $X$ under the action of $G$ can be lower bounded by the orbit of $P^i X$ for each $i$. We may therefore multiply $X$ by powers of $p$ until all its entries are multiples of $p^{k-1}$ and lower bound the orbit of the resulting character. However, the resulting matrix has entries in $\mathbb{F_p}$  and the lower bound on the orbit of $X$ was established for that case by Gurevich and Howe~\cite{gurevich2017small}.
\end{proof}

\subsection{Cosystolic Expansion for $G_r[X]$}
In this section we show that $G_r[X]$ has sufficiently strong cosystolic expansion. To prove this we use a local to global theorem of~\cite{DiksteinD23} that shows that when the links are coboundary expanders then the complex is a cosystolic expander. Let us first show that the vertex links of the Chapman-Lubotzky complex are spherical buildings of type C, as well as tensor products of type A and type C buildings.

\begin{lemma}\label{lem:vertex-links-CL}
The Chapman-Lubotzky complex $X$ discussed in Section~\ref{sec:vc} has vertex links that are spherical buildings of either type $C_{n-1}$, tensor of type $A_1$ and type $C_{n-2}$,
or tensor of type $C_k$ and type $C_{n-k-1}$. Furthermore, $X$ is an $O(1/\sqrt{p})$-one-sided local spectral expander.
\end{lemma}

\begin{proof}
Consider the complex $X$, and note that 
by construction the links of $X$ are the same as the links of $\Gamma$. The affine building $\Gamma$ has a Coxeter diagram of type $\tilde{C}_n$. By~\cite[Proposition 3.16]{AB}, the links
of $\Gamma$ correspond to spherical buildings 
whose Coxeter diagram is the result of 
deleting one vertex from the diagram 
$\tilde{C}_n$ (see also~\cite[Lemma 3.1.13]{Essert}). By inspection, it follows that 
all links are spherical buildings that are 
either type $C_{n-1}$, tensor of type $A_1$ and type $C_{n-2}$,
or tensor of type $C_k$ and type $C_{n-k-1}$. These complexes correspond to subspaces over the field $\F_p$ (see~\cite{nelken2023geometric} for a description of the $\tilde{C}_n$ building).

Using Lemmas~\ref{lem:eps-product-C} and~\ref{lem:eps-product-A} we know that the spherical buildings of type C and type A are 
$O(1/\sqrt{p})$-product distributions or equivalently $O(1/\sqrt{p})$-one-sided local spectral expanders. This implies that the same also holds for their tensor products. Using the trickling-down theorem, Theorem~\ref{thm:trickling-down} we then get that $X$ is also an $O(1/\sqrt{p})$-one-sided local spectral expander.
\end{proof}

Below we state the formal definition of cosystolic expansion for graphs and complexes that will be useful for us. 
\begin{definition}
We say a graph $G$ is a $C$-cosystolic expander with respect to $\S_m$ if the following holds. Let $\Phi$ be a Unique-Games instance on $G$ over $\S_m$ with $\incons(\Phi) =\delta$. Then
one can change the constraints of $\Phi$ on $C\delta$-fraction of the edges to get $\Phi'$ such that $\incons(\Phi')=0$. We say that a complex $X$ is a $C$-cosystolic expander over $\S_m$ if the graph $G_1[X]$ is a $C$-cosystolic expander over $\S_m$. 
\end{definition}
In literature the above definition is known as the $1$-cosystolic expansion of $X$ over $\S_m$, and has various generalizations to higher levels of $X$ (when $\S_m$ is replaced by an Abelian group) but we refrain from stating those definitions.

Let us now state the local to global theorem. For a complex $X$ let $Y = X^{\leq R}$ denote the ``cutoff'' of $X$, i.e. $Y = (X(1),\ldots, X(R))$, equipped with the same distributions.
\begin{theorem}[\cite{DiksteinD23} Theorem 1.2 modified]\label{thm:dd-local-to-global}
There exists a constant $R$ such that for all $\beta,\lambda >0$ and $m,d, C \in \N$ with $\beta,\lambda \leq \poly(1/C)$ the following holds. Let $X$ be a $d$-dimensional complex such  that $X^{\leq R}$ is a $\lambda$-one-sided local spectral expander and for all $v \in X(1)$, $G_1[X_v]$ is a $(C,\beta)$-coboundary expander over $\S_m$. Then $X$ is a $\poly(C)$-cosystolic expander over $\S_m$.   
\end{theorem}

\begin{proof}
First, \cite[Theorem 1.2]{DiksteinD23} asserts that if the one-skeletons of all links of $X$ are $\lambda$-one-sided expanders and
for all $v \in X(1)$ the complex $X_v$ is a $C$-coboundary expander over $\S_m$, then $X$ is a $\poly(C)$ cosystolic expander over $\S_m$. 

We discuss how to change their argument to get the stronger theorem. Firstly their argument only uses the link expansion to derive spectral gaps of the up-down walks on $X$ on $R' < R$ levels. Therefore we can look at $Y = X^{\leq R}$ instead and using the local spectral expansion of $Y$, derive that these up-down walks on $Y$ are sufficiently expanding. But these are the same as the walks on $X$ therefore hence the latter also have the required expansion properties. Hence we only require the local spectral expansion assumption on $X^{\leq R}$.

Now let us discuss the requirements on the coboundary expansion of links -- we only have that the 1-links are $(C,\beta)$-coboundary expanders over $\S_m$, instead of $C$-coboundary expanders. Their argument works as is, except in the step that they apply coboundary expansion. They use the fact that given any UG instance: $\incons(\Phi) \geq \frac{1}{C}\viol(\Phi)$, but we can instead get that, $\incons(\Phi) \geq \frac{1}{C}\viol(\Phi) - \frac{\beta}{C}$. This error of $\beta/C$ gets absorbed into the other additive errors that depend on ``$\eta$'' (check the proof overview) and $\lambda$ as long as $\beta \leq \poly(1/C)$. Hence this error does not affect the rest of the argument and we get the same conclusion.
\end{proof}

\begin{lemma}\label{lem:CL-cosystolic}
For all $r \ll d\ll q$ the following holds. Let $X$ be the $d$-dimensional Chapman-Lubotzky complex over $\F_q$. Then $G_r[X]$ is a $2^{O(r^{0.99}\log r)}$-cosystolic expander over $\S_m$ for all $m \in \N$.  
\end{lemma}

\begin{proof}
Consider the complex $X^r$, with 
\[
X^r(1) = \{v~|~v \in X(r)\},
\qquad 
X^r(2) = \{(u,v)~|~ u,v \in X(r) \text{ such that } u \cup v \in X(2r)\}
\]
and more generally
\[
X^r(\ell) = 
\{(u_1,\ldots,u_{\ell}): u_1,\ldots,u_{\ell} \in X(r) \text{ such that } u_1 \cup \ldots \cup u_{\ell} \in X(\ell r)\}.
\] 
Note that that $G_r[X] = G_1[X^r]$.
We intend to use Theorem~\ref{thm:dd-local-to-global}, and for that we 
verify that the constant-sized links have one-sided expansion on the 1-skeletons, and that the
$1$-links of $X^r$ are coboundary expanders over $\S_m$.

\paragraph{One-sided local spectral expansion of $(X^r)^{\leq R}$:} We will show that $(X^r)^{\leq R}$ is an $\exp(-r^{10})$-one-sided local spectral expander. Let us fix an $i$-face $F$ of $X^r$ for $i\leq R-2$ and upper bound the second eigenvalue of the 1-skeleton of $(X^r)_F$. Let $Y$ be the complex $X_F$ with $d':=\dim(Y)=d-ir$. Then bounding the second eigenvalue of the 1-skeleton corresponds to bounding the second eigenvalue of the random walk $W$ that picks a random $r$-face $A \in Y(r)$, goes up to a random $2r$-face $(A,B) \in Y(2r)$ and then outputs the $r$-face $B \in Y(r)$. We can check that $W$ is $1-O(r^2/d')$-close to the walk $W'$ that picks a random $d'$-face $I$ in $Y$ and two uniformly random $r$-faces inside $I$. This is because $W'$ conditioned on outputting two disjoint $r$-faces is the same as $W$. The probability that $W'$ outputs disjoint faces is at least $1-O(r^2/d')$, therefore giving us that the second eigenvalues of $W$ and $W'$ differ by at most $O(r^2/d')$. But $W'$ is the same as the up-down walk from $X(r) \rightarrow X(d') \rightarrow X(r)$, which by Lemma~\ref{lem:spectral_gap_of_graphs_from_HDX} has a second eigenvalue of at most $O(r/d')$ since $X_F$ is a $O(1/\sqrt{q})$-one-sided local spectral expander. This gives us that every $i$-link of $X^r$ has a second eigenvalue of at most $O(r/(d-ir))$ which is at most $2^{-r^{12}}$ by taking $d$ to be a large enough function of $r,R$, as required.

\paragraph{Coboundary expansion of 1-links of $X^r$:} We will show that for all links $I \in X^r(1)$, $G_1[(X^r)_I]$ are $(2^{O(r^{0.99}\log r)},\exp(-r^{10}))$-coboundary expanders. Fix some $I \in X^r(1)$. By definition of $X^r$, $I$ corresponds to an $r$-face in $X$, which we also denote by $I$. Let $v \in X(1)$ be some vertex that belongs to $I$. By Lemma~\ref{lem:vertex-links-CL} we know that $X_v$ is a spherical building either of type $C_{n-1}$, tensor of type $A_1$ and type $C_{n-2}$,
or tensor of type $C_k$ and type $C_{n-k-1}$. Let $Y = X_v$, associated with the distribution $\mu$ over its maximal faces. We get that the complex $X_I$ is some $r-1$-link of $Y$ denoted by $Y_{S \rightarrow A_0}$ and associated with the distribution $\mu|Y_S = A_0$. One can check that the complex $(X^r)_I$ is then equal to $(X_I)^r$ which in turn equals $(Y_{S \rightarrow A_0})^r$. So we get that, $G_1[(X^r)_I] = G_1[(Y_{S \rightarrow A_0})^r] = G_r(\mu|X_S = A_0)$ which is a $(2^{O(r^{0.99}\log r)},2^{-\Omega(r^{12})})$-coboundary expander by Lemma~\ref{lem:ug-cobdry-spherical}. So we get that all the vertex links of $X^r$ are coboundary expanders.

In the two paragraphs above, we have shown that $X^r$ satisfies the hypotheses of Lemma~\ref{thm:dd-local-to-global}. Therefore applying the lemma on $X^r$ we get that $X^r$ or equivalently the graph $G_1[X^r]=G_r[X]$ is a $2^{O(r^{0.99}\log r)}$-cosystolic expander over $\S_m$.
\end{proof}

\subsection{Chapman-Lubotzky complex is a UG Coboundary Expander}
\begin{theorem}\label{lem:CL-coboundary}
For all $m,r \ll d \ll q \in \N$ the following holds. Let $X$ be the $d$-dimensional Chapman-Lubotzky complex over $\F_q$. Then $X$ is an $(m,r,\alpha(r),\alpha(r))$-UG coboundary expander for 
\[
\alpha(r) = 2^{O(r^{0.99}\log r)}.
\]
\end{theorem}

\begin{proof}
From Lemma~\ref{lem:CL-vc} we know that $X$ has vanishing 1-cohomology over $\S_m$. We also know that $X$ is a well-connected clique complex~\cite{ChapmanL}.
We now use~\cite[Lemma 3.7]{DiksteinD-agreement}, which asserts that if a complex $X$ is a well-connected clique complex with $G_1[X]$ having vanishing cohomology then $G_r[X]$ also has vanishing cohomology. We also know from Lemma~\ref{lem:CL-cosystolic} that $G_r[X]$ is an $2^{O(r^{0.99}\log r)}$-cosystolic expander over $\S_m$. Combining the two facts proves that $G_r[X]$ is an $2^{O(r^{0.99}\log r)}$-coboundary expander over $\S_m$, therefore an $(m,r,\alpha(r),\alpha(r))$-UG coboundary expander.
\end{proof}

\subsection{Proof of Theorem~\ref{thm:main}}
Fix $\eps,\delta>0$ and take 
$r$ and $m$ sufficiently large.
Take the Chapman Lubotzky 
complex $X$ for sufficiently large 
$n$ and prime $p$, which is a one-sided local spectral expander~\cite{ChapmanL}. Combining this with  Theorem~\ref{lem:CL-coboundary} we get that  $X$ satisfies the conditions
of Theorem~\ref{thm:BM}, and therefore
we conclude that the canonical direct product tester of $X$ has soundness at most $\delta$.
\qed

\section{Acknowledgements}
We would like to thank Alex Lubotzky for several 
helpful discussions. In particular, Alex suggested to us the use of the Chapman-Lubotzky complex and sent us a preliminary
version of~\cite{ChapmanL}. We thank 
Yotam Dikstein, Irit Dinur and Alex Lubotzky for
telling us that the Chapman-Lubotzky complex has vanishing $1$-cohomology over $\S_m$. We would also like to thank Shai Evra, Michael Chapman and Tali Kaufman for helpful conversations and the Simons Institute for the Theory of Computing, where part of this work took place when the authors were participating in the ``Analysis and TCS: New Frontiers'' program. 
\bibliographystyle{alpha}
\bibliography{references}

\end{document}